\documentclass{amsart}

\usepackage{amsmath,amssymb,amsthm,mathrsfs,hyperref,tikz, tikz-cd,enumerate}
\newtheorem{theorem}{Theorem}
\newtheorem{lemma}{Lemma}
\newtheorem{proposition}{Proposition}
\newtheorem{definition}{Definition}
\newtheorem{corollary}{Corollary}

\newcommand{\im}{\operatorname{im}}

\newcommand{\op}{\operatorname}

\title{Sachs equations and plane waves IV: projective differential geometry}
\author{Jonathan Holland}
\address{Compunetix\\
  2420 Mosside Blvd \# 1\\
  Monroeville, PA 15146}
\author{George Sparling}
\address{University of Pittsburgh\\
  Department of Mathematics\\
  301 Thackeray Hall\\ Pittsburgh, PA 15260
}
\date{\today}

\begin{document}
\begin{abstract}
  This article gives an invariant representation of the curvature of a plane wave spacetime in terms of the Schwarzian of a curve in the Lagrangian Grassmannian.  It develops a general theory of cross ratios and Schwarzians of curves in what it terms the {\em middle Grassmannian}.  Most of the theory is developed in infinite dimensions, where the middle Grassmannian is defined via cocycles on the space of complemented subspaces of a topological vector space.  In the case of Hilbert spaces, the middle Grassmannian coincides with the set of closed subspaces whose Hilbert dimension and codimension are the same cardinal.  We show how to define the cross ratio of four mutually complementary subspaces, and give some applications to algebraic geometry.  We then study differential calculus in the middle Grassmannian of a Banach space over a complete normed field, defining the tangent vector and covector to a curve, followed by the Schwarz invariant.  We give a characterization of the vanishing of the Schwarz invariant in terms of hyperbolic structures.  Then, we specialize to the case of (real) Hilbert spaces, and define a symplectic structure, which allows us to consider the Grassmannian of Lagrangian subspaces and the {\em Lagrange--Schwarzian} of positive curves in the Lagrangian Grassmannian, which we show represents the curvature of a plane wave.
\end{abstract}
\maketitle

\tableofcontents
\section{Introduction}

This article is of a piece with \cite{SEPI}, \cite{SEPII}, and \cite{SEPIII}, which constitute a systematic study of plane wave spacetimes.  In \cite{SEPII}, the authors introduced a microcosm as a conformal plane wave that is homogeneous with respect to its conformal automorphism group.  In \cite{SEPIII}, we then showed how microcosms correspond in a natural way to orbits of one parameter groups in the symplectic group.  In a planned later article in this series, we shall discuss Fourier analysis of fields on plane waves.  Throughout the series, we have developed a sort of {\em dictionary} between projective geometry and plane waves, given in the rough Table \ref{Dictionary}.  An aim of this paper is to describe the entries of this table more completely, and to elucidate the connections between them.

\begin{table}
  \begin{tabular}{|c|c|c|}
    \hline
  Symplectic spaces & Plane waves & ref\\
    \hline
    Lagrangian subspace & Rosen spacetime & \cite{SEPI} \\
    Hamiltonian curve & Metric tensor & \cite{SEPIII} \\
    Canonical transformation & Isomorphism of Rosen universes & \cite{SEPII} \\
    Lagrange--Schwarzian & Spacetime curvature & \cite{SEPI} \\
    Schwarzian derivative & Matter content & \cite{SEPI} \\
    Hooke--Newton equation & Jacobi equation & \cite{SEPI} \\
    Hamilton--Jacobi equation & Sachs equation & \cite{SEPI}-\cite{SEPIII} \\
    Group of canonical transformations & Microcosm& \cite{SEPIII}\\
    \hline
\end{tabular}
\caption{Dictionary of symplectic spaces to plane waves.  The last column is labeled with a citation to the most relevant articles in this series where the connection is already discussed at some length.  }\label{Dictionary}
\end{table}

This article begins by defining and recalling some of the basic constructions from linear symplectic geometry in finite dimensions, and also serves to orient our general perspective on plane waves.  A general reference for the projective differential geometry is the last chapter of \cite{ovsienko2004projective}, but we have adapted it in such a way that it is natural for the study of plane waves.  Some things are more natural when viewed from the plane wave point of view, and others more natural from the point of view of Lagrangian Grassmannians in symplectic spaces.

Next, in \S\ref{CrossRatioSection}, we study cross ratios in (potentially) infinite-dimensional spaces.  Although we have primarily Hilbert spaces in mind, much of this goes through in significantly more generality, without much essential modification, to topological vector spaces.  Here the main innovation is that of the ``middle Grassmannian'', which reduces to the Grassmannian of $n$-planes in $2n$ dimensions in the finite dimensional case, and gives a way of formulating the proposition that the planes should have complements that are the same dimension in a Hilbert space.


An amusing application to projective algebraic geometry is given in \S\ref{ApplicationProjective}, where we show (Theorem \ref{TransversalsTheorem}) that any set of four $n$-planes in general position in a projective space of dimension $2n+1$ has exactly $n+1$ traversals. This is presumably a well-known fact in algebraic geometry, but the proof is an interesting illustration of our methodology.  

Next, in \S\ref{HyperbolicSection}, we formulate a key idea that shall be used throughout, namely that of a {\em hyperbolic structure} on a topological vector space, which is by definition a representation of the algebra of split quaternions.  Hyperbolic structures define a canonical family of curves, which we call {\em affine geodesics}, in the middle Grassmannian.

Because we need the open mapping theorem to obtain useful decompositions, we confine attention to Banach spaces beginning in \S\ref{ReflectionsBanach}.  In these sections, we examine the projective geometry of certain Grassmannians attached to a topological vector space, using reflection operators.  Then \S\ref{DifferentialCalculus} introduces the differential calculus of curves in the middle Grassmannian.  In particular, the {\em tangent covector} (\S\ref{TangentCovector}) is a natural construction.

Next, we show how to formulate the Schwarzian in our language.  First, \S\ref{OneDSchwarzian} shows how this is done in one dimension.  Then, \S\ref{SchwarzianGeneral} gives the Schwarzian for a curve in the middle Grassmannian.  The main theorems are Theorem \ref{VanishingSchwarzian} and Theorem \ref{VanishingSchwarzianPrime}, which characterize the vanishing of the Schwarzian.

Then in \S\ref{SymplecticLagrangian}, we formulate an infinite-dimensional symplectic space modeled on a Hilbert space.  To a first approximation, a symplectic space is just a real Hilbert space with a complex structure, but in which the inner product is ``forgotten'' so that only the symplectic form induced by the complex structure is ``remembered''.  We show how to make this definition natural, and develop some properties.  A good general referene on the Lagrangian Grassmannian in Hilbert spaces is \cite{furutani2004fredholm}, and there is some overlap with the results proven in this section, where short proofs have been given for completeness.

Hyperbolic structures in symplectic spaces have more structure than just vector spaces, and they can be used to block-decompose symplectic spaces in a nice way.  The details are described in \S\ref{CL11Symplectic}.  The affine geodesics associated with hyperbolic structures on a symplectic space are studied next.  Of special importance are the hyperbolic structures that are positive-definite, since these define a compatible Hilbert space structure on the symplectic space.  A {\em positive curve} in a Lagrangian Grassmannian is a curve whose tangent vector is a positive definite Hermitian form.  Positive curves give rise naturally to positive hyperbolic structures at each point, which in turn determines an {\em osculating geodesic} in \S\ref{PositivitySection}.

Finally in \ref{LiftSection}, we show how every positive curve in the Lagrangian Grassmannian admits a unique lift to the symplectic group relative to any positive hyperbolic structure.  This is then used to construct the Schwarzian derivative, in a gauge given by a positive hyperbolic structure, which we relate to the normal tidal curvature of a plane wave, and prove gauge invariance.


\section{Symplectic vector spaces and plane waves}\label{SymplecticPlaneWaves}
We refer the reader to \cite{SEPI} for a general account of plane waves.  We here consider only the Rosen and Brinkmann forms.
\begin{definition}
  Let $\mathbb U$ denote a real interval and $\mathbb X$ a real Euclidean space of dimension $n$, and put $\mathbb M=\mathbb U\times\mathbb R\times\mathbb X$. Then:
  \begin{itemize}
  \item The one parameter group of diffeomorphisms $\mathcal D_t:\mathbb M\to\mathbb M$ defined by $\mathcal D_t(u,v,x)=(u,e^{2t}v,e^tx)$ is called the {\em standard dilation};
  \item An {\em automorphism} of the pair $(\mathbb M,\mathcal D_t)$ is a diffeomorphism on $\mathbb M$ onto $\mathbb M$ that commutes with the dilation;
  \item For $p(u)$ a symmetric endomorphism of $\mathbb X$ depending smoothly on $u\in\mathbb U$, the metric tensor on $\mathbb M$ given by
    $$G_\beta(p) = 2\,du\,dv + x^Tp(u)x\,du^2 - dx^Tdx $$
    is called the {\em Brinkmann metric} associated to $p(u)$; and
  \item For $h(u)$ a positive-definite symmetric endomorphism of $\mathbb X$, the metric tensor on $\mathbb M$ given by
    $$G_\rho(h) = 2\,du\,dv + x^Th(u)x\,du^2 - dx^Tdx $$
    is called the {\em Rosen metric} associated to $h(u)$.
  \end{itemize}
\end{definition}
This section gives an account of the {\em Brinkmannization map}, that maps a Rosen spacetime to a corresponding Brinkmann spacetime.  The Brinkmannization was introduced as such in \cite{SEPII}.  We shall here interpret the Brinkmannization in terms of symplectic vector spaces and Lagrangian matrices.  Accordingly, this section begins with an introduction to symplectic vector spaces and the various objects associated with them.

\subsection{Symplectic vector spaces}
A symplectic vector space is a pair $(\mathbb V,\omega)$, where $\mathbb V$ is a real vector space, of dimension $2n$, and $\omega$ is a constant skew form on $\mathbb V$ of rank $2n$.  The basic example of a symplectic vector space is that attached to a real Euclidean space.  A real Euclidean space is an $n$-dimensional vector space $\mathbb X$, equipped with a transpose operation ${}^T:\mathbb X\to\mathbb X^*$.  Without any essential loss in generality, we may take $\mathbb X=\mathbb R^n$, which we shall always think of as being {\em column} vectors, and the transpose to be the usual vector transpose.  As usual in Hilbert space, ${}^T$ is used to conjugate operators (matrix transpose) as well.
\begin{definition}
  Let $\mathbb X$ be an Euclidean space.  The {\em standard symplectic form} on the space $\mathbb X\oplus\mathbb X$ is the form
  $$\omega_{\mathbb X}(x_1\oplus x_2,y_1\oplus y_2) = x_1^Ty_2-x_2^Ty_1.$$
\end{definition}

\begin{definition}
Given two symplectic vector spaces $(\mathbb V,\omega), (\bar{\mathbb V},\bar\omega)$, a {\em (linear) canonical transformation} is a linear map $\phi:\mathbb V\to\bar{\mathbb V}$ such that $\phi^*\bar\omega=\omega$.
\end{definition}
We shall only consider linear canonical transformations unless otherwise noted.

A standard lemma is
\begin{lemma}
  Given a symplectic vector space $(\mathbb V,\omega)$ of dimension $2n$, there exists a canonical transformation onto a standard symplectic space $(\mathbb X\oplus\mathbb X,\omega_{\mathbb X})$.
\end{lemma}
\begin{proof}
  Any skew form is orthogonally diagonalizable into $2\times 2$ skew blocks, so the problem is reduced to the case of $n=1$, where it is obvious.
\end{proof}

\begin{definition}
The {\em symplectic group} $\op{Sp}(\mathbb V,\omega)$ of a symplectic vector space $(\mathbb V,\omega)$ is the group of (linear) canonical transformations of $\mathbb V$ to itself.
\end{definition}
For example, let us determine the condition for a block matrix $\begin{bmatrix}A&B\\C&D\end{bmatrix}$ to belong to $\op{Sp}(\mathbb X\oplus\mathbb X,\omega)$ for the standard symplectic form $\omega$ (in the presence of a basis of $\mathbb X$, this is canonically isomorphic to the Lie group $\op{Sp}(2n,\mathbb R)$)
\begin{align*}
  \omega &= \begin{bmatrix}0&I_{\mathbb X}\\ -I_{\mathbb X}&0\end{bmatrix}\\
         &=\begin{bmatrix}A&B\\C&D\end{bmatrix}^T\omega\begin{bmatrix}A&B\\C&D\end{bmatrix}\\
         &=\begin{bmatrix}A^TC-C^TA&A^TD-C^TB\\ B^TC-D^TA & B^TC-D^TB\end{bmatrix}
\end{align*}
\[A^TC=C^TA,\quad B^TC=C^TB,\quad A^TD-C^TB = I_{\mathbb X}.\]

A useful metaphor is that the standard symplectic space specifies the kinematics of a free system.  We want to study (``time''-dependent) dynamics, and the following gives a model of that:
\begin{definition}
  Let $(\mathbb V,\omega)$ be a symplectic space.  A {\em dynamics} is a smooth one-parameter family $\phi(u)$ of (linear) canonical transformations of $(\mathbb V,\omega)$ onto a standard symplectic space $\mathbb X\oplus\mathbb X$, depending smoothly on $u\in\mathbb U$, a real interval.
\end{definition}
Fixing a particular canonical transformation of $\mathbb V$ to $\mathbb X\oplus\mathbb X$, a dynamics is thus a curve in the symplectic group.

\subsubsection{Key example of a dynamics}
The dynamics under consideration henceforth are determined by a given symmetric matrix $p(u)$, depending smoothly on $u$, that we call the {\em potential}:
\begin{definition}
  Let $p:\mathbb U\to\op{End}(X)$ be a smooth, symmetric endomorphism of $\mathbb X$.  Let $\mathbb J(p)$ denote the space of smooth functions $J:\mathbb U\to\mathbb X$ such that
  $$\ddot J + pJ = 0.$$
\end{definition}
Then, for each $u\in\mathbb U$, put $\phi(u)J = J(u)\oplus\dot J(u)$.  For each $u$, let $\omega_u$ be the pullback of the standard symplectic form to a symplectic form on $\mathbb J(p)$.  Then
\begin{lemma}
  The symplectic form $\omega_u$ on $\mathbb J(p)$, given by
  $$\omega_u(J,K) = J(u)^T\dot K(u) - \dot J^T(u)K(u)$$
  does not depend on $u$.
\end{lemma}
\begin{proof}
  This follows by differentiating and using the fact that $p$ is symmetric.
\end{proof}
Therefore, equipping $\mathbb J(p)$ with the form $\omega$ (which is independent of $u$), the mapping $\phi(u)$ is a canonical transformation of $(\mathbb J(p),\omega)$ to the standard symplectic space.  So $(\mathbb J(p),\omega,\phi)$ is a dynamics.

\subsection{Lagrangian matrices}
\begin{definition}
  A {\em Lagrangian matrix} is a smooth function $L:\mathbb U\to\op{End}(\mathbb X)$ such that $L^T\dot L = \dot L^TL$ and the rank of the homomorphism $L\oplus\dot L:\mathbb X\to\mathbb X\oplus\mathbb X$ is full.
\end{definition}
In particular, a {\em nonsingular Lagrangian matrix} is a Lagrangian matrix taking values in $\op{GL}(\mathbb X)$.  With $\mathbb X=\mathbb R^n$, the rank condition of a Lagrangian matrix means that the $2n\times n$ block matrix $\begin{bmatrix}L\\\dot L\end{bmatrix}$ has rank $n$. 

A Lagrangian matrix $L:\mathbb U\to\op{End}(\mathbb X)$ determines a mapping $L\oplus\dot L:\mathbb U\to \op{Hom}(\mathbb X,\mathbb X\oplus\mathbb X)$, whose image is contained in a Lagrangian subspace.  That is, a Lagrangian matrix determines a curve in the Lagrangian Grassmannian attached to the space $\mathbb X\oplus\mathbb X$ with its standard symplectic form.


Given a symmetric matrix $p(u)$, depending smoothly on $u$, we have shown that $(\mathbb J(p),\omega)$ is a symplectic vector space.
\begin{definition}
  Given a Lagrangian matrix $L$, if there exists a smooth symmetric matrix $p(u)$ such that the {\em Hooke--Newton equation} $\ddot L+pL=0$ holds, then $p$ is called the {\em potential} for the Lagrangian matrix $L$.
\end{definition}
The potential $p$ (if it exists) is uniquely determined by $L$, and a Lagrangian matrix $L$ has potential $p$ if and only if every column of $L$ belongs to $\mathbb J(p)$.  A nonsingular Lagrangian matrix always has a potential.

\begin{definition}
Given a smooth symmetric potential $p:\mathbb U\to\op{End}(\mathbb X)$, the set of Lagrangian matrices with potential $p$ on $\mathbb U$ is denoted by $\mathcal{LM}(p)$.
\end{definition}

There is a natural right action of $\op{GL}(\mathbb X)$ on the right of $\mathcal{LM}(p)$, namely $L(u)\to L(u)g$ for $g\in\op{GL}(\mathbb X)$.  The next sections concern the orbits of this action.

\subsection{Lagrangian subspaces}
\begin{definition}
  Let $(\mathbb V,\omega)$ be a symplectic space.  A subspace $\mathbb L\subset\mathbb V$ is called {\em isotropic} if $\omega_{\mathbb L\times\mathbb L}=0$.  An isotropic space that is maximal with respect to subspace inclusion is called {\em Lagrangian}.
\end{definition}
As an example, the subspaces $\mathbb X\oplus 0$ and $0\oplus\mathbb X$ in the standard symplectic space $\mathbb X\oplus\mathbb X$ are Lagrangian.

Another example is as follows. As above, for a given potential $p$, let $\mathbb J(p)$ denote the space of solutions to $\ddot J+pJ=0$ on a domain $\mathbb U$.
\begin{lemma}
  For each $u\in\mathbb U$, let $\mathbb H_p(u)\subset\mathbb J(p)$ be the subspace of $J\in\mathbb J(p)$ such that $J(u)=0$.  Then $\mathbb H_p(u)$ is Lagrangian.
\end{lemma}
\begin{proof}
  The subspace is isotropic because $J(u)\oplus \dot J(u) = 0 \oplus \dot J(u)$ is clearly isotropic at $u$, and the symplectic form does not depend on $u$.  We have to show that $\mathbb H_p(u)$ is maximal.  Because $J$ satisfies the second order equation $\ddot J+pJ=0$, it is uniquely determined by the initial conditions $J(u),\dot J(u)$ (at fixed $u$).  Any solution with $J(u)=0$ is therefore uniquely determined by a given $\dot J(u)\in\mathbb X$.  Therefore $\mathbb H_p(u)$ is linearly isomorphic to $\mathbb X$, and so is Lagrangian.
\end{proof}

\begin{lemma}
  To each Lagrangian matrix $L$ for the potential $p$, the column space of $L$ is a Lagrangian subspace of $\mathbb J(p)$.
\end{lemma}
\begin{proof}
  For fixed $u$, the columns of $\phi(u)L$ are linearly independent, and their span is isotropic, both because $L$ is Lagrangian.
\end{proof}

\subsection{The Lagrangian Grassmannian}
\begin{definition}
  The {\em Lagrangian Grassmannian} $\mathcal{LG}(\mathbb V,\omega)$ of a symplectic space $(\mathbb V,\omega)$ is the set of Lagrangian subspaces of $\mathbb V$.   
\end{definition}
While this section concerns the real case predominately, we can consider also the {\em complexified} Lagrangian Grassmannian $\mathcal{LG}_{\mathbb C}(\mathbb V,\omega)$ of maximal complex-linear subspaces of $\mathbb V\otimes\mathbb C$ on which $\omega$ is isotropic (regarded as a complex-bilinear form).

\begin{definition}
Let $\mathbb A\subset\mathbb V$ be a Lagrangian subspace of the (fixed) symplectic space $(\mathbb V,\omega)$.  Let $\mathcal U(\mathbb A)\subset\mathcal{LG}(\mathbb V,\omega)$ be the set of all Lagrangian subpsaces $\mathbb Y$ that are complementary from $\mathbb A$.  The topology generated by the subbase $\mathcal U(\mathbb A)$ is the {\em affine topology} on $\mathcal{LG}(\mathbb V,\omega)$.
\end{definition}
Fix $\mathbb A,\mathbb B\subset\mathbb V$ that are mutually complementary.  Then to any element $\mathbb Y$ of $\mathcal U(\mathbb A)$, there exists a (unique) symmetric operator $Y:\mathbb B\to\mathbb A$ such that $\mathbb Y$ is the graph of $Y$ over $\mathbb B$:
$$\mathbb Y = \{x + Yx|x\in\mathbb B\}.$$
Here an operator is symmetric if $\omega(b,Yb')+\omega(Yb,b')=0$ for all $b,b'\in\mathbb B$.  The term ``symmetric'' comes from comparison with the case of a standard symplectic space $\mathbb X\oplus\mathbb X$, with the Lagrangian subspaces $\mathbb X\oplus 0$ and $0\oplus\mathbb X$.  Then  a mapping $Y:\mathbb X\oplus 0\to 0\oplus\mathbb X$ is symmetric if and only if it is symmetric as an operator from $\mathbb X$ to itself, because $\omega(x,Yx')+\omega(Yx,x') = x^TYx' - (x')^TYx = x^T(Y-Y^T)x'$.

We prove this later.  For the present purpose $\mathbb Y\mapsto Y$ defines an {\em affine chart} on $\mathcal U(\mathbb A)$ (depending on the auxiliary Lagrangian subspace $\mathbb B$).  It is clear that the affine topology is coarser than the Zariski topology.

The affine charts $\mathcal U(\mathbb A)$ give $\mathcal{LG}(\mathbb V,\omega)$ the structure of a real-analytic manifold of dimension $\frac12 n(n+1)$, which is the real part of a complex projective variety $\mathcal{LG}_{\mathbb C}(\mathbb V,\omega)$.  Therefore $\mathcal{LG}(\mathbb V,\omega)$ is compact.  It is also connected (for its topology as a smooth manifold).

\begin{definition}
The {\em tangent space} to $\mathcal{LG}(\mathbb V,\omega)$ at a point $\mathbb H$ is the space of symmetric linear functions $\mathbb H\to\mathbb T/\mathbb H$.
\end{definition}

The following is obvious:
\begin{lemma}
The Lagrangian Grassmannian of $\mathbb J(p)$ is quotient space $\mathcal{LM}(p)/\op{GL}(\mathbb X)$ of the space of Lagrangian matrices $L$ with potential $p$, by the right action of $\op{GL}(\mathbb X)$: $L(u)\sim L(u)g$, $g\in\op{GL}(\mathbb X)$.
\end{lemma}

We can therefore think of Lagrangian subspaces in terms of Lagrangian matrices.
\begin{definition}
A {\em nonsingular curve} in $\mathcal{LG}(\mathbb V,\omega)$ is a smooth function $\mathbb H:\mathbb U\to\mathcal{LG}(\mathbb V,\omega)$ whose tangent vector at each $u\in\mathbb U$ is a linear isomorphism $\dot{\mathbb H}(u):\mathbb H(u)\to\mathbb V/\mathbb H(u)$.  A nonsingular curve is called {\em positive} if the tangent vector is a positive definite form, in the sense that $\omega(\dot H(u)h,h) > 0$ for all $h\in\mathbb H(u), h\not=0$.
\end{definition}

\begin{lemma}\label{HamiltonianLemma}
  For $t\in\mathbb U$, let $\mathbb H_p(t)\subset\mathbb J(p)$ be the Lagrangian subspace represented by the matrix $L$ with potential $p$ such that $L(t)=0,\dot L(t)=I$, $\ddot L+pL=0$.  Then $\mathbb H_p:\mathbb U\to\mathcal{LG}(\mathbb J(p))$ is a nonsingular curve in $\mathcal{LG}(\mathbb J(p))$.
\end{lemma}
\begin{proof}
  For $u,t\in\mathbb U$, denote by $L(u,t)$ the solution to $L(t,t)=0$ and $L_u(t,t)=I$, $L_{uu}(u,t)+p(u)L(u,t)=0$.  Then, by smoothness in initial conditions for linear second-order equations, $L$ is smooth.

  Now $L_t(t,t)$ is invertible for all $t$. Indeed, suppose that $L_t(t,t)X=0$.  Then $L_t(u,t)X$ is a solution to $L_t(t,t)X=L_{ut}(t,t)X=0$, $L_{uut}(u,t)X + p(u)L_t(u,t)X=0$, which therefore vanishes identically for all $u$, hence $L_t(u,t)X=0$.  Therefore $L(u,t)X$ is constant in $t$.  But $L(t,t)X=0$ (because $L(t,t)=0$), so $L(u,t)X=0$ for all $u$.  Since $L$ is Lagrangian, this implies that $X=0$.
\end{proof}

\begin{definition}
The curve $\mathbb H_p$ from Lemma \ref{HamiltonianLemma} is called the {\em Hamiltonian curve} of the potential $p$.
\end{definition}

We call it the Hamiltonian curve because it defines the Hamiltonian for the Schr\"odinger equation, which we shall discuss in a later article.

\subsection{Plane waves}
We now specialize to the case of a plane wave.  Recall that that $\mathbb M=\mathbb U\times\mathbb R\times\mathbb X$, where $\mathbb U$ is a real interval and $\mathbb X$ is a Euclidean space.
\begin{definition}
  The {\em Brinkmannization} corresponding to an invertible Lagrangian matrix $L$ is the automorphism of $\mathbb M$
\begin{equation}\label{BrinkmannizationEq}
  \beta_L(u,v,x) = \left(u,\quad v + 2^{-1}x^TL^T\dot Lx,\quad Lx\right).
\end{equation}
\end{definition}
In \cite{SEPI} and \cite{SEPII}, we show that there every Rosen metric has a Brinkmannization which is a Brinkmann metric:
\begin{lemma}
  Let $L$ be a Lagrangian matrix such that $L^TL=h$.  Then $G_\rho(h)=\beta_L^*G_\beta(p)$ where $p$ is determined by $\ddot L + pL=0$.  Conversely, if $p$ is given and $L$ is a Lagrangian solution to $\ddot L + pL=0$ then $G_\rho(L^TL)=\beta_L^*G_\beta(p)$ on any open subset of $\mathbb U$ where $L$ is nonsingular.
\end{lemma}

\subsubsection{Determining the Hamiltonian curve}
We now determine the Hamiltonian curve of a plane wave, given a Brinkmannization $L$ having potential $p$.  Let $H:\mathbb U\to\op{End}(\mathbb X)$ be a symmetric endomorphism, smooth in $u$, such that $L^T\dot HL=I$.
\begin{lemma}
For each $u_0\in\mathbb U$, the matrix $L_{u_0}(u)=L(u)(H(u)-H(u_0))$ is Lagrangian with potential $p$, and represents the Hamiltonian curve $\mathbb H_p(u_0)$.
\end{lemma}
\begin{proof}
  We have $L_{u_0}(u_0)=0$, so it is sufficient to show that $L_{u_0}$ is Lagrangian with potential $p$.  Assume without loss of generality that $u_0=0$ and $H(0)=0$.  Then
  \begin{align*}
    \dot L_0 &= \dot LH + L^{-T}\\
    L_0^T\dot L_0 &= H^TL^T\dot LH + H
  \end{align*}
  which is symmetric, because $L^T\dot L$ is symmetric ($L$ is Lagrangian) and $H$ is symmetric.

  We must show that $L_0$ has potential $p$.  We have
  \begin{align*}
    \ddot L_0 &= \ddot L H + 2\dot L\dot H + L\ddot H\\
              &= -pL_0 + 2\dot LL^{-1}L^{-T} - L(L^{-1}\dot LL^{-1}L^{-T} + L^{-1}L^{-T}\dot L^TL^{-T})\\
              &= -pL_0
  \end{align*}
  because $L^{-T}\dot L^T$ is symmetric.
\end{proof}

\subsection{Affine geodesics}
In \S\ref{AffineGeodesicSection} below, we introduce affine geodesics in general.  Here we discuss the standard example of an affine geodesic in the Lagrangian Grassmannian $\mathcal{LG}(\mathbb X\oplus\mathbb X)$ of a standard Hermitian symplectic space.  A standard affine geodesic is a curve corresponding to a potential of zero.  Thus the typical affine geodesic has Lagrangian matrix
$$L(u) = (u-t)\mathcal I_{\mathbb X}$$
where $t$ is a base point.

The Rosen plane wave corresponding to this affine geodesic is thus 
$$2\,du\,dv - (u-t)^2dx^Tdx,$$
where is singular at the point $u=t$, and flat on its domain.  Also, $\dot H = (u-t)^{-2}$, so $H = - (u-t)^{-1} + H_0$.  The associated solution of the Sachs equation $\dot S+S^2=0$ is $S=(u-t)^{-1}$.

We now construct a {\em tangent} affine geodesic to a plane wave.  Let $L$ be a Lagrangian matrix for the Brinkmannization of a Rosen plane wave $G_\rho(L^TL)$.  Fix a base point $u_0\in\mathbb U$.  Let $L_{u_0}$ be the Lagrangian matrix with zero potential which shares the same initial conditions at $u=u_0$ as $L$:
$$L_{u_0}(u_0)=L(u_0) =: L_0,\quad \dot L_{u_0}(u_0) = \dot L(u_0) =: \dot L_0,\quad \ddot L_{u_0} = 0.$$
Thus
\begin{align*}
  L_{u_0}(u) &= L_0 + (u-u_0)\dot L_0 = L_0 + (u-u_0)S_0 L_0\\
  h_{u_0}(u) = L_{u_0}(u)^TL_{u_0}(u) &= L_0^TL_0 + 2(u-u_0)L_0^TS_0 L_0 + (u-u_0)^2 L_0^TS_0^2L_0
\end{align*}
where we have put $S_0=\dot L_0L_0^{-1}$.  Note that $h_{u_0}(u_0)=h(u_0)$ and $\dot h_{u_0}(u_0)=\dot h(u_0)$, so the affine geodesic makes first order contact.  (In our later perspective the integral $H$ of $h^{-1}$ is more fundamental, and so affine geodesics make second order contact, and thus are described later as ``osculating'', but it amounts to the same thing.)

\subsection{Osculating microcosms}
The osculating geodesic makes first order contact with a Rosen.  Recall the definition of microcosm:
\begin{definition}
  Let $\mathbb M=\mathbb R\times\mathbb R\times\mathbb X$.  A {\em microcosm} on $\mathbb M$ is a metric of the form
  $$G_\alpha = du\,\alpha - dx^Tdx,\quad \alpha=2\,dv - 2x^T\omega dx + x^Tpx\,du^2$$
  where $\omega$ and $p$ are constant $n\times n$ matrices with $\omega$ skew and $p$ symmetric.  A microcosm is {\em twist-free} if $\omega=0$.
\end{definition}
\begin{lemma}
  Every plane wave has a unique twist-free microcosm making second order contact at a point.
\end{lemma}
\begin{proof}
  As in the construction of the osculating affine geodesic, let $L_{u_0}$ be such that
  $$L_{u_0}(u_0) = L(u_0) =: L_0,\quad \dot L_{u_0}(u_0) = \dot L(u_0) =: \dot L_0,\quad \ddot L_{u_0}(u) + p(u_0)L_{u_0}(u) = 0.$$
  We thus have, with $p_0=p(u_0)$,
  $$L_{u_0}(u) = L_0 + (u-u_0)\dot L_0 - \frac12(u-u_0)^2p_0L_0 + O((u-u_0)^3).$$
  We now have $\ddot h(u_0) = \ddot h_{u_0}(u_0)$, $\dot h(u_0) = \dot h_{u_0}(u_0)$, and $h(u_0)=h_{u_0}(u_0)$.  Thus the metrics $G_\rho(L^TL)$ and $G_\rho(L_{u_0}^TL_{u_0})$ make second order contact at $u=u_0$.  (In particular, note that they have the same normal curvature.)
\end{proof}

It is easy to see that it is not generally possible to incorporate a non-zero twist $\omega$ to make higher order contact with a microcosm, basically because of degrees of freedom: the $n^2$ parameters $p_{u_0}$, $\omega_{u_0}$ cannot be fit to the $n(n+1)$ parameters $p_0,\dot p_0$.  The tidal curvature $P$ of a microcosm satisfies
$$\dot P + [\omega,P] = 0$$
and so this gives the obstruction to third order contact.  

\section{Reflections and cross ratios}\label{CrossRatioSection}
Let $\mathbb T$ be a topological vector space over a topological field $\mathfrak K$ of characteristic different from $2$.   A {\em morphism} of toplogical vector spaces over $\mathfrak K$ is a continuous $\mathfrak K$-linear map which is open onto its image.  Thus in a class of topological vector spaces for which the open mapping theorem holds (e.g., Banach spaces), a morphism is the same thing as a continuous linear operator.  By an {\em isomorphism} of topological vector spaces over $\mathfrak K$, we mean an invertible $\mathfrak K$-linear homomorphism, which is continuous and has continuous inverse.

Recall that linear subspaces $\mathbb A,\mathbb B\subset \mathbb T$ are said to be {\em complementary} if there exists a projection morphism\footnote{By definition, this means that $[\mathbb A\mathbb B]_+$ is an endomorphism such that $[\mathbb A\mathbb B]_+^2=[\mathbb A\mathbb B]_+$.} $[\mathbb A\mathbb B]_+ : \mathbb T\to \mathbb T$ whose image is $\mathbb A$ and kernel is $\mathbb B$.      Complementary subspaces are thus automatically closed, and there exists a unique pair of commuting continuous linear operators
$$[\mathbb{AB}]_+,[\mathbb{AB}]_-:\mathbb T\to\mathbb T,$$
whose images are $\mathbb A$ and $\mathbb B$, respectively.  That is, we have
$$[\mathbb{AB}]_\pm^2=[\mathbb{AB}]_\pm $$
$$\op{im} [\mathbb{AB}]_+ = \mathbb A,\quad \ker [\mathbb{AB}]_+=\mathbb B$$
$$\op{im} [\mathbb{AB}]_- = \mathbb B,\quad \ker [\mathbb{AB}]_-=\mathbb A$$
$$\mathcal I = [\mathbb{AB}]_+ + [\mathbb{AB}]_-.$$
The associated {\em reflection operator} is denoted by $[\mathbb{AB}]=[\mathbb{AB}]_+-[\mathbb{AB}]_-$.  This is the unique operator whose $+1$ eigenspace is $\mathbb A$ and whose $-1$ eigenspace is $\mathbb B$.

Note that
\begin{itemize}
\item $[\mathbb{AB}]^2=\mathcal I$ is the identity operator
\item $[\mathbb{BA}]=-[\mathbb{AB}]$
\item $[\mathbb{AB}]=[\mathbb{AB}]_+-[\mathbb{AB}]_-$
\item $\mathcal I=[\mathbb{AB}]_++[\mathbb{AB}]_-$
\item $[\mathbb{AB}]_+^2=[\mathbb{AB}]_+$,\quad $[\mathbb{AB}]_-^2=[\mathbb{AB}]_-$
\end{itemize}

{\bf Example:}  If $\mathbb T$ is a Hilbert space and $\mathbb A\subset\mathbb T$ is a closed subspace, then $\mathbb A$ and $\mathbb A^\perp$ are complementary, and $[\mathbb A\mathbb A^\perp]_+$ is the orthogonal projection onto $\mathbb A$.

We now define an auxiliary operator that shall be of fundamental importance.

\begin{definition}
  Let $\mathbb A,\mathbb B,\mathbb C,\mathbb D\subset\mathbb T$ be subspaces such that $\mathbb A$ is complementary to $\mathbb D$, and $\mathbb B$ is complementary to $\mathbb C$.  Let
  \[\binom{\mathbb{AB}}{\mathbb{CD}} = \frac12([\mathbb{AD}]+[\mathbb{BC}]).\]
\end{definition}
Note that it follows immediately from the definition that
\begin{equation}\label{EasySymmetries}
  \binom{\mathbb{AB}}{\mathbb{CD}} = \binom{\mathbb{BA}}{\mathbb{DC}} =  - \binom{\mathbb{CD}}{\mathbb{AB}}=  - \binom{\mathbb{DC}}{\mathbb{BA}}.
\end{equation}

We shall say that a linear operator $T$ {\em exchanges subspaces $\mathbb A$ with $\mathbb B$} if $T(\mathbb A)\subset\mathbb B$ and $T(\mathbb B)\subset \mathbb A$ (i.e., not necessarily surjectively or injectively).
\begin{theorem}\label{SqrtCrossRatio}
Suppose that each of $(\mathbb A,\mathbb D)$ and $(\mathbb B,\mathbb C)$ is a pair of complementary subspaces.  The operator $\binom{\mathbb{AB}}{\mathbb{CD}}$ exchanges the subspaces $\mathbb A$ with $\mathbb B$ and $\mathbb C$ with $\mathbb D$:
  \begin{equation}\label{ExchangesSubspaces}
    \binom{\mathbb{AB}}{\mathbb{CD}}\,: \begin{tikzcd}
    \mathbb A\arrow[r,bend left] & \mathbb B\arrow[l, bend left]
  \end{tikzcd},
  \begin{tikzcd}
    \mathbb C \arrow[r, bend left]&\mathbb D \arrow[l, bend left]
  \end{tikzcd}
\end{equation}
and the following are equivalent:
\begin{enumerate}[(a)\qquad]
\item \(\displaystyle\binom{\mathbb{AB}}{\mathbb{CD}}\) is an isomorphism, and thus all arrows in \eqref{ExchangesSubspaces} are isomorphisms;
  \vspace{10pt}
\item $(\mathbb A,\mathbb C)$ and $(\mathbb B,\mathbb D)$ each is a complementary pair. 
\end{enumerate}
Moreover, when these equivalent conditions hold,
\begin{equation}\label{InverseBinom}
\displaystyle\mathcal I=\binom{\mathbb{AB}}{\mathbb{CD}}\binom{\mathbb{AB}}{\mathbb{DC}}=\binom{\mathbb{AB}}{\mathbb{CD}}\binom{\mathbb{BA}}{\mathbb{CD}}.
\end{equation}
\end{theorem}
The conclusion of the theorem can be illustrated by a picture, denoting complementarity by dashed lines
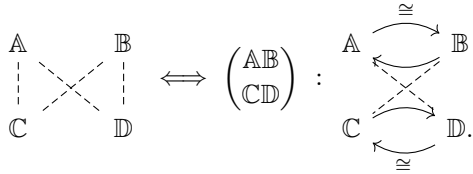
\[\begin{tikzcd}
    \mathbb A \arrow[d,-,dashed]\arrow[dr,-,dashed] & \mathbb B\arrow[dl,-,dashed]\arrow[d,-,dashed]\\
    \mathbb C & \mathbb D
  \end{tikzcd}\iff \binom{\mathbb{AB}}{\mathbb{CD}}\ :
  \begin{tikzcd}
    \mathbb A \arrow[r,bend left,"\cong"]\arrow[dr,-,dashed] & \mathbb B\arrow[dl,-,dashed]\arrow[l, bend left]\\
    \mathbb C \arrow[r,bend left] & \mathbb D\arrow[l,bend left,"\cong"].
  \end{tikzcd}
\]
\begin{proof}
  We shall show that $\binom{\mathbb{AB}}{\mathbb{CD}}$ exchanges $\mathbb A$ with $\mathbb B$, since it then also exchanges $\mathbb C$ with $\mathbb D$ by symmetry \eqref{EasySymmetries}.  Let $a\in\mathbb A$.  Since $\mathbb{B,C}$ are complementary, write $a=b+c$ for some unique $b\in\mathbb B$ and $c\in\mathbb C$.  Then we have
  \[\binom{\mathbb{AB}}{\mathbb{CD}}a = \tfrac12([\mathbb{AD}]+[\mathbb{BC}])a = \tfrac12( a + b-c) = b.\]
  Thus \(\binom{\mathbb{AB}}{\mathbb{CD}}\mathbb A\subset\mathbb B\).   Similarly $\binom{\mathbb{AB}}{\mathbb{CD}}\mathbb B\subset\mathbb A$.

  We now prove equivalence of (a) and (b).  To prove (b)$\implies$(a), it is clearly sufficient to show that (b) implies equation \eqref{InverseBinom}.

  Let $a\in\mathbb A$.  Then
  \begin{align*}
    \binom{\mathbb{AB}}{\mathbb{CD}}\binom{\mathbb{AB}}{\mathbb{DC}}a
    &= \frac12\binom{\mathbb{AB}}{\mathbb{CD}}(a + b-d),\quad a = b+d \\
    &= \binom{\mathbb{AB}}{\mathbb{CD}}b,\quad b = a-d \\
    &= \frac12(a+d+b)\\
    &= a.
  \end{align*}
  Likewise, we show that for $d\in\mathbb D$
  \[\binom{\mathbb{AB}}{\mathbb{CD}}\binom{\mathbb{AB}}{\mathbb{DC}}d = d,\]
  and so
  \[\binom{\mathbb{AB}}{\mathbb{CD}}\binom{\mathbb{AB}}{\mathbb{DC}} = \mathcal I\]
  because $\mathbb A$ and $\mathbb D$ are complementary.

  To show that (a)$\implies$(b), we shall show that
  \begin{equation}\label{AC+}
  [\mathbb{AC}]_+ = [\mathbb{AD}]_+\binom{\mathbb{AB}}{\mathbb{CD}}^{-1}
  \end{equation}
  and in so doing, it will also prove that the subspaces $(\mathbb A,\mathbb C)$ are complementary (by definition).  The proof that $(\mathbb B,\mathbb D)$ are complementary follows from an identity $[\mathbb{BD}]_+=[\mathbb{BC}]_+\binom{\mathbb{AB}}{\mathbb CD}^{-1}$, which is a consequence of \eqref{AC+} and symmetry.  Now, to prove \eqref{AC+} we resove it into blocks:
  \begin{align*}
    [\mathbb{AD}]_+\binom{\mathbb{AB}}{\mathbb{CD}}^{-1}[\mathbb{AD}]_+&=[\mathbb{AD}]_+\\
    [\mathbb{AD}]_+\binom{\mathbb{AB}}{\mathbb{CD}}^{-1}[\mathbb{BC}]_-&=0.
  \end{align*} 
  The second of these follows because $\binom{\mathbb{AB}}{\mathbb{CD}}$ exchanges the subspaces $\mathbb C$ and $\mathbb D$.
  For the first, note that
  \begin{align*}
    \binom{\mathbb{AB}}{\mathbb{CD}}
    & =[\mathbb{AD}]_+ - [\mathbb{BC}]_-\\
    \mathcal I
    &= \binom{\mathbb{AB}}{\mathbb{CD}}\binom{\mathbb{AB}}{\mathbb{CD}}^{-1} \\
    &= 2[\mathbb A\mathbb D]_+\binom{\mathbb{AB}}{\mathbb{CD}}^{-1} - 2[\mathbb B\mathbb C]_-\binom{\mathbb{AB}}{\mathbb{CD}}^{-1}.
  \end{align*}
  Then
  $$2[\mathbb A\mathbb D]_+\binom{\mathbb{AB}}{\mathbb{CD}}^{-1}[\mathbb A\mathbb D]_+ = [\mathbb A\mathbb D]_+ + 2[\mathbb B\mathbb C]_-\binom{\mathbb{AB}}{\mathbb{CD}}^{-1}[\mathbb A\mathbb D]_+.$$
  Now the second term is zero, because $\binom{\mathbb{AB}}{\mathbb{CD}}$ exchanges the spaces $\mathbb A$ and $\mathbb B$.
\end{proof}

\begin{corollary}
  Suppose that $\mathbb A_i$ and $\mathbb B_i$ are complementary for $i=1,2$.  Then the quadruple $(\mathbb A_1,\mathbb B_1,\mathbb A_2,\mathbb B_2)$ is pairwise complementary if and only if each of the operators $[\mathbb A_1\mathbb B_1]\pm[\mathbb A_2\mathbb B_2]$ is an isomorphism.
\end{corollary}

\begin{corollary}
  Let $(\mathbb A,\mathbb B,\mathbb C)$ be a triple of pairwise complementary closed subspaces of $\mathbb T$.  Then all three subspaces are isomorphic.
\end{corollary}

\begin{lemma}\label{GraphLemma}
  Let  $(\mathbb B,\mathbb A_1,\mathbb A_2)$ a complementary triple of subspaces of $\mathbb T$.  Then, for all $b\in\mathbb B$,
$$b = [\mathbb A_1\mathbb A_2]_+b +\frac12([\mathbb{BA}_1] + [\mathbb{BA}_2])[\mathbb A_1\mathbb A_2]_+b.$$
\end{lemma}
\begin{proof}
  Write $b=a_1+a_2$, $a_i\in\mathbb A_i$.  Then
  \begin{align*}
    \frac12([\mathbb{BA}_1] + [\mathbb{BA}_2])[\mathbb A_1\mathbb A_2]_+b
    &= \frac12([\mathbb{BA}_1] + [\mathbb{BA}_2])a_1\\
    &= -[\mathbb{BA}_2]_-a_1 = -[\mathbb{BA}_2]_-(b-a_2)\\
    &= a_2
  \end{align*}
  as required.
\end{proof}

\begin{corollary}\label{GraphTheorem}
  Let $(\mathbb B,\mathbb A_1,\mathbb A_2)$ be a complementary triple.  Then $\mathbb B$ is the graph of a morphism $B:\mathbb A_1\to\mathbb A_2$:
  $$\mathbb B = \{x + Bx | x\in\mathbb A_1\}.$$
  Moreover $B$ is the restriction to $\mathbb A_1$ of the reflection
  $$\binom{\mathbb{BB}}{\mathbb A_1\mathbb A_2}=\frac12([\mathbb{BA}_1] + [\mathbb{BA}_2])$$
  on $\mathbb T$ that is the identity on $\mathbb B$ and isomorphically interchanges $\mathbb A_1$ and $\mathbb A_2$.
\end{corollary}
\begin{proof}
  The first claim in the corollary follows from Lemma \ref{GraphLemma}, and the second from Theorem \ref{SqrtCrossRatio}.
\end{proof}

  \begin{definition}
    Suppose that each of $(\mathbb{A,D})$ and $(\mathbb{B,C})$ is a complementary pair.  The {\em cross ratio} of the four subspaces $\mathbb{A,B,C,D}$ is the endomorphism of $\mathbb T$:
    \begin{equation}\label{CrossRatioSubspaces}
      [\mathbb A,\mathbb B;\mathbb C,\mathbb D] = \binom{\mathbb{AB}}{\mathbb{CD}}^2 
    \end{equation}
  \end{definition}

  It follows by Theorem \ref{SqrtCrossRatio} that if the quadruple $(\mathbb{A,B,C,D})$ is pairwise complementary, then
  $$[\mathbb A,\mathbb B;\mathbb C,\mathbb D] = ([\mathbb{AC}]+[\mathbb{BD}])^{-1}([\mathbb{AD}]+[\mathbb{BC}]) = ([\mathbb{AD}]+[\mathbb{BC}])([\mathbb{AC}]+[\mathbb{BD}])^{-1}.$$
  Thus the nomenclature ``cross ratio'' is plausible, given the definition of the cross ratio of points of the projective line:
  \[[a,b;c,d] = \frac{(a-c)(b-d)}{(a-d)(b-c)}.\]
  In fact,
  \begin{proposition}
    Suppose that $\mathbb T$ is two-dimensional and $\mathbb{A,B,C,D}$ distinct points of $\mathbb{PT}$, given in an affine chart by
    $$\mathbb A=\op{span}\begin{bmatrix}1\\a\end{bmatrix},\quad\mathbb B=\op{span}\begin{bmatrix}1\\b\end{bmatrix},\quad\mathbb C=\op{span}\begin{bmatrix}1\\c\end{bmatrix},\quad\mathbb D=\op{span}\begin{bmatrix}1\\d\end{bmatrix}.$$
    Then
    $$([\mathbb{AC}]+[\mathbb{BD}])^{-1}([\mathbb{AD}]+[\mathbb{BC}])=[a,b;c,d]\mathcal I$$
    is a scalar multiple of the identity endomorphism of $\mathbb T$, where the scalar factor is the cross ratio defined by:
    $$[a,b;c,d]=\frac{(a-c)(b-d)}{(a-d)(b-c)}.$$
  \end{proposition}

  \begin{proof}
    Let $X=([\mathbb{AC}]+[\mathbb{BD}])^{-1}([\mathbb{AD}]+[\mathbb{BC}])$.  Theorem \ref{SqrtCrossRatio} implies that $X$ has at least four distinct invariant subspaces, and therefore is a multiple of the identity (since $\mathbb T$ is only two-dimensional).

    To compute the eigenvalue,
    $$([\mathbb{AD}]+[\mathbb{BC}])\begin{bmatrix}1\\a\end{bmatrix}=(\mathcal I + [\mathbb{BC}])\begin{bmatrix}1\\a\end{bmatrix}=2\frac{a-c}{b-c}\begin{bmatrix}1\\b\end{bmatrix}$$
    and
    $$([\mathbb{AC}]+[\mathbb{BD}])\begin{bmatrix}1\\b\end{bmatrix}=(\mathcal I + [\mathbb{BD}])\begin{bmatrix}1\\a\end{bmatrix}=2\frac{a-d}{b-d}\begin{bmatrix}1\\a\end{bmatrix}.$$
    Here, we have used the formulas
    $$[\mathbb{BC}]=\frac{1}{b-c}\begin{bmatrix}-b-c&2\\-2bc&b+c\end{bmatrix}$$
    and
    $$[\mathbb{AC}]=\frac{1}{a-c}\begin{bmatrix}-a-c&2\\-2ac&a+c\end{bmatrix},$$
    respectively.  So, in conclusion
    $$X\begin{bmatrix}1\\a\end{bmatrix}=\left(2\frac{a-d}{b-d}\right)^{-1}\left(2\frac{a-c}{b-c}\right)\begin{bmatrix}1\\a\end{bmatrix}$$
    which gives the claimed eigenvalue, of $[a,b;c,d]$.
  \end{proof}

  \begin{lemma}\label{ThreeSubspaces}
        If $\mathbb{A,B,C}$ are pairwise complementary, then, for every endomorphism $\mathcal A$ of $\mathbb A$, there exists a unique endomorphism $\mathcal{T_A}$ of $\mathbb T$ having $\mathbb{A,B,C}$ as invariant subspaces, such that $\mathcal{T_A}|_{\mathbb A}=\mathcal A$.  Moreover, $\mathcal{T_A}|_{\mathbb B}=\phi A\phi^{-1}$ and $\mathcal{T_A}|_{\mathbb C}=\psi \mathcal A\psi^{-1}$ where $\phi:\mathbb B\to\mathbb A$ and $\psi:\mathbb C\to\mathbb A$ are two (unique) isomorphisms.  In fact $\phi$ and $\psi$ are the restrictions to $\mathbb B$ and $\mathbb C$, respectively, of continuous linear reflections on $\mathbb T$.
  \end{lemma}

  \begin{proof}
    We wish to extend $\mathcal A$ off $\mathbb A$.  The extension is required to leave $\mathbb B$ invariant.  Define the maps $\phi,\psi$ via Theorem \ref{GraphTheorem} as
    \begin{align*}
      \phi &= \frac12([\mathbb{CA}] + [\mathbb{CB}])\\
      \psi &= \frac12([\mathbb{AB}] + [\mathbb{CB}])\\      
    \end{align*}
    Then put
    $$\mathcal{T_A} = \mathcal A[\mathbb{AB}]_+ + \phi\mathcal A\phi[\mathbb{AB}]_- $$
    Note that the range of $\phi[\mathbb{AB}]_-$ is in $\mathbb A$, so that the composite $\mathcal A\phi[\mathbb{AB}]_-$ is well-defined and continuous.  
    
    We now prove that $\mathbb C$ is an invariant subspace of $\mathcal{T_A}$.  If $c\in\mathbb C$, then $c=a+b$, $a=[\mathbb{AB}]_+c,b=[\mathbb{AB}]_-c$,  gives its decomposition with respect to $\mathbb A\oplus\mathbb B$.  Note that
    $$\phi b = \frac12\left([\mathbb{CA}]b-b\right) = \frac12\left([\mathbb{CA}](c-a)-b\right) = \frac12(c+a-b) = a.$$
    Thus
    \begin{equation}\label{TAc}
      \mathcal{T_A}c = \mathcal A a + \phi\mathcal A\phi b = \mathcal A a + \phi\mathcal A a.
    \end{equation}
    Now by Theorem \ref{GraphTheorem}, $\mathbb C$ is the graph of $\phi$ on $\mathbb A$, which is to say that the right hand side of \eqref{TAc} belongs to $\mathbb C$.

    We claim next that $\mathcal{T_A}|_{\mathbb C} = \psi\mathcal A\psi^{-1}$.  Write $c=a+b$ as above.  Then
    $$\psi c = \frac12\left([\mathbb{AB}]c+c\right) = \frac12(a-b+c) = a.$$
    Hence,
    $$\psi\mathcal A\psi c = \psi\mathcal Aa = \frac12\left(\mathcal Aa + [\mathbb{CB}]\mathcal Aa\right),$$
    and \eqref{TAc} gives
    $$\mathcal T_Ac = \mathcal A a + \phi\mathcal Aa = \mathcal A a + \frac12\left(-\mathcal Aa + [\mathbb{CB}]\mathcal Aa\right) $$
    and these are equal.

    Finally, we have to show uniqueness.  By subtracting two extensions of the same $\mathcal A$, we reduce to the case $\mathcal A=0$.  Then we have to show that if $\mathcal T$ is zero on $\mathbb A$, and $\mathbb B,\mathbb C$ are $\mathcal T$-invariant, then $\mathcal T$ is the zero operator.  Let $b\in\mathcal B$, and write $b=a+c$ for some $a\in\mathbb A$ and $c\in\mathbb C$.  Then $\mathcal Tb = \mathcal Ta + \mathcal Tc = \mathcal Tc$.  But $\mathcal Tb\in\mathbb B$ and $\mathcal Tc\in\mathbb C$, which are complementary subspaces, and therefore $\mathcal Tb=\mathcal Tc=0$.  So the restriction of $\mathcal T$ to $\mathbb B$ is zero and, since the restriction of $\mathcal T$ to the complementary subspace $\mathbb A$ is also zero, it follows that $\mathcal T=0$.
  \end{proof}

  \begin{theorem}
    If the quadruple of subspaces $(\mathbb A,\mathbb B,\mathbb C,\mathbb D)$ is pairwise complementary, then the restriction of the cross ratio $[\mathbb A,\mathbb B;\mathbb C,\mathbb D]$ to any of the four (invariant) subspaces $\mathbb A,\mathbb B,\mathbb C,\mathbb D$ is similar to the restriction to any other subspace, where the similarity is conjugation by a reflection operator that exchanges the two subspaces.  
  \end{theorem}
  \begin{proof}
    Since the cross ratio leaves all of the subspaces invariant, Lemma \ref{ThreeSubspaces} implies that the restriction of $[\mathbb A,\mathbb B;\mathbb C,\mathbb D]$ to any of the four subspaces has at most one extension to a linear operator that leaves them all invariant, and the extension given in the Lemma is of the sort described in the statement of the theorem.
  \end{proof}
  
  \begin{corollary}
    When $\mathbb T$ is finite dimensional, the cross ratio $[\mathbb A,\mathbb B;\mathbb C,\mathbb D]$ has the same characteristic polynomial when restricted to any of the four invariant subspaces $\mathbb A,\mathbb B,\mathbb C,\mathbb D$.  If $\lambda$ is an eigenvalue of $[\mathbb A,\mathbb B;\mathbb C,\mathbb D]$, then $[\mathbb A,\mathbb B;\mathbb C,\mathbb D]$ leaves the eigenspaces $\mathbb A_\lambda,\mathbb B_\lambda,\mathbb C_\lambda,\mathbb D_\lambda$ invariant, and the direct sum of any two of these is isomorphic to $\mathbb T_\lambda = \mathbb A_\lambda\oplus\mathbb B_\lambda$.  Finally, the cross ratio of operators on $\mathbb T_\lambda$ satisfies:
    $$[\mathbb A,\mathbb B;\mathbb C,\mathbb D]|_{\mathbb T_\lambda} = [\mathbb A_\lambda,\mathbb B_\lambda;\mathbb C_\lambda,\mathbb D_\lambda] = \lambda \mathcal I_{\mathbb T_\lambda}. $$
  \end{corollary}

\subsection{The anharmonic group}

\begin{theorem}
  Let $(\mathbb{A,B,C,D})$ be pairwise complementary subspaces of $\mathbb T$, and let
  \begin{align*}
    [\mathbb{A,B};\mathbb{C,D}]&=\Lambda.\qquad \text{Then:}\\
    [\mathbb{A,B};\mathbb{D,C}]&=\Lambda^{-1}\\
    [\mathbb{A,C};\mathbb{B,D}]&=I-\Lambda\\
    [\mathbb{A,C};\mathbb{D,B}]&=(I-\Lambda)^{-1}\\
    [\mathbb{A,D};\mathbb{B,C}]&=I-\Lambda^{-1}\\
    [\mathbb{A,D};\mathbb{C,B}]&=(I-\Lambda^{-1})^{-1}.
    \end{align*}
  \end{theorem}
  \begin{proof}

    We have already seen that the cross ratio is invertible.  By Theorem \ref{SqrtCrossRatio}, $[\mathbb{AD}] + [\mathbb{BC}]$ is invertible, with inverse $\frac14([\mathbb{AC}]+[\mathbb{BD}])$.  Therefore,
    $$[\mathbb{A,B};\mathbb{D,C}] = \frac14([\mathbb{AC}]+[\mathbb{BD}])^2 = 4([\mathbb{AD}]+[\mathbb{BC}])^{-2} = [\mathbb{A,B};\mathbb{C,D}]^{-1}$$
    which is the first identity.

    For the second identity, we have
    \begin{align*}
      I - \Lambda &= I - \frac14([\mathbb{AD}]+[\mathbb{BC}])^2\\
                  &= I - \frac12 I - \frac14\left([\mathbb{AD}][\mathbb{BC}] + [\mathbb{BC}][\mathbb{AD}]\right)\\
                  &= \frac12 I + \frac14\left([\mathbb{AD}][\mathbb{CB}] + [\mathbb{CB}][\mathbb{AD}]\right)\\
                  &= \frac14([\mathbb{AD}]+[\mathbb{CB}])^2 = [\mathbb{A,C};\mathbb{B,D}].
    \end{align*}

    The remaining identities now follow from the first two.

  \end{proof}

\section{Application to projective geometry}\label{ApplicationProjective}

The objective of this section is to establish the following:

\begin{theorem}\label{TransversalsTheorem}
  In a projective space $\mathbb P^{2n+1}$ of dimension $2n+1$ over an algebraically closed field of characteristic different from two, any set of four $n$-planes $\mathbb P^n$ in general position has exactly $(n+1)$ transversals.
\end{theorem}

By definition, a transversal is a line that intersects all four planes.

\begin{lemma}
  Let $\mathbb{A,B,C,D}$ denote the four subspaces.  Then any two-dimensional eigenspace of $[\mathbb{A,B;C,D}]$ is a transversal.  Conversely, any transversal is (a two-dimensional subspace of) an eigenspace of $[\mathbb{A,B;C,D}]$.
\end{lemma}

\begin{proof}
  The first direction is clear, since any eigenspace of $[\mathbb{A,B;C,D}]$ splits as a direct sum of two eigenspaces in each of the six pairings of the four subspaces:  $[\mathbb{A,B;C,D}]_\lambda = \mathbb A_\lambda\oplus\mathbb B_\lambda=\mathbb A_\lambda\oplus\mathbb C_\lambda$ (etc).

  Conversely, we must show that a transversal is an eigenspace.  If $\mathbb T_0$ is a transversal, then there are one-dimensional subspaces $\mathbb A_0,\mathbb B_0,\mathbb C_0,\mathbb D_0$, of $\mathbb{A,B,C,D}$ respectively, of which $\mathbb T_0$ is the direct sum of any two.

  Writing $\mathbb T_0=\mathbb A_0\oplus\mathbb B_0$, we have $[\mathbb{AB}]|_{\mathbb T_0} =  [\mathbb A_0\mathbb B_0]$, because if $t=a_0+b_0$ with $a_0\in\mathbb A_0$ and $b_0\in\mathbb B_0$, then $[\mathbb{AB}]t = [\mathbb{AB}](a_0+b_0) = a_0-b_0.$  So $\mathbb T_0$ is invariant under $[\mathbb{AB}]$, and by symmetry is invariant under all six reflection operators formed from the four subspaces.  Consequently, $\mathbb T_0$ is invariant under the cross ratio $[\mathbb{A,B;C,D}]$, and $[\mathbb{A,B;C,D}]_{\mathbb T_0} = [\mathbb A_0,\mathbb B_0;\mathbb C_0,\mathbb D_0]$
\end{proof}  
\begin{proof}[Proof of Theorem \ref{TransversalsTheorem}]

The lemma identifies the transversals with the two-dimensional eigenspaces of the operator $[\mathbb{A,B;C,D}]$.  This operator generically has $(n+1)$ eigenspaces over an algebraically closed field, all of which are two-dimensional.

\end{proof}

Theorem \ref{TransversalsTheorem} has a nice geometrical interpretation when $n=1$.  A quadric in $\mathbb P^3$ is ruled by two families of lines: $\alpha$ and $\beta$ lines.  Three lines $\mathbb A,\mathbb B,\mathbb C$ in general position in $\mathbb P^3$ lie on a unique quadric, as $\alpha$ lines.  A fourth line $\mathbb D$, generic with respect to the first three, intersects the quadric in a pair of points $X,Y$.  Through $X$ the unique $\beta$ line is a transversal to all four lines $\mathbb{A,B,C,D}$, and likewise through $Y$ the unique $\beta$ line is such a transversal.  This yields the two transversals.

Let $\mathbb V$ be a two-dimensional real vector space, and $\mathbb X\subset\op{Hom}(\mathbb V,\mathbb V)$ the three-dimensional space of all trace-free endomorphisms of $\mathbb V$.  The determinant $\det X$ defines a quadratic form for $X\in\mathbb X$, of indefinite signature.  The set of reflection operators $[\mathbb{AB}]$ for $\mathbb{A,B}\in\mathbb{PV}$ complementary one-dimensional subspaces, is the quadric
$$\mathcal Q = \{X\in\mathbb X\mid\det X = -1\},$$
an hyperboloid of one sheet.  The hyperboloid $\mathcal Q$ inherits a metric, of signature $(1,1)$, from the ambient (flat) metric on $\mathbb X$, and it is ruled by two families of null geodesics, which we now describe.

For $\mathbb A$ fixed, let
$$\alpha_{\mathbb A} = \{ [\mathbb{AB}]\in\mathcal Q | \mathbb B\in\mathbb{PV}, \mathbb B\not=\mathbb A\}.$$
Then $\alpha_{\mathbb A}$ is a line in $\mathcal Q$.  Likewise
$$\beta_{\mathbb A} = \{ [\mathbb{BA}]\in\mathcal Q | \mathbb B\in\mathbb{PV}, \mathbb B\not=\mathbb A\}$$
defines another line in $\mathcal Q$.  Observe that this sets up a pair of natural one-to-one correspondences, $\alpha,\beta$ from $\mathbb{PV}$ to one of each set of lines on $\mathbb{PV}$.

Each point of the quadric $\mathcal Q$ lies on two lines, an $\alpha$ line and $\beta$ line, which correspond to the $+1$ and $-1$ eigenspaces of the associated reflection operator, respectively.  This defines a map $(\alpha,\beta):\mathcal Q\to\mathbb{PV}\times\mathbb{PV}$, which is a birational equivalence.  (It can be made into an isomorphism by interpreting the intersection of parallel lines on $\mathcal Q$ as points at infinity.)

Choose affine parameters $x$ and $y$ on the $\alpha$ and $\beta$ copies of $\mathbb{PV}$.  Then the associated element of $\mathcal Q$ is:
$$Q(x,y) = I + 2(x-y)^{-1}\begin{bmatrix}1\\y\end{bmatrix}[-x\ \ 1] $$

\begin{lemma}
  The metric induced by the determinant in these coordinates is
  $$g = \det(dQ) = -4\frac{dx\,dy}{(x-y)^2} $$
\end{lemma}

\begin{proof}
  \begin{align*}
    dQ &= \frac2{(x-y)^2}\left(-(dx-dy)\begin{bmatrix}-x&1\\-xy&y\end{bmatrix} 
    +(x-y)\begin{bmatrix}-dx&0\\ -x\,dy-y\,dx&dy\end{bmatrix}\right)\\
       &= \frac2{(x-y)^2}\begin{bmatrix}y\,dx -x\,dy& dy-dx\\ y^2dx-x^2dy & x\,dy-y\,dx\end{bmatrix}\\
    \det dQ &= \frac{4}{(x-y)^4}\left(-(x\,dy-y\,dx)^2 - (dy-dx)(y^2dx-x^2dy)\right)\\
       &=-\frac{4}{(x-y)^2}dx\,dy
  \end{align*}
\end{proof}

Let $\mathbb P_1,\mathbb P_2,\mathbb Q_1,\mathbb Q_2$ be distinct points in $\mathbb{PV}$.  Corresponding to these points are the four lines $\alpha_{\mathbb P_1},\alpha_{\mathbb P_2},\beta_{\mathbb Q_1},\beta_{\mathbb Q_2}$.  These four lines bound a quadrilateral.  If the points $\mathbb P_1,\mathbb P_2,\mathbb Q_1,\mathbb Q_2$ are not interlaced on the circle, then the quadrilateral is null homotopic, and has a bounded interior.

\begin{theorem}
    Let $\mathbb P_1,\mathbb P_2,\mathbb Q_1,\mathbb Q_2$ be distinct, non-interlaced points in $\mathbb{PV}$.  Let $A$ denote the (hyperbolic) area of the quadrilateral in $\mathcal Q$ bounded by the quadrilateral $\alpha_{\mathbb P_1},\alpha_{\mathbb P_2},\beta_{\mathbb Q_1},\beta_{\mathbb Q_2}$.  Then
  $$e^{-A} = [\mathbb P_1,\mathbb P_2,\mathbb Q_1,\mathbb Q_2]^4 .$$
\end{theorem}
\begin{proof}
  Choose coordinates so that $\mathbb P_i=\op{span}\begin{bmatrix}1\\p_i\end{bmatrix}$ and $\mathbb Q_i=\op{span}\begin{bmatrix}1\\q_i\end{bmatrix}$, and $p_i>q_j$ for all $i,j=1,2$.  The area of the quadrilateral is
  \begin{align*}
    A &= \int_{q_1}^{q_2}\int_{p_1}^{p_2} \frac{-4\,dx\,dy}{(x-y)^2} \\
    &=-4(\log(p_1-q_1) + \log(p_2-q_2) - \log(p_2-q_1)-\log(p_1-q_2))
  \end{align*}
\end{proof}

\section{Hyperbolic structures}\label{HyperbolicSection}
\subsection{The Clifford algebra}
\begin{definition}
  The Clifford algebra $\binom{1,-1}{\mathfrak K}$ is the algebra over $\mathfrak K$ with a pair of generators $\mathbf{j,k}$, satisfying $\mathbf j^2=-1,\mathbf k^2=1$, and $\mathbf{jk}+\mathbf{kj}=0$, together with the $\mathbb Z_2$-grading such that $\mathbf{j,k}$ have odd degree, and $\mathbf{1,jk}$ have even degree.
\end{definition}
Thus, apart from the grading, the Clifford algebra is the algebra of split quaternions.

Because the generators anticommute, $\binom{1,-1}{\mathfrak K}$ is a four-dimensional algebra.  A basis for the algebra is $\mathbf 1,\mathbf j,\mathbf k,\mathbf{jk}$.  In fact, it is isomorphic to the algebra of $2\times 2$ matrices over $\mathfrak K$, with the following identifications:
$$\mathbf j = \begin{bmatrix}0&-1\\1&0\end{bmatrix},\quad\mathbf k=\begin{bmatrix}1&0\\0&-1\end{bmatrix},$$
$$\mathbf 1=\begin{bmatrix}1&0\\0&1\end{bmatrix},\quad\mathbf{jk}=\begin{bmatrix}0&1\\1&0\end{bmatrix}.$$
 Define an antiinvolution $\tau$ of $\binom{1,-1}{\mathfrak K}$, by extending $\tau(\mathbf 1)=\mathbf 1,\tau(\mathbf j)=-\mathbf j,\tau(\mathbf k)=-\mathbf k$, to the unique antiinvolution, $\tau(\mathbf{jk})=\tau(\mathbf k)\tau(\mathbf j)=(-\mathbf k)(-\mathbf j)=-\mathbf{jk}$.

If $x$ is an element of the Clifford algebra, then $\tau(x)x$ is a multiple of the identity:
\begin{proposition}
  $$\tau(x)x = (\det x)\mathbf 1.$$
\end{proposition}
Thus $\tau(x)$ is the element of the Clifford algebra corresponding to the classical adjoint of $x$.
\begin{proof}
  The stated identity is true for the basis $\mathbf{1,j,k,jk}$.  So, it is sufficient to show that the form $q(x) = \tau(x)x$ is diagonal in that basis.  That is, we must prove that $\tau(a)b+\tau(b)a=0$ for all distinct pairs $(a,b)$ in $\{\mathbf{1,j,k,jk}\}$.  This comes to six conditions:
  $$\tau(\mathbf 1)\mathbf j+\tau(\mathbf j)\mathbf 1=0,\quad \tau(\mathbf 1)\mathbf k+\tau(\mathbf k)\mathbf 1=0,\quad \tau(\mathbf 1)\mathbf j\mathbf k+\tau(\mathbf j\mathbf k)\mathbf 1=0,$$
  $$\tau(\mathbf j\mathbf k)\mathbf j+\tau(\mathbf j)\mathbf j\mathbf k=0,\quad \tau(\mathbf j\mathbf k)\mathbf k+\tau(\mathbf k)\mathbf j\mathbf k=0,$$
  $$\tau(\mathbf j)\mathbf k+\tau(\mathbf k)\mathbf j=0$$
  all of which are readily observed.
\end{proof}

In particular, the set of elements $x$ of $\binom{1,-1}{\mathfrak K}$ satisfying $\tau(x)x\not=0$ is a copy of $\op{GL}_2(\mathfrak K)$.  Denote by $\op{Spin}(2,1)\subset\binom{1,-1}{\mathfrak K}$ be the set of $x$ satisfying $\tau(x)x=1$.  Then, $\op{Spin}(2,1)\cong\op{SL}_2(\mathfrak K)$.  An element $x\in\op{Spin}(2,1)$ acts on the quaternion algebra by the adjoint action
\begin{equation}\label{CliffordInner}
  \op{Ad}_x:y\mapsto \tau(x)yx = x^{-1}yx.
\end{equation}

The center of $\op{Spin}(2,1)$, which consists of $\pm\mathbf 1$, acts trivially on the quaternion algebra $\binom{1,-1}{\mathfrak K}$  The elements of $\op{Spin}(2,1)$ that are even with respect to the grading of $\binom{1,-1}{\mathfrak K}$ form a subgroup isomorphic to the group $\op{SO}(1,1)$, given by all linear combinations of the form $x=u\mathbf 1+v\mathbf{jk}=\begin{bmatrix}u&v\\v&u\end{bmatrix}$, where $\tau(x)x=u^2-v^2=1$.

\subsection{Hyperbolic structures}\label{Cl11Structures}

Let $\mathbb T$ be a topological vector space.  A pair $\mathcal{J,K}$ consisting of a complex structure $\mathcal J$, and an anticommuting reflection, defines in a natural way a representation of the Clifford algebra $\binom{1,-1}{\mathfrak K}$ on $\mathbb T$.

\begin{definition}
  A hyperbolic structure on the vector space $\mathbb T$ is a pair $(\mathcal J,\mathcal K)$ of isomorphisms of $\mathbb T$, such that $\mathcal J^2=-\mathcal I$, $\mathcal K^2=\mathcal I$, $\mathcal{JK}+\mathcal{KJ}=0$. 
\end{definition}

Equivalently, a $\binom{1,-1}{\mathfrak K}$ structure is associated with a homomorphism of unital algebras $\mathscr J:\binom{1,-1}{\mathfrak K}\to L(\mathbb T)$, where
$$\mathcal J=\mathscr J(\mathbf j),\quad\mathcal K=\mathscr J(\mathbf k).$$

For any hyperbolic structure $\mathcal J,\mathcal K$, the subspaces $\mathbb K_{\pm} = \ker(\mathcal I\mp\mathcal K)$ and $\mathbb J_{\pm}=\ker(\mathcal I\mp \mathcal{JK})$ are pairwise complementary.  We have
\begin{align*}
  [\mathbb K_-\mathbb K_+] &= \mathcal K\\
  [\mathbb J_-\mathbb J_+] &= -\mathcal J\mathcal K\\
  [\mathbb J_+\mathbb K_+] &= \mathcal{JK}-\mathcal J-\mathcal K\\
  [\mathbb J_-\mathbb K_-] &= -\mathcal J\mathcal K-\mathcal J+\mathcal K\\
  [\mathbb J_+\mathbb K_-] &= \mathcal J\mathcal K + \mathcal J + \mathcal K\\
  [\mathbb J_-\mathbb K_+] &= -\mathcal J\mathcal K + \mathcal J - \mathcal K.
\end{align*}
\begin{proof}\mbox{}
  
  \begin{itemize}
  \item The first is obvious.
  \item For the second, if $x\in\mathbb J_-$, then $-\mathcal{JK}x=x$, while if $x\in\mathbb J_+$, then $-\mathcal{JK}x=-x$.
  \item Let $x\in\mathbb J_+$.  Then $\mathcal J\mathcal Kx = x$ and $\mathcal Kx = -\mathcal Jx$.  Thus $(\mathcal J\mathcal K-\mathcal J-\mathcal K)x=\mathcal J\mathcal Kx=x$.  On the other hand, let $x\in\mathbb K_+$.  Then $\mathcal Kx=x$, and so $(\mathcal J\mathcal K-\mathcal J-\mathcal K)x=-\mathcal Kx=-x$.
  \item Let $x\in\mathbb J_-$.  Then $\mathcal J\mathcal Kx = -x$ and $\mathcal Kx = \mathcal Jx$.  Thus $(-\mathcal J\mathcal K-\mathcal J+\mathcal K)x=-\mathcal J\mathcal Kx=x$.  On the other hand, let $x\in\mathbb K_-$.  Then $\mathcal Kx=-x$, and so $(-\mathcal J\mathcal K-\mathcal J+\mathcal K)x=\mathcal Kx=-x$.
  \item Let $x\in\mathbb J_+$.  Then $\mathcal J\mathcal Kx = x$ and $\mathcal Kx = -\mathcal Jx$.  Thus $(\mathcal J\mathcal K+\mathcal J+\mathcal K)x=\mathcal J\mathcal Kx=x$.  On the other hand, let $x\in\mathbb K_-$.  Then $\mathcal Kx=-x$, and so $(\mathcal J\mathcal K+\mathcal J+\mathcal K)x=\mathcal Kx=-x$.
  \item Let $x\in\mathbb J_-$.  Then $\mathcal J\mathcal Kx = -x$ and $\mathcal Kx = \mathcal Jx$.  Thus $(-\mathcal J\mathcal K+\mathcal J-\mathcal K)x=-\mathcal J\mathcal Kx=x$.  On the other hand, let $x\in\mathbb K_+$.  Then $\mathcal Kx=x$, and so $(-\mathcal J\mathcal K+\mathcal J-\mathcal K)x=-\mathcal Kx=-x$.
  \end{itemize}
\end{proof}

There is a block decomposition
$$\mathcal J = \begin{bmatrix}0 & J\\ -J^{-1}&0\end{bmatrix}_{\mathbb K_+\mathbb K_-}$$
where $J = \mathcal K_+\mathcal J\mathcal K_- : \mathbb K_-\to\mathbb K_+$ is an isomorphism.


Given a hyperbolic structure $(\mathcal J,\mathcal K)$, its orbit under the action of the group $O(1,1)$ can be described explicitly using the rational parameterization $C(s)=\frac{1+s^2}{1-s^2},S(s)=\frac{2s}{1-s^2}$, for $s^2\not=1$.  Then
$$(\mathcal J(s),\mathcal K(s))=\left(C(s)\mathcal J + S(s)\mathcal K,S(s)\mathcal J+C(s)\mathcal K\right)$$
is also a hyperbolic structure.  Note that $\mathcal J(s)^2=-\mathcal I$, $\mathcal K(s)^2=\mathcal I$, and $$\mathcal J(s)\mathcal K(s)+\mathcal K(s)\mathcal J(s)=0.$$

\begin{definition}
  The orbit of a given hyperbolic structure $(\mathcal J,\mathcal K)$ under the action of $\op{O}(1,1)$ as described above, is called the {\em pencil} associated to $(\mathcal J,\mathcal K)$.
\end{definition}

\subsection{Affine geodesics}\label{AffineGeodesicSection}

\begin{definition}
  Let $\mathbb A\in\mathcal{G}_2(\mathbb T)$, and $\mathcal J$ an isomorphism of $\mathbb T$ such that $\mathcal J^2=-\mathcal I$.  Then the {\em affine geodesic} associated to the pair $(\mathcal J, \mathbb A)$ is the curve in $\mathcal{G}_2(\mathbb T)$ given by
  $$\mathbb L(\mathcal J,\mathbb A)(u) = \{ ua+\mathcal Ja\ |\ a\in\mathbb A\}.$$
\end{definition}
The parameter $u$ is the {\em affine parameter} on the geodesic.  Note that affine geodesics, by definition, have a preferred parametrization.

Since the pair $(\mathbb A,\mathcal J)$ is associated with the hyperbolic structure $(\mathcal J,\mathcal K)$, where $\mathcal K=[\mathbb A(\mathcal J\mathbb A)]$, it shall be conventient to regard $\mathbb L$ as a function of the hyperbolic structure, and so we write $\mathbb L(\mathcal J,\mathcal K)$.

Note that the affine geodesic, as a function of the real variable $u$, is well-defined as a function with values in subspaces.   These subspaces are in fact Lagrangian.  Indeed, to show that $\mathbb L(\mathbb A,\mathcal J)(u)$ is Lagrangian, we let $a_1,a_2\in\mathbb A$, and compute
$$\omega(ua_1+\mathcal Ja_1,ua_2+\mathcal Ja_2) = 0 $$
using along the way the identity $\omega(a_1,\mathcal Ja_2)+\omega(\mathcal Ja_1,a_2)=0$.

Summarizing, for each hyperbolic structure $(\mathcal J,\mathcal K)$, we have associated a nondegenerate curve $\mathbb L(\mathcal J,\mathcal K) : \mathfrak K\to\mathcal{G}(\mathbb T)$.

The following describes how the affine geodesic behaves on passing to another hyperbolic structure in a given pencil:
\begin{proposition}\label{Kintertwined}
  Let $(\mathcal J,\mathcal K)$ be a given hyperbolic structure, and for $s\not=\pm 1$ let
  $$\mathcal J(s) = C(s)\mathcal J + S(s)\mathcal K,\quad \mathcal K(s) = C(s)\mathcal K + S(s)\mathcal J$$
  be an element of the pencil associated to $(\mathcal J,\mathcal K)$, where $S(s)=2s/(1-s^2)$ and $C(s)=(1+s^2)/(1-s^2)$.  Then, for all $u\not=-1/s$,
  $$\mathbb L(\mathcal J(s),\mathcal K(s))(u) = \mathbb L(\mathcal J,\mathcal K)\left(\frac{u+s}{su+1}\right).$$
\end{proposition}

\begin{proof}
  We first claim that the eigenspace of $\mathbb K(s)$, for the $+1$ eigenvalue is
  $$\mathbb K(s)_+=\{k+s\mathcal Jk\ |\ k\in\mathbb K_+\}.$$
  It is sufficient to verify that $\mathcal K(s)(k+s\mathcal Jk)=k+s\mathcal Jk$ for all $k\in\mathbb K_+$.  Then
  $$\mathcal K(s)(k+s\mathcal Jk) = (C(s)-sS(s))k + (-sC(s) + S(s))\mathcal Jk = k+s\mathcal Jk $$
  as required.

  Also,
  $$\mathcal J(s)(k+s\mathcal Jk) = (S(s)-sC(s))k+(C(s)-sS(s))\mathcal Jk=sk+\mathcal Jk $$

  Now,
  \begin{align*}
    \mathbb L(\mathcal J(s),\mathcal K(s))(u)
    &= \left\{ux + \mathcal J(s)x\ \mid\ x\in\mathbb K(s)_+ \right\}\\
    &=\left\{( uk + us\mathcal Jk) + \mathcal J(s)(k+s\mathcal Jk)\ \mid\ k\in\mathbb K_+ \right\}\\
    &=\left\{( uk + us\mathcal Jk) + (sk+\mathcal Jk)\ \mid\ k\in\mathbb K_+ \right\}\\
    &=\left\{( u + s )k + (us+1)\mathcal Jk\ \mid\ k\in\mathbb K_+ \right\}\\
    &=\left\{\frac{u + s}{su+1}k + \mathcal Jk\ \mid\ k\in\mathbb K_+ \right\},
  \end{align*}
  after a change of variable $k\to (su+1)k$.
\end{proof}

\subsection{Action of $\op{SL}_2(\mathfrak K)$}

\begin{theorem}\label{SL2line}
Allow $\op{SL}_2(\mathfrak K)$ to act on $\mathfrak K$ by fractional linear transformations.  Then, for $\phi\in\op{SL}_2(\mathfrak K)$,
$$\mathbb L(\mathcal J,\mathcal K)\circ\phi=\mathbb L(\op{Ad}_\phi\mathcal J,\op{Ad}_\phi\mathcal K)$$
wherever both sides are defined.
\end{theorem}

\begin{proof}

  We consider $\phi\in\op{SL}_2(\mathfrak K)$ as a unit quaternion,
  $$\phi = a\mathbf 1 + b\mathbf j + c\mathbf k + d\mathbf{jk} = \begin{bmatrix}a+c&-b+d\\ b+d& a-c\end{bmatrix}.$$
  $$(a,b,c,d\in\mathfrak K,\quad a^2+b^2-c^2-d^2=1)$$
  Suppose that $\op{Ad}_\phi\mathcal K x = x$, i.e.,
  $$\mathcal K\mathscr J(\tau\phi) x = \mathscr J(\tau\phi) x.$$
  So $\mathscr J(\tau\phi) x\in\mathbb K_+$.  Let $\bar x = \mathscr J(\tau\phi) x$.  Note that
  $$\phi\mathbf j = a\mathbf j - b\mathbf 1 - c\mathbf{jk} + d\mathbf k.$$
  Applying this identity, we have
  \begin{align*}
    ux + \op{Ad}_\phi\mathcal Jx
    &= u\mathscr J(\phi) \bar x + \mathscr J(\phi)\mathcal J\bar x\\
    &= (u(a+c) + (-b+d))\bar x + (u(b+d)+(a-c))\mathcal J\bar x.
  \end{align*}
  Finally, let $y=(u(b+d)+(a-c))\bar x\in\mathbb K_+$, so that
  $$ux + \op{Ad}_\phi\mathcal Jx = \frac{u(a+c)+(-b+d)}{u(b+d)+(a-c)}y + \mathcal J y = \phi(u)y + \mathcal Jy,$$
  as required.
\end{proof}

\subsection{The reflection formula}\label{ReflectionsBanach}
The purpose of this section is to prove Theorem \ref{ReflectionFormula}, which gives a formula for a generic reflection in terms of a fixed pair of separated subspaces of $\mathbb T$ that provide an infinite-dimensional analog of an ``affine chart''. 


\begin{lemma} Suppose that $\mathbb A\in\mathcal U(\mathbb B)$ and $\mathbb X\in\mathcal U(\mathbb Y)$.  Let $X,Y:\mathbb A\to\mathbb B$ be the morphisms whose graphs are $\mathbb X$ and $\mathbb Y$, respectively.  Then $Y-X:\mathbb A\to\mathbb B$ is an isomorphism.
\end{lemma}
\begin{proof}
Consider the mapping
$$M:\mathbb A\oplus\mathbb A \to \mathbb X \oplus \mathbb Y $$
given by
$$M(x,y) = (x+Xx,y+Yy).$$
Then $M$ is a morphism, which is one-to-one and also onto.  Therefore $M$ is an isomorphism (because morphisms are open by definition).

  We factorize $M$ as $M=NE$ where $E:\mathbb A\oplus\mathbb A\to\mathbb A\oplus\mathbb A$ is defined by $E(x,y)=(x+y,y)$ and $N:\mathbb A\oplus\mathbb A\to \mathbb A + \mathbb B$ is $N(x,y)=(x, Xx + (Y-X)y)$.
Because $M$ and $E$ are isomorphisms, $N$ is also an isomorphism.  The restriction of $N$ to the subspace $0\oplus\mathbb A$ maps one-to-one and onto a closed subspace of $\mathbb B$.  In fact, it must map onto all of $\mathbb B$, because $\mathbb B$ is in the range of $N$, and $N(x,y)\in\mathbb B$ if and only if $x=0$.  
\end{proof}

\begin{definition}
  For any reflection $\mathcal K$ in $\mathbb T$, if $M : \mathbb T\to\mathbb T$ is a morphism, and $\mathcal K_+M\mathcal K_+=A, \mathcal K_+M\mathcal K_-=B, \mathcal K_-M\mathcal K_+=C$, and $\mathcal K_-M\mathcal K_-=D$, we shall denote $M$ by the $2\times 2$ block matrix
  $$M = \begin{bmatrix}A&B\\C&D\end{bmatrix}_{\mathcal K}.$$
  The subscript shall sometimes be omitted where it is clear from context what $\mathcal K$ is.
\end{definition}

\begin{theorem}\label{ReflectionFormula}
  Suppose that $\mathbb A\in\mathcal U(\mathbb B)$ and $\mathbb X\in\mathcal U(\mathbb Y)$.  Then the reflection $[\mathbb{XY}]$ has a decomposition as a block matrix $\mathbb A\oplus\mathbb B\to\mathbb A\oplus\mathbb B$,
$$[\mathbb{XY}] =
\begin{bmatrix}
  -(X-Y)^{-1}(X+Y) & 2(X-Y)^{-1}\\
  -2X(X-Y)^{-1}Y & (X+Y)(X-Y)^{-1}
\end{bmatrix}_{[\mathbb A\mathbb B]}.
$$
\end{theorem}
\begin{proof}
  We write $[\mathbb{XY}]$ in terms of a pair of block matrices $\mathbb A\oplus\mathbb A\to\mathbb A\oplus\mathbb B$ as:
  $$[\mathbb{XY}] = \begin{bmatrix}
    I&I\\
    X&Y
  \end{bmatrix}
  \begin{bmatrix}
    I&-I\\
    X&-Y
  \end{bmatrix}^{-1}$$
  The right matrix is
  $$\begin{bmatrix}
    I&-I\\
    X&-Y
  \end{bmatrix}^{-1}=
  \begin{bmatrix}
    -(X-Y)^{-1}Y & (X-Y)^{-1}\\
    -(X-Y)^{-1}X&(X-Y)^{-1}
  \end{bmatrix}.$$
  Multiplying now gives the Lemma, after noting that
  $$ X(X-Y)^{-1}Y = Y(X-Y)^{-1}X.$$
\end{proof}
\begin{corollary}
  Suppose that $\mathbb B,\mathbb X\in\mathcal U(\mathbb A)$.  Then the reflection $[\mathbb{AX}]$ has a decomposition as a block matrix $\mathbb A\oplus\mathbb B\to\mathbb A\oplus\mathbb B$,
$$[\mathbb{AX}] =
\begin{bmatrix}
  I & -2X^{-1}\\
  0 & -I
\end{bmatrix}_{[\mathbb A\mathbb B]}.
$$
\end{corollary}

\section{Differential calculus in Grassmannians}\label{DifferentialCalculus}
Here and henceforth, assume that $\mathbb T$ is a Banach space over a complete archimedean field of characteristic zero.  

\subsection{Grassmannians}
\begin{definition}
The {\em Grassmannian} of subspaces of $\mathbb T$ is the set of all closed linear subspaces $\mathbb A\subset\mathbb T$.  The Grassmannian of $\mathbb T$ is denoted by $\mathcal G_0(\mathbb T)$.
\end{definition}

\begin{definition}
  Let $U$ be a set.  A function $\mathbb H:U\to \mathcal{G}_0(\mathbb T)$ is called nondegenerate if, for all distinct $s,t\in U$, $\mathbb H(s)$ and $\mathbb H(t)$ are complementary.  A $k$-tuple $(\mathbb A_1,\mathbb A_2,\dots,\mathbb A_k)$ is called nondegenerate if its constituents are pairwise complementary (i.e., $i\mapsto \mathbb A_i$ is a nondegenerate function from the $k$-element set into $\mathcal G_0(\mathbb T)$).
\end{definition}

We next define a filtration $\mathcal G_0\supset\mathcal G_1\supset\mathcal G_2\supset\mathcal G_3\supset\cdots$ by subsets, each having successively stronger topologies.  

\begin{definition}
  Let $k$ be a positive integer, and let $\mathcal G_k(\mathbb T)\subset\mathcal G_0(\mathbb T)$ be the set of all possible values of $\mathbb H(0)$ for some nondegenerate function $\mathbb H:[k]\to\mathcal G_0(\mathbb T)$ on the $(k+1)$-element set $[k]=\{0,1,\dots,k\}$:
  $${\mathcal G}_k(\mathbb T) = \{ \mathbb H(0)\ |\ \mathbb H:[k]\to \mathcal G(\mathbb T)\quad\text{nondegenerate}\}.$$
  On $\mathcal G_k(\mathbb T)$ define a subbase for a topology, whose open sets $\mathcal U(\mathbb A_1,\mathbb A_2,\dots,\mathbb A_k)$ comprise the set of $\mathbb A_0\in\mathcal G_0(\mathbb T)$ such that the function $\mathbb H(i) = \mathbb A_i$, $i\in[k]$, is nondegenerate.
\end{definition}

The first few $\mathcal G_k(\mathbb T)$ are then as follows.
\begin{itemize}
\item $\mathcal G_1(\mathbb T)$ comprises the {\em complemented} subspaces of $\mathbb T$.  The basic open set $\mathcal U(\mathbb A)$ is the set of all complements of $\mathbb A$ in $\mathbb T$.
\item $\mathcal G_2(\mathbb T)$ is the space of $\mathbb A\in\mathcal G_0(\mathbb T)$ for which there is a triple $(\mathbb A,\mathbb B,\mathbb C)$ of pairwise complementary subspaces, which we call the {\em middle Grassmannian}.  In a Hilbert space, for example, the middle Grassmannian is the set of closed subspaces whose dimension and codimension are the same cardinal, and this justifies the name.  The basic open set $\mathcal U(\mathbb A,\mathbb B)$ is the set of $\mathbb C$ such that the triple $(\mathbb A,\mathbb B,\mathbb C)$ is pairwise complementary.
\item $\mathcal G_3(\mathbb T)$ is the space of $\mathbb A$ such that there exists a pairwise complementary quartuple $(\mathbb A,\mathbb B,\mathbb C,\mathbb D)$.  It is on $\mathcal G_3(\mathbb T)$ that the cross ratio of four subspaces is naturally formulated.
\end{itemize}

We note the following, which is obvious from the definitions:
\begin{lemma}
  $$\mathbb A_0\in \mathcal U(\mathbb A_1,\dots,A_k)\implies \mathbb A_i\in\mathcal U(\mathbb A_0,\mathbb A_1,\dots,\widehat{\mathbb A_i},\dots,\mathbb A_k)$$
  for all $i\in[k]$, where the hat denotes omission of the corresponding term.
\end{lemma}

  

\begin{definition}
Two points $\mathbb A_0,\mathbb A_1\in\mathcal G_2(\mathbb T)$ of the middle Grassmannian are said to be {\em separated} if there exists an $\mathbb A_2\in\mathcal G_2(\mathbb T)$ such that $\mathbb H(i)=\mathbb A_i$, $i=0,1,2$, is a nondegenerate function from the three point set into the Grassmannian $\mathcal G(\mathbb T)$.
\end{definition}

Thus separatedness is a stronger condition than complementarity in general.  However, in a Hilbert space they are the same:

\begin{lemma}
Suppose that $\mathbb T$ is a real Hilbert space and $\mathbb A,\mathbb B\in\mathcal G_2(\mathbb T)$ are complementary.  Then $\mathbb A$ and $\mathbb B$ are separated.
\end{lemma}
\begin{proof}
  Because they have equal dimension, there is an isomorphism $T:\mathbb A\to\mathbb B$.  Then take the subspace to be:
  $$\mathbb C = \{x + Tx\ |\ x\in\mathbb A\} = \{y + T^{-1}y\ |\ y\in\mathbb B\}.$$
  Then $\mathbb A+\mathbb C=\mathbb B + \mathbb C = \mathbb A+\mathbb B=\mathbb T$, and $\mathbb A\cap\mathbb C=\mathbb B\cap\mathbb C=\{0\}$.  We exhibit the reflection $[\mathbb{AC}]$, that of $[\mathbb{BC}]$ given by a symmetrical construction.  For $x\in\mathbb A$, $[\mathbb{AC}]x = x$ and $[\mathbb{AC}](x+Tx) = -x-Tx$,, so $[\mathbb{AC}](Tx)=-2x-Tx$. Therefore,
  $$[\mathbb{AC}] = [\mathbb{AB}]_+ - (2T^{-1}+I)[\mathbb{AB}]_-$$
  which is continuous.
\end{proof}

\subsection{Example}
We give an example to show that some of the basic intuitions about complementarity from the finite dimensional case can fail in infinite-dimensions.  Specificially, we shall construct pairwise complementary triples $(\mathbb A_1,\mathbb A_2,\mathbb A_3)$ and $(\mathbb B_1,\mathbb B_2,\mathbb B_3)$, such that $\mathbb A_3\cap\mathbb B_1=\{0\}$ but $\mathbb A_3+\mathbb B_1\not=\mathbb T$.

Suppose that $\mathfrak K$ is a complete normed field, and let $A$ be the set of sequences $a=(a_0,a_1,\dots)$, with $a_n\in\mathfrak K$ and all but finitely many $a_n=0$.  Let $S(a_0,a_1,\dots)=(0,a_0,a_1,\dots)$ be the right-shift operator.  Finally, let $\rho$ be any $S$-invariant norm on $A$, and $\ell^\rho$ denote the completion of $A$ under the norm $\rho$.  Then $\ell^\rho$ is a Banach space, with an unconditional Schauder basis given by the unit coordinate vectors $e_n$ (which is the zero sequence except at $n$, where it is one.)

 Let $\mathbb A_1$ be the closed span of $\{e_{2n+1}|n\in\mathbb N\}$, $\mathbb A_2$ the closed span of $\{e_{2n}|n\in\mathbb N\}$, and $\mathbb A_3$ the closed span of $\{e_{2n}+e_{2n+1}|n\in\mathbb N\}$.  Then $\mathbb A_1,\mathbb A_2,\mathbb A_3$ are pairwise complementary, with reflections
\begin{align*}
[\mathbb A_1\mathbb A_2](a_0,a_1,\dots) &= (-a_0,a_1,-a_2,\dots)\\
[\mathbb A_1\mathbb A_3](a_0,a_1,\dots) &= (-a_0,a_1-2a_0,-a_2,a_3-2a_2,\dots)\\
[\mathbb A_2\mathbb A_3](a_0,a_1,\dots) &= (a_0-2a_1,-a_1,a_2-2a_3,-a_3,\dots).
\end{align*}
Therefore $\mathbb A_i$ belong to the middle Grassmannian of $\ell^\rho$.

Let $\mathbb B_1\subset \ell^\rho$ be the closed span of $\{e_{3n}|n\in\mathbb N\}$, $\mathbb B_2$ the closed span of $\{e_{3n+1},e_{3n+2}\}$, and $\mathbb B_3$ the closed span of $\{e_{6n}+e_{3n+1},e_{6n+3}+e_{3n+2}\}$.  Then the triple $(\mathbb B_1,\mathbb B_2,\mathbb B_3)$ is nondegenerate, with reflection operators
\begin{align*}
  [\mathbb B_1\mathbb B_2](a_0,a_1,\dots) &= (a_0,-a_1,-a_2,a_3,\dots)\\
  [\mathbb B_1\mathbb B_3]e_{3n} &= e_{3n}\\
  [\mathbb B_1\mathbb B_3]e_{3n+1} &= -2e_{6n}-e_{3n+1}\\
  [\mathbb B_1\mathbb B_3]e_{3n+2} &= -2e_{6n+3}-e_{3n+2}\\
  [\mathbb B_2\mathbb B_3]e_{3n+1} &= e_{3n+1}\\
  [\mathbb B_2\mathbb B_3]e_{3n+2} &= e_{3n+2}\\
  [\mathbb B_2\mathbb B_3]e_{6n} &= -2e_{3n+1}-e_{6n}\\
  [\mathbb B_2\mathbb B_3]e_{6n+3} &= -2e_{3n+2}-e_{6n+3}.
\end{align*}

Thus $(\mathbb A_1,\mathbb A_2,\mathbb A_3)$ and $(\mathbb B_1,\mathbb B_2,\mathbb B_3)$ are each nondegenerate triples.  Now, note that on the one hand $\mathbb A_3\cap \mathbb B_1=\{0\}$, but on the other hand $\mathbb A_3+\mathbb B_1\not=\ell^\rho$.  (For example $e_4$ does not belong to the intersection.)  Therefore $\mathbb A_3$ and $\mathbb B_1$ are not complementary.


\subsection{Tangent vector}

\begin{definition}
The tangent space to the Grassmannian $\mathcal{G}(\mathbb T)$ at a point $\mathbb A$ is the space of morphisms $\mathbb A\to\mathbb T/\mathbb A$.
\end{definition}

To make good use of this definition, we shall need an operator topology $\sigma$ on the space $\mathscr L(\mathbb T,\mathbb T)$ of endomorphisms of $\mathbb T$, which is at least as strong as the ``strong operator topology''.  To do calculus it is desirable, but probably not always essential, to make the topology metrizable.  Rather than strive for the utmost generality, we have therefore opted to use the {\bf norm topology} on $\mathscr L(\mathbb T,\mathbb T)$

With these assumptions in place, we now move on to questions of differential calculus.


\begin{definition}
  Let $\mathbb H:U\to\mathcal U(\mathbb B)\cap\mathcal G_2(\mathbb T)$ be a function from an open neighborhood of $u_0$ in $\mathcal K$.  A bounded linear operator $\mathbb H'_{\mathbb B}(u_0):\mathbb H(u_0)\to\mathbb T/\mathbb H(s)$ is called a {\em tangent vector} to $\mathbb H$ at $u_0$ relative to $\mathbb B$ if
  $$[\mathbb H(s)\mathbb B] - [\mathbb H(t)\mathbb B] = (u-s)\mathbb H'_{\mathbb B}(u_0) + o(\|u-s\|_{\mathfrak K})$$
  as $(s,t)\to (u_0,u_0)$ in $U\times U$.  If $\mathbb H$ has a tangent vector at $u_0$, it is called {\em differentiable} at $u_0$.
\end{definition}
(Note that our definition of differentiability is sometimes referred to as {\em strict} differentiability.)

Relative to the splitting $\mathbb H(u_0)\oplus\mathbb B$, we can write the curve $\mathbb H(u)$ as the image of the block matrix
$$\begin{bmatrix}I\\ H(u)\end{bmatrix}:\mathbb H(u_0)\oplus\mathbb B\to \mathbb H(u_0)\subset\mathbb T.$$
Here $H(u_0)=0$, and the derivative is $H'(u_0)$, which maps $\mathbb H(u_0)$ into $\mathbb B$, which is complementary to $\mathbb H(u_0)$.

The derivative $\mathbb H'_{\mathbb B}(u_0):\mathbb H(u_0)\to\mathbb T/\mathbb H(u_0)$ is independent of the choice of $\mathbb B$ in the definition.  More precisely, we have
\begin{theorem}
  Suppose that $\mathbb H:U\to\mathcal U(\mathbb B)\cap\mathcal U(\mathbb B')$.  If $\mathbb H$ is differentiable at $u_0\in U$ relative to $\mathbb B$, then it is differentiable relative to $\mathbb B'$, and two derivatives coincide as operators on $\mathbb H(u_0)\to\mathbb T/\mathbb H(u_0)$:
$$\mathbb H'_{\mathbb B}(u_0) = \mathbb H'_{\mathbb B'}(u_0).$$
\end{theorem}
\begin{proof}
  Take $u_0=0$ without loss of generality.
  In the splitting $\mathbb T=\mathbb H(0)\oplus\mathbb A$, we have
  $$[\mathbb H(u)\mathbb A] = \begin{bmatrix}I&0\\ 2H(u)&-I\end{bmatrix}_{[\mathbb H(0)\mathbb A]} $$
  where $H(u):\mathbb H(0)\to\mathbb A$, such that $H(0)=0$.  Differentiability of $[\mathbb H(u)\mathbb A]$ at $u=0$ is equivalent to the existence of $H'(0)$.  Also
  $$[\mathbb H(u)\mathbb B]= \begin{bmatrix}
    -(H(u)-B)^{-1}(H(u)+B) & 2(H(u)-B)^{-1}\\
    -2H(u)(H(u)-B)^{-1}B & (H(u)+B)(H(u)-B)^{-1}
  \end{bmatrix}_{[\mathbb H(0)\mathbb A]}
  $$
  where $B:\mathbb H_0\to\mathbb A$ such that $\mathbb B = (\mathcal I+B)\mathbb H_0$.  The derivative of $[\mathbb H(u)\mathbb B]$ at $u=0$ thus exists.  Subtracting and simplifying,
  $$[\mathbb H(u)\mathbb B] - [\mathbb H(u)\mathbb A] = \begin{bmatrix}
    2(H(u)-B)^{-1}H(u) & -2(H(u)-B)^{-1}\\
    2H(u)(H(u)-B)^{-1}H(u) & -2 H(u)(H(u)-B)^{-1}
  \end{bmatrix}_{[\mathbb H(0)\mathbb A]}.
  $$
  Evaluating on $\mathbb H(0)$ selects out the first column.  The derivative of the first column at $u=0$ has bottom entry $0$, because $H(0)=0$.  Therefore $$\left.\frac{d}{du}\right|_{u=0}\left([\mathbb H(u)\mathbb B] - [\mathbb H(u)\mathbb A]\right)\mathbb H(0) \subset\mathbb H(0),$$ and so the two derivatives coincide as operators $\mathbb H(0)\to\mathbb T/\mathbb H(0)$.
\end{proof}

\subsection{Tangent covector}\label{TangentCovector}

In fact, it is simpler to describe the tangent {\em covector} to a nondegenerate curve.

\begin{definition}
  The {\em cotangent space} to $\mathcal G_2(\mathbb T)$ at $\mathbb A$ is the set of endomorphisms $\mathcal J$ on $\mathbb T$ such that:
  $$\im\mathcal J \subset \mathbb A\subset \ker\mathcal J.$$
\end{definition}

\begin{definition}
  Let $U$ be an open neighborhood of $u_0$ in the field $\mathfrak K$.  A nondegenerate function $\mathbb H:U\to\mathcal G(\mathbb T)$ is called {\em codifferentiable at $u_0$} if there exists an endomorphism $\mathbb H_\prime(u_0)$ of $\mathbb T$ such that, as $h\to 0$,
  $$[\mathbb H(u_0+h)\mathbb H(u_0)] = 2h^{-1}\mathbb H_\prime(u_0) + o(h^{-1}).$$
  The operator $\mathbb H_\prime(u_0)$ is called the {\em tangent covector} to $\mathbb H$ at $u_0$.
\end{definition}

\begin{theorem}\label{DifferentialOfCurve}
  For a nondegenerate curve $\mathbb H$ codifferentiable at $u_0$, the tangent covector $\mathbb H_\prime(u_0)$  is a covector: that is, $\mathbb H_\prime(u_0)|\mathbb H(u_0)=0$ and $\op{im}\mathbb H_\prime(u_0)\subset \mathbb H(u_0)$.  
\end{theorem}

\begin{proof}
  Taking $u_0=0$, we must compute the limit $\lim_{u\to 0}u[\mathbb H(u)\mathbb H(0)]$.  Fix a splitting $\mathbb T=\mathbb A\oplus\mathbb B$ as in Theorem \ref{ReflectionFormula}. Then $u[\mathbb H(u)\mathbb H(0)] =$
  $$\begin{bmatrix}-\left(\frac{H(u)-H(0)}{u}\right)^{-1}(H(u)+H(0))& 2\left(\frac{H(u)-H(0)}{u}\right)^{-1}\\
-2H(u)\left(\frac{H(u)-H(0)}{u}\right)^{-1}H(0)& (H(u)+H(0))\left(\frac{H(u)-H(0)}{u}\right)^{-1}
\end{bmatrix}_{[\mathbb A\mathbb B]}.$$
By hypothesis, the limit as $u\to 0$ exists.  Denote by $\dot H(0)^{-1}$ the limit
$$\dot H(0)^{-1}=\lim_{u\to 0}\left(\frac{H(u)-H(0)}{u}\right)^{-1}$$
which exists by hypothesis, being the northeast block of the matrix above.  Then 
\begin{align*}
  \mathbb H_\prime(0)
  &= \frac12\begin{bmatrix}-2\dot H(0)^{-1}H(0) & 2\dot H(0)^{-1}\\ -2H(0)\dot H(0)^{-1}H(0)& 2H(0)\dot H(0)^{-1} \end{bmatrix}_{[\mathbb A\mathbb B]} \\
  &= \begin{bmatrix}I\\H(0)\end{bmatrix}\dot H(0)^{-1}\begin{bmatrix}-H(0) & I\end{bmatrix}
\end{align*}
which evidently satisfies $\mathbb H_\prime(0)^2=0$ and $\op{im}\mathbb H_\prime(0)\subset\mathbb H(0)\subset\ker\mathbb H_\prime(0)$.
\end{proof}

\begin{lemma}
  If $\mathbb H$ is differentiable and codifferentiable at $u_0$, then
$$\mathbb H_\prime(u_0)\mathbb H'(u_0) = 2\mathcal I_{\mathbb H(u_0)}.$$
\end{lemma}
\begin{proof}
  Without loss of generality, take $u_0=0$.
  Select $\mathbb A$ complementary to $\mathbb H(0)$.  Then
  \begin{align*}
  [\mathbb H(u)\mathbb A] &= \begin{bmatrix}I&0\\ 2H(u)&-I\end{bmatrix}_{[\mathbb H(0)\mathbb A]}\\ 
  [\mathbb H(0)\mathbb A] &= \begin{bmatrix}I&0\\ 0&-I\end{bmatrix}_{[\mathbb H(0)\mathbb A]} \\
  [\mathbb H(u)\mathbb H(0)] &= \begin{bmatrix}-I&2H(u)^{-1}\\ 0&I\end{bmatrix}_{[\mathbb H(0)\mathbb A]} \\
    \frac12[\mathbb H(u)\mathbb H(0)]([\mathbb H(u)\mathbb A]-[\mathbb H(0)\mathbb A]) &= \frac12\begin{bmatrix}-I&2H(u)^{-1}\\ 0&I\end{bmatrix}\begin{bmatrix}0&0\\ 2H(u)&0\end{bmatrix}_{[\mathbb H(0)\mathbb A]} \\
    &= \begin{bmatrix}2&0\\H(u)&0\end{bmatrix}_{[\mathbb H(0)\mathbb A]}.
  \end{align*}
  Now, because $H$ is differentiable at $u=0$, it is continuous at $u=0$, and $\lim_{u\to 0}H(u) = H(0)=0$.
\end{proof}

\begin{lemma}
  Let $\phi : U\to U'$ be a function from on open set in $\mathfrak K$ to a subset $U'\subset\mathfrak K$, such that $\phi$ is differentiable at $u_0\in U$ and $\phi'(u_0)\not=0\in\mathfrak K$.  If $\mathbb H$ is codifferentiable at $u=\phi(u_0)$, then $\mathbb H\circ\phi$ is codifferentiable at $u=u_0$, and
  $$(\mathbb H\circ\phi)_\prime(u_0) = \phi'(u_0)^{-1}\mathbb H_\prime\circ\phi(u_0).$$
\end{lemma}

\begin{proof}
  \begin{align*}
    2(\mathbb H\circ\phi)_\prime(u_0)
    &=\lim_{h\to 0} h\left[\mathbb H(\phi(u_0+h))\mathbb H(\phi(u_0))\right]\\
    &=\frac1{\phi'(u_0)}\lim_{h\to 0} (\phi(u_0+h)-\phi(u_0))\left[\mathbb H(\phi(u_0+h))\mathbb H(\phi(u_0))\right]\\
    &=\frac1{\phi'(u_0)}\lim_{z\to 0} z\left[\mathbb H(\phi(u_0)+z)\mathbb H(\phi(u_0))\right]
  \end{align*}
  where $z=\phi(u_0+h)-\phi(u_0)$.
\end{proof}

\section{The Schwarzian}
As in the last section, $\mathbb T$ is a Banach space over a complete archimedean field of characteristic zero.  

  \subsection{The one-dimensional Schwarzian}\label{OneDSchwarzian}
  We review the one-dimensional Schwarzian in our language.  Suppose then that $\mathbb T$ is a two-dimensional vector space over the reals.  Let $\mathbb U\subset\mathbb R$ be an open interval and $\mathbb H:\mathbb U\to\mathbb{PV}$ a smooth local diffeomorphism.  Fix $u_0\in\mathbb U$, and choose $\epsilon>0$ such that the open interval $\mathbb U_\epsilon=(u_0-\epsilon,u_0+\epsilon)$ is contained in $\mathbb U$, and $\mathbb H$ is one-to-one on $\mathbb U_\epsilon$.

  We have the following standard result, \cite{ovsienko2004projective}), which we prove in \S\ref{SchwarzianGeneral}:

\begin{theorem}
  Let $p_1,q_1,p_2,q_2\in\mathbb U_\epsilon$.  Then there exists a real number $\mathcal S({\mathbb H})(u_0)$, the Schwarz invariant of $\mathbb H$ at $u_0$, such that
$$\frac{[\mathbb H(p_1),\mathbb H(q_1);\mathbb H(p_2),\mathbb H(q_2)]}{[p_1,q_1;p_2,q_2]} =\left(1 + \frac16\mathcal S(\mathbb H)(u_0)(p_1-q_1)(p_2-q_2)\right)\mathcal I \quad + O(\epsilon^3)$$
  Written in terms of coordinates where $\mathbb H(u) = \op{span}\begin{bmatrix}1\\H(u)\end{bmatrix}$: 
  $$S(\mathbb H)(u_0)=\frac{\dddot H(u_0)}{\dot H(u_0)} - \frac32 \frac{\ddot H(u_0)^2}{\dot H(u_0)^2}. $$
  Furthermore, $S(\mathbb H)$ is a quadratic differential on $\mathbb U$, in the sense that if $g:\mathbb V\to\mathbb U$ is a diffeomorphism from an interval $\mathbb V$ to $\mathbb U$, then
    \begin{align*}
    &\frac{[\mathbb H(g(p_1)),\mathbb H(g(q_1));\mathbb H(g(p_2)),\mathbb H(g(q_2))]}{[g(p_1),g(q_1);g(p_2),g(q_2)]} =\\
    &\qquad=\left(1 + 6^{-1}\mathcal S(\mathbb H)(u_0)\dot g(u_0)^2(p_1-q_1)(p_2-q_2)\right)\mathcal I \quad + O(\epsilon^3).\end{align*}
\end{theorem}

For example, and to check factors, suppose that $H(p) = e^p$.  Then the Schwarzian is constant:
$$\mathcal S(H) = -\frac12.$$
On the other hand, take
$$p_1 = t,\quad p_2=2t,\quad q_1=3t,\quad q_2=4t $$
where $t$ is real.  So $[p_1,p_2;q_1,q_2]=4/3$ and
$$[H(p_1),H(p_2);H(q_1),H(q_2)] = \frac{(1+e^t)^2}{1+e^t+e^{2t}} = \frac43 - \frac{t^2}{9} +O(t^3),$$
and so
$$\frac{[H(p_1),H(p_2);H(q_1),H(q_1)]}{[p_1,p_2;q_1,q_2]} = 1 - \frac{t^2}{12} + O(t^3) = 1 + \frac16(p_1-p_2)(q_1-q_2)\mathcal S(H)(0).$$




\begin{theorem}
   $S_{\mathbb H}\equiv 0$ throughout $\mathbb U$ if and only if $\mathbb H$ is a projective mapping from the affine interval $\mathbb U$ to the projective line $\mathbb{PV}$.
\end{theorem}

\begin{proof}
  If $\mathbb H$ is a projective mapping, then it is readily shown that
  $$\left([\mathbb H(p_1)\mathbb H(p_2)]+[\mathbb H(q_1)\mathbb H(q_2)]\frac{}{}\right)^{-1}\left([\mathbb H(p_1)\mathbb H(q_2)]+[\mathbb H(p_2)\mathbb H(q_1)]\frac{}{}\right)=[p_1,q_1;p_2,q_2]$$
  so that $S_{\mathbb H}=0$. 

  Conversely, suppose that $S_{\mathbb H}\equiv 0$.  This amounts, locally, to the differential equation $2\dot H\dddot H-3\ddot H^2=0$, for the real-valued function $H$ such that $\mathbb H(u)=\op{span}\begin{bmatrix}1\\H(u)\end{bmatrix}$.   Let $g=\dot H^{-1}$.  Then, $g$ satisfies the equation
  \begin{equation}\label{discriminantODE}\dot g^2-2g\,\ddot g = 0.\end{equation}
  From \eqref{discriminantODE}, we may infer that $\dddot g=0$, so that $g$ is given by its quadratic Taylor series.  Then observe that the left-hand side of \eqref{discriminantODE} is the discriminant, so $g(t)=c^{-1}(a-t)^2$ for some real constants $a,c$.  Finally then,
  $$H(t) = \int \frac{dt}{g(t)} = c(a-t)^{-1} + b $$
  for some $b$.  Since $H$ is fractional-linear, it is a projective mapping.
\end{proof}

\subsection{The Schwarzian in general}\label{SchwarzianGeneral}
Throughout this section, $\mathbb T$ is a Banach space over the field $\mathfrak K=\mathbb R$ or $\mathbb C$.

\begin{theorem}\label{GeneralizedSchwarzianTheorem}
 Let $\mathbb H: U_\epsilon\to\mathcal G(\mathbb T)$ is a nondegenerate function from an open ball $U_\epsilon$ centered at $u_0\in\mathfrak K$ into the Grassmannian of $\mathbb T$, that is codifferentiable and at least thrice (strictly) differentiable at $u_0$.  Let $p_1,p_2,p_3,p_4\in U_\epsilon$.  Then  
$$\left.\frac{[\mathbb H(p_1),\mathbb H(q_1);\mathbb H(p_2),\mathbb H(q_2)]}{[p_1,q_1;p_2,q_2]}\right|_{\mathbb H(p_1)} =\mathcal I +  \frac16(p_1-q_1)(p_2-q_2)\mathcal S(\mathbb H)(p_1) \quad + o(\epsilon^2)$$
where $\mathcal S(\mathbb H)(p_1)$ is bounded linear endomorphism of the subspace $\mathbb H(p_1)$.
\end{theorem}

\begin{proof}
  Restricted to $\mathbb H(p_1)$, we have
  \begin{align*}[\mathbb H(p_1),\mathbb H(q_1);\mathbb H(p_2),\mathbb H(q_2)] &= \frac14([\mathbb H(p_1)\mathbb H(q_2)]+[\mathbb H(q_1)\mathbb H(p_2)])^2\\&=[\mathbb H(p_1)\mathbb H(q_2)]_+[\mathbb H(q_1)\mathbb H(p_2)]_+.
  \end{align*}
  
  Use a splitting $\mathbb T=\mathbb A\oplus\mathbb B$ in terms of which the operators will be denoted by block matrices, such that
  $$\mathbb H(p) = \op{im}\begin{bmatrix}I\\ H(p)\end{bmatrix} $$
  where $I$ is the identity endomorphism of $\mathbb A$ and $H(p)$ is a linear map $H(p):\mathbb A\to\mathbb B$.  Then
  $$[\mathbb H(p)\mathbb H(q)] = \begin{bmatrix}-(H(p)-H(q))^{-1}(H(p)+H(q)) & 2(H(p)-H(q))^{-1}\\2(H(p)^{-1}-H(q)^{-1})^{-1} & (H(p)+H(q))(H(p)-H(q))^{-1}\end{bmatrix}_{[\mathbb{AB}]}.$$
  So
  \begin{align*}
    [\mathbb H(p)\mathbb H(q)]_+ &= \begin{bmatrix}-(H(p)-H(q))^{-1}H(q) & (H(p)-H(q))^{-1}\\(H(p)^{-1}-H(q)^{-1})^{-1} & H(p)(H(p)-H(q))^{-1}\end{bmatrix}\\
                                 &= \begin{bmatrix}-(H(p)-H(q))^{-1}H(q) & (H(p)-H(q))^{-1}\\-H(p)(H(p)-H(q))^{-1}H(q) & H(p)(H(p)-H(q))^{-1}\end{bmatrix}\\
                                   \end{align*}

  Then
  \begin{align*}
    [\mathbb H(q_1)\mathbb H(p_2)]_+\begin{bmatrix}I\\ H(p_1)\end{bmatrix}&=\begin{bmatrix}(H(q_1)-H(p_2))^{-1}(H(p_1)-H(p_2))\\ H(q_1)(H(q_1)-H(p_2))^{-1}(H(p_1)-H(p_2))\end{bmatrix}\\
    &=\begin{bmatrix}I\\H(q_1)\end{bmatrix}(H(q_1)-H(p_2))^{-1}(H(p_1)-H(p_2))
  \end{align*}
  Next, we have
  $$[\mathbb H(p_1)\mathbb H(q_2)]_+\begin{bmatrix}I\\ H(q_1)\end{bmatrix}=\begin{bmatrix}I\\H(p_1)\end{bmatrix}(H(p_1)-H(q_2))^{-1}(H(q_1)-H(q_2)).$$
  So
  \begin{align*}
    [\mathbb H(p_1)&\mathbb H(q_2)]_+[\mathbb H(p_2)\mathbb H(q_1)]_+\begin{bmatrix}I\\H(p_1)\end{bmatrix} \\&= \begin{bmatrix}I\\H(p_1)\end{bmatrix}(H(p_1)-H(q_2))^{-1}(H(q_1)-H(q_2))(H(q_1)-H(p_2))^{-1}(H(p_1)-H(p_2)).
  \end{align*}
  Summarizing,
  $$\frac{[\mathbb H(p_1),\mathbb H(q_1);\mathbb H(p_2),\mathbb H(q_2)]}{[p_1,q_1;p_2,q_2]}\begin{bmatrix}I\\H(p_1)\end{bmatrix} =  \begin{bmatrix}I\\H(p_1)\end{bmatrix}X$$
  where
  $$ X=\left(\frac{H(p_1)-H(q_2)}{p_1-q_2}\right)^{-1}\left(\frac{H(q_1)-H(q_2)}{q_1-q_2}\right)\left(\frac{H(q_1)-H(p_2)}{q_1-p_2}\right)^{-1}\left(\frac{H(p_1)-H(p_2)}{p_1-p_2}\right).$$
Take $p_1=0$ and expand each of the four factors of $X$ in a series to order $o(\epsilon^2)$.  For the purposes of the calculation, let $h_1=\dot H(p_1), h_2=\ddot H(p_1), h_3=\dddot H(p_1)$.  Then,

  \begin{align*}
    \left(\frac{H(p_1)-H(q_2)}{p_1-q_2}\right)^{-1}
    &= h_1^{-1}-\frac12h_1^{-1}h_2h_1^{-1}q_2\, - \frac16 S_Hh_1^{-1}q_2^2 + o(\epsilon^2)
  \end{align*}
  where
  $$S_H = h_1^{-1}h_3 - \frac32 \left(h_1^{-1}h_2\right)^2.$$

  \begin{align*}
    \frac{H(q_1)-H(q_2)}{q_1-q_2} &= h_1+\frac12h_2(q_1+q_2) + \frac16 (q_1^2+q_1q_2+q_2^2) h_3 + o(\epsilon^2)
  \end{align*}

  \begin{align*}
    \left(\frac{H(q_1)-H(p_2)}{q_1-p_2}\right)^{-1}
    &= h_1^{-1}-\frac12h_1^{-1}h_2h_1^{-1}(q_1+p_2)\, -\\
    &\quad- \frac16 S_Hh_1^{-1}(q_1+p_2)^2 + \frac16h_1^{-1}h_3h_1^{-1}q_1p_2+ o(\epsilon^2)
  \end{align*}

  \begin{align*}
    \frac{H(p_1)-H(p_2)}{p_1-p_2} &= h_1+\frac12h_2p_2 + \frac16h_3p_2^2 + o(\epsilon^2)
  \end{align*}

  Expand the product into homogeneous parts:
  $$X = X_0 + X_1 + X_2 + o(\epsilon^2)$$
  Where:
  \begin{align*}
    X_0 &= h_1^{-1}h_1h_1^{-1}h_1= I\\
    X_1 &= -\frac12 h_1^{-1}h_2q_2 + \frac12h_1^{-1}h_2(q_1+q_2)- \\
        &\qquad - \frac12h_1^{-1}h_2(q_1+p_2) + \frac12h_1^{-1}h_2p_2\\
        &=0
  \end{align*}

  \begin{align*}
    X_2 &= -\frac16S_H(q_2^2+(q_1+p_2)^2) + \frac16(q_1^2+q_1q_2+q_2^2+q_1p_2+p_2^2)h_1^{-1}h_3\\
        &\quad + \frac14 (h_1^{-1}h_2)^2\left((-q_2)(q_1+q_2 - q_1 - p_2+p_2) +(q_1+q_2)(-q_1-p_2 + p_2) - (q_1+p_2)p_2 \right)\\
    &=-\frac16S_H(q_2^2+(q_1+p_2)^2) + (q_1^2+q_1q_2+q_2^2+q_1p_2+p_2^2)\left(\frac16h_1^{-1}h_3-\frac14 (h_1^{-1}h_2)^2\right)\\
        &=-\frac16S_H(q_2^2+(q_1+p_2)^2-q_1^2-q_1q_2-q_2^2-q_1p_2-p_2^2)\\
        &=-\frac16S_Hq_1(p_2-q_2) = \frac16S_H(p_1-q_1)(p_2-q_2),
  \end{align*}
  as required.
\end{proof}

Let $\mathbb H:U\to \mathcal G(\mathbb T)$ be a smooth nondegenerate curve.  Suppose that $\phi:V\to U$ is a bijection.  Then $\mathbb H\circ\phi$ is again a nondegenerate curve.

\begin{proposition}
  If $\mathbb H: U\to\mathcal G(\mathbb T)$ is smooth, codifferentiable, and thrice continuously differentiable, and $\phi:V\to U$ is a smooth diffeomorphism, then
  $$\mathcal S(\mathbb H\circ\phi) = (\phi')^2\,\mathcal S(\mathbb H)\!\circ\!\phi + \mathcal S(\phi)\mathcal I.$$
\end{proposition}

\begin{proof}
  We use the same notation as in the proof of Theorem \ref{GeneralizedSchwarzianTheorem} immediately above.  Then
  $$(H\circ\phi)'=\phi'H'\circ\phi $$
  $$(H\circ\phi)''=(\phi')^2H''\circ\phi + \phi''H'\circ\phi $$
  $$(H\circ\phi)'''=(\phi')^3H'''\circ\phi + 3\phi'\phi''H''\circ\phi + \phi'''H'\circ\phi $$
  So
  $S_{H\circ\phi}=\left((\phi')^2h_1^{-1}h_3 +3\phi''h_1^{-1}h_2 + (\phi')^{-1}\phi'''I\right)-\frac32\left(\phi'h_1^{-1}h_2 + (\phi')^{-1}\phi''I\right)^2, $
  which, when expanded out gives the desired result.
\end{proof}




\begin{definition}
  A curve $\mathbb H:U\to\mathcal G(\mathbb T)$ is called {\em regular} if it is nondegenerate, codifferentiable, and smooth.
\end{definition}

\begin{theorem}\label{VanishingSchwarzian}
  The following conditions for a regular curve $\mathbb H:U\to\mathcal{G}(\mathbb T)$ are equivalent:
  \begin{itemize}
  \item $\mathcal S_{\mathbb H}\equiv 0$;
  \item In a splitting $\mathbb T=\mathbb A\oplus \mathbb B$ such that $\mathbb H(u)=\op{im}\begin{bmatrix}I\\H(u)\end{bmatrix}$,
    $$H(u) = (I+C(u-t))^{-1}(A+B(u-t))$$
    where $A,B:\mathbb A\to\mathbb B$ are morphisms, $C$ is an endomorphism of $\mathbb B$, and $t\in U$.
  \end{itemize}
  
\end{theorem}

A geometrical version of Theorem \ref{VanishingSchwarzian} is given further below, in Theorem \ref{VanishingSchwarzianPrime}.

\begin{proof}

  For the forward implication, we solve the equation for $H(u)$ a real smooth (three times continuously differentiable) function of the real variable $u$, near a point $u = u_0$, such that $\dot H(u_0)$ is an isomorphism:
\[ 2\dddot H - 3\ddot H\dot H^{-1}\ddot H = 0.\]
If $H(u)$ satisfies this equation, then the equation shows that $\dddot H(u)$ is at least once continuously differentiable,  so without loss of generality, we may assume that $H(u)$ is four times continuously differentiable.   Put $G(u) = \dot H(u)^{-1}$, so   $G(u)$ is three times continuously differentiable and an isomorphism.  Then we have:
\[ \ddot H =  \left(G^{-1}\right)'   = - G^{-1} \dot G G^{-1}, \hspace{10pt} \dddot H =  - G^{-1} \ddot G G^{-1} + 2G^{-1} \dot G G^{-1}\dot GG^{-1}\]
\[ 0 = 2\dddot H  -  3\ddot H\dot H^{-1}\ddot H = - 2G^{-1} \ddot G G^{-1} + G^{-1} \dot G G^{-1} \dot G G^{-1}\]
\[ \ddot G =  2^{-1}\dot G G^{-1} \dot G.\]
Differentiating once more, we get:
\[ 4\dddot G = 2\ddot G G^{-1} \dot G +  2\dot GG^{-1} \ddot G  -  2\dot G G^{-1} \dot G G^{-1} \dot G  = 0.\]
Integrating, we get that $G(u)$ is quadratic in the variable $u$.   So  the quantity $G(u)$ is given exactly by its quadratic Taylor series, based at $u = t$, which, using the equation $\ddot G = 2^{-1} \dot GG^{-1} \dot G$ reads as follows:
\[ G(u) =  G(u_0) + \dot G(u_0)(u - u_0) +  4^{-1}\dot G(u_0)G(u_0)^{-1} \dot G(u_0)(u - u_0)^2\]
\[  =  G(u_0)\left(I  +   2^{-1}(u - u_0)\left(G(u_0)\right)^{-1}\dot G(u_0)\right)^2  =  C^{-1} (I + A(u - u_0))^2, \]
\[ C = G(u_0)^{-1} = \dot H(u_0), \hspace{7pt} A = 2^{-1} \left(G(u_0)\right)^{-1}\dot G(t)  = -2^{-1}\ddot H(u_0) \dot H(u_0)^{-1}.\]
In terms of $H(u)$, we have the differential relation:
\[  (I + A(u - u_0))^{2}\dot H(u) =  C = \dot H(u_0). \]

Differentiating both sides with respect to $u$ gives:
\[ (I + A(u -u_0))^2 \ddot H + 2A(I + A(u -u_0))\dot H = 0, \]
\[(I + A(u - u_0)) \ddot H + 2A\dot H = 0,\]
\[ \left((I + A(u - u_0))H\right)'' = 0.\]
Integrating this equation, we get the formula, for constant morphisms $T$ and $R$:
\[ (I + A(u - u_0)) H = T(u - u_0) + R. \]
Differentiating, we get:
\[ AH + (I + A(u - u_0))^{-1} C = T, \]
\[ (I + A(u -u_0))T  =  C + AP(u - u_0) + AR, \]
\[ T = C + AR, \]
\[ H(u) = (I + A(u -u_0))^{-1}((C + AR)(u -u_0) + R)  = R + (u - u_0) (I + A(u -u_0))^{-1}C\]
\[ = R + (u - u_0) (C^{-1} + C^{-1}A(u -u_0))^{-1}\]
\[ = R + (u - u_0) (P + Q(u -u_0))^{-1}, \]
\[ R = H(u_0), \hspace{7pt} P =   C^{-1} = \dot H(u_0)^{-1}, \hspace{7pt} Q = C^{-1} A  = -2^{-1} \dot H(u_0)^{-1}\ddot H(u_0) \dot H(u_0)^{-1}.\]
Note that $P$ is an isomorphism.   Summmarizing, we have the general solution 
\begin{align*}
  G(u) &= 4^{-1}G(u_0)\left(2I + \left(G(u_0)\right)^{-1}\dot G(u_0)(u - u_0)\right)^2 \\
\notag  &= 4^{-1} \left(2G(u_0) + \dot G(u_0)(u - u_0)\right)\left(G(u_0)\right)^{-1}\left(2G(u_0) + \dot G(u_0)(u - u_0)\right),
\end{align*}
\begin{equation}\label{VanishingSchwarzianH}
H(u) = H(u_0) +   2(u - u_0)\dot H(u_0)\left(2\dot H(u_0) -  \ddot H(u_0)(u - u_0)\right)^{-1} \dot H(u_0), 
\end{equation}
\[ G(u) = \dot H(u)^{-1}, \hspace{10pt} \dot G(u) = -  \dot H(u)^{-1}\ddot H(u)\dot H(u)^{-1}.\]

This may be put in the form
$$H(u) = A + \tfrac12(u-u_0)(I+C(u-u_0))^{-1}B = (I+C(u-u_0))^{-1}((I+C(u-u_0))A + B/2)$$
where $A=H(u_0)$, $B=2\dot H(u_0)$, $C=-\ddot H(u_0)\dot H(u_0)^{-1}/2$.

For the converse, if $H(u) = (I+(u-u_0)C)^{-1}(A+(u-u_0)B)$, then
\begin{align*}
  \dot H(u)&= (I+(u-u_0)C)^{-2}(-C)(A+B(u-u_0)) + (I+(u-u_0)C)^{-2}B\\
       &= (I+(u-u_0)C)^{-2}[-C(A+B(u-u_0))+(I+(u-u_0)C)B]\\
       &= (I+(u-u_0)C)^{-2}(B-CA).
\end{align*}
So, with $G(u)=\dot H(u)^{-1}$, it is now straightforward to verify that $2\ddot G=\dot GG^{-1}\dot G$.
\end{proof}

\subsection{A formula for the Schwarzian}
\begin{theorem}\label{SchwarzianFormula}
  Let $\mathbb H:\mathbb U\to\mathcal G(\mathbb T)$ be a regular curve, and let $\mathcal K$ be a reflection $[\mathbb H(u_0)\mathbb A]$ for some $\mathbb A\in\mathcal G(\mathbb T)$ complementary to all $\mathbb H(u)$. Then:
  \begin{itemize}
    \item $\mathbb H_\prime'(u)$ is conjugate to $\mathcal K$.
    \item Writing $\mathbb H(u) = (I\oplus H(u))\mathbb K_+$, we have
      $$\mathbb H_\prime''(u)^2 = \begin{bmatrix}2h_1^{-1}h_3-3h_1^{-1}h_2h_1^{-1}h_2&& h_1^{-1}h_3h_1^{-1}h_2h_1^{-1}-h_1^{-1}h_2h_1^{-1}h_3h_1^{-1}\\&&\\0&&2h_3h_1^{-1}-3h_2h_1^{-1}h_2h_1^{-1}\end{bmatrix}_{\mathcal K}$$
where $h_1=\dot H(u_0),h_2=\ddot H(u_0),h_3=\dddot H(u_0)$ and $H(u_0)=0$.
    \end{itemize}
    
  \end{theorem}
  
  
  \begin{proof}
    Relative to the reflection $\mathcal K$, we have the block form $\mathbb H(u) = \op{im}\begin{bmatrix}I\\ H(u)\end{bmatrix}$ and $H(u_0)=0$.
    Then
    $$\mathbb H_\prime(u) = \begin{bmatrix}I\\ H(u)\end{bmatrix}\dot H(u)^{-1}\begin{bmatrix}H(u) & -I\end{bmatrix}
    =\begin{bmatrix}\dot H(u)^{-1}H(u)& -\dot H(u)^{-1}\\ H(u)\dot H(u)^{-1}H(u) & -H(u)\dot H(u)^{-1}\end{bmatrix}$$
    The derivative is
    \begin{align*}
      \mathbb H_\prime'(u)
      &= \begin{bmatrix}I-\dot H(u)^{-1}\ddot H(u)\dot H(u)^{-1}H(u)& \dot H(u)^{-1}\ddot H(u)\dot H(u)^{-1}\\ 2H(u) - H(u)\dot H(u)^{-1}\ddot H(u)\dot H(u)^{-1}H(u)& H(u)\dot H(u)^{-1}\ddot H(u)\dot H(u)^{-1}-I\end{bmatrix}\\
      &=
        \begin{bmatrix}
          I & \dot H^{-1}\ddot H\\
          H & H\dot H^{-1}\ddot H - 2\dot H
        \end{bmatrix}
        \begin{bmatrix}
          I &0\\
          0 & -I
        \end{bmatrix}
        \begin{bmatrix}
          I-\dot H^{-1}\ddot H(2\dot H)^{-1}H & \dot H^{-1}\ddot H(2\dot H)^{-1}\\
          (2\dot H)^{-1}H & - (2\dot H)^{-1}
        \end{bmatrix}
    \end{align*}
    Therefore, $\mathbb H_\prime'(u)$ is conjugate to $\mathcal K$.

  For the second statement, we put $h_1=\dot H(0), h_2=\ddot H(0),$ and $h_3=\dddot H(0)$.  Then:
  $$\mathbb H_\prime''(0) = \begin{bmatrix}-h_1^{-1}h_2&& h_1^{-1}h_3h_1^{-1}\!\!\!-\!2h_1^{-1}h_2h_1^{-1}h_2h_1^{-1}\\
    &&\\
    2h_1&& h_2h_1^{-1}\end{bmatrix} $$
  so that
  $$
  \mathbb H_\prime''(0)^2 = \begin{bmatrix}2h_1^{-1}h_3-3h_1^{-1}h_2h_1^{-1}h_2& 
&h_1^{-1}h_3h_1^{-1}h_2h_1^{-1}-h_1^{-1}h_2h_1^{-1}h_3h_1^{-1}\\&&\\0&&2h_3h_1^{-1}-3h_2h_1^{-1}h_2h_1^{-1}\end{bmatrix}
$$
as required.
\end{proof}

\subsection{Geometry of the vanishing Schwarzian}

  To conclude the section, we now give a geometrical version of Theorem \ref{VanishingSchwarzian}.  While it is possible to use that theorem to prove this result, we instead here give a direct proof.

\begin{theorem}\label{VanishingSchwarzianPrime}
  The following conditions for a regular curve $\mathbb H:\mathbb U\to\mathcal{G}(\mathbb T)$ are equivalent:
  \begin{itemize}
  \item $\mathcal S_{\mathbb H}\equiv 0$;
  \item There exists a hyperbolic structure $(\mathcal{J,K})$ such that
    $$\mathbb H(u) = L(\mathcal{J,K})(u).$$
  \end{itemize}
Moreover, this hyperbolic structure is unique.
\end{theorem}

\begin{proof}
  We shall work near a base point $u=u_0$, which we may as well take to be zero.  We work in a splitting $\mathbb T=\mathbb H(0)\oplus\mathbb B$, and let $\mathbb H(u)$ as usual be the graph of a symmetric $H(u):\mathbb H(0)\to\mathbb B$, where $H(0)=0$.  
  Theorem \ref{VanishingSchwarzian} shows that for $\mathcal S_{\mathbb H}\equiv 0$, we have
  \[ H(u) = 2uh_1\left(2h_1 -  h_2u\right)^{-1} h_1, \]
  where $h_1=\dot H(0), h_2=\ddot H(0)$.  Note that
  \begin{align*}
    \dot H(u) &= 2h_1\left(2h_1 -  h_2u\right)^{-1} h_1 + 2uh_1\left(2h_1-h_2u\right)^{-1}h_2\left(2h_1-h_2u\right)^{-1}h_1\\
              &= 2h_1\left(2h_1 -  h_2u\right)^{-1}(2h_1)\left(2h_1-h_2u\right)^{-1}h_1              
  \end{align*}
  Then
  \begin{align}
   \notag \mathbb H_\prime(u)
    &= \begin{bmatrix}I\\ 2uh_1\left(2h_1 -  h_2u\right)^{-1} h_1\end{bmatrix}\dot H^{-1}\begin{bmatrix}-2uh_1\left(2h_1 -  h_2u\right)^{-1} h_1 & I\end{bmatrix} \\
   \notag &=4^{-1}\begin{bmatrix}
      -2uh_1^{-1}(2h_1-h_2u) & h_1^{-1}\left(2h_1 -  h_2u\right)h_1^{-1}\left(2h_1-h_2u\right)h_1^{-1}\\
      -4u^2h_1 & 2u(2h_1-h_2u)h_1^{-1}
    \end{bmatrix}\\
   \label{Hdecomp} &= \begin{bmatrix}
      0 & h_1^{-1}\\
      0 & 0
    \end{bmatrix}
    + u \begin{bmatrix}
      -I & -h_1^{-1}h_2h_1^{-1}\\
      0 & I
    \end{bmatrix}
          + u^2 \begin{bmatrix}
            2^{-1}h_1^{-1}h_2 & 4^{-1}h_1^{-1}h_2h_1^{-1}h_2h_1^{-1}\\
            -h_1 & -2^{-1}h_2h_1^{-1}
      \end{bmatrix}
  \end{align}  
  We now make the following observation:
  \begin{lemma}
    Let $\mathbb H(u)=L(\mathcal J,\mathcal K)(u)$.  Then
    $$\mathbb H_\prime(u) = 2^{-1}\mathcal J(\mathcal I+\mathcal K) + u\mathcal K + 2^{-1}u^2\mathcal J(\mathcal I-\mathcal K).$$
  \end{lemma}

  Indeed, decompose $\mathcal J,\mathcal K$ relative to $\mathbb K_\pm$ as
  $$\mathcal K=\begin{bmatrix}I&0\\0&-I\end{bmatrix}, \mathcal J=\begin{bmatrix}0&J\\-J^{-1}&0\end{bmatrix}.$$
  By definition
  $$L(\mathcal J,\mathcal K)(u) = \{ uk + \mathcal Jk | k\in\mathbb K_+\}.$$
  Fixing $u,t$ near zero, we have
  \begin{align*}
    [L(\mathcal J,\mathcal K)(u)L(\mathcal J,\mathcal K)(t)]
    &= (t-u)^{-1}\begin{bmatrix}-(t+u)&-2Jtu\\ 2J^{-1}&t+u\end{bmatrix}\\
    &=(t-u)^{-1}((u+t)\mathcal K + (\mathcal J+\mathcal J\mathcal K) + ut(\mathcal J-\mathcal J\mathcal K).
  \end{align*}
  The lemma now follows by multiplying by $\tfrac12(t-u)$ and taking the limit as $t\to u$.

  Continuing \eqref{Hdecomp} from above, we therefore must have
  \begin{align*}
    \mathcal K &= \begin{bmatrix}
      -I & -h_1^{-1}h_2h_1^{-1}\\
      0 & I
    \end{bmatrix}\\
    \mathcal J &= 
\begin{bmatrix}
      0 & h_1^{-1}\\
      0 & 0
    \end{bmatrix}
          + \begin{bmatrix}
            2^{-1}h_1^{-1}h_2 & 4^{-1}h_1^{-1}h_2h_1^{-1}h_2h_1^{-1}\\
            -h_1 & -2^{-1}h_2h_1^{-1}
          \end{bmatrix}\\
               &= 
                 \begin{bmatrix}
                   2^{-1}h_1^{-1}h_2 & h_1^{-1}+4^{-1}h_1^{-1}h_2h_1^{-1}h_2h_1^{-1}\\
                   -h_1 & -2^{-1}h_2h_1^{-1}
                 \end{bmatrix} 
  \end{align*}
  It is then easily checked that $\mathcal J^2=-\mathcal I$ and $\mathcal K^2=\mathcal I$.  Moreover,
  \begin{align*}
    \mathcal K\mathcal J &=
    \begin{bmatrix}
      2^{-1}h_1^{-1}h_2 & -h_1^{-1}+4^{-1}h_1^{-1}h_2h_1^{-1}h_2h_1^{-1}\\
      -h_1  & -2^{-1}h_2h_1^{-1}
    \end{bmatrix}\\
    \mathcal J\mathcal K &=
    \begin{bmatrix}
      -2^{-1}h_1^{-1}h_2 & h_1^{-1}-4^{-1}h_1^{-1}h_2h_1^{-1}h_2h_1^{-1}\\
      h_1  & 2^{-1}h_2h_1^{-1}
    \end{bmatrix}
  \end{align*}
  So $\mathcal J\mathcal K=-\mathcal K\mathcal J$.  Thus $(\mathcal J,\mathcal K)$ is the required hyperbolic structure.
\end{proof}

\section{Symplectic structures and Lagrangian subspaces}\label{SymplecticLagrangian}
\subsection{Hermitian symplectic spaces}
\begin{definition}\label{SymplecticStructure}
  Let $\mathbb T$ be a real Hilbert space.  An {\em Hermitian symplectic structure} on $\mathbb T$ is a bounded linear operator $\Omega :\mathbb T\to\mathbb T$ such that $\Omega^*=-\Omega $ and $\Omega^2=-I$.  When it is convenient, we shall use the {\em symplectic form}:
  $$\omega(x,y)=\langle x,\Omega y\rangle,\quad x,y\in\mathbb T.$$
  The pair $(\mathbb T,\Omega)$ is called an {\em Hermitian symplectic space}.
\end{definition}

\begin{definition}
  An isomorphism $\mu:\mathbb T\to\mathbb T'$, from one Hermitian symplectic space $(\mathbb T,\Omega)$ to another $(\mathbb T',\Omega')$, is called {\em unitary} if it is an isometry of Hilbert spaces that commutes with the symplectic structures: $\mu\Omega = \Omega'\mu$.  The group of unitary automorphisms of an Hermitian symplectic space $\mathbb T$ is denoted $\mathfrak U(\mathbb T)$.
\end{definition}

Here is an example:
\begin{definition}
  Let $\mathbb X$ be a real Hilbert space.  The {\em standard symplectic space} of $\mathbb X$ is the space $\mathbb X\oplus\mathbb X$, with symplectic form
  $$ \omega(x_1\oplus x_2, y_1\oplus y_2) = \langle x_1,y_2\rangle - \langle x_2,y_1\rangle.$$
\end{definition}
The direct sum inner product is compatible with the symplectic form, so with $\Omega(x\oplus y)=(-y)\oplus x$, we obtain an Hermitian symplectic structure.

An Hermitian symplectic space is thus essentially a complex structure on a real Hilbert space, i.e., a complex Hilbert space.  However, we shall usually think of the symplectic form $\omega$ as being the one of primary interest, with the Hilbert inner product being essentially a gauge.  Thus, from our perspective, it shall not be advantageous to identify Hermitian symplectic spaces with complex Hilbert spaces.
\begin{definition}
A {\em Hilbertable space} $\mathbb T$ is a real vector space with an equivalence class of inner products making it a Hilbert space, where two Hilbert inner products are equivalent if there is an isomorphism of one on the other.
\end{definition}
We shall usually think of the symplectic form $\omega$ as being the primary structure of interest, the inner product only entering to define the topology on $\mathbb T$, and the existence of an inner product is used to define orthogonal complementation but it is not essential that this be canonical (only that it exists).  Thus an ostensibly more general definition would be the following.
\begin{definition}
Let $\mathbb T$ be a real Hilbertable space, and let $e:\mathbb T\to\mathbb T^{**}$ be the isomorphism with the bidual.  An isomorphism $\Omega:\mathbb T\to\mathbb T^*$ is called a {\em symplectic structure} if $\Omega^*e=-\Omega$.  A {\em symplectic space} is the pair $(\mathbb T,\Omega)$.
\end{definition}
We show how this is not actually more general, because the inner product can be selected to accord with Definition \ref{SymplecticStructure}.
Let $\langle-,-\rangle$ be any inner product compatible with the topology on $\mathbb T$, under which it becomes a Hilbert space.  Also, let $t:\mathbb T\to\mathbb T^*$ be isomorphism $t(x) = \langle x,-\rangle$, which satisfies $t^*e=t$.  Then, because $\Omega$ is an isomorphism, the bilinear form $\langle x,y\rangle_\Omega = \langle t^{-1}\Omega x,t^{-1}\Omega y\rangle$ is positive definite, and equivalent to the inner product $\langle-,-\rangle$.  Using $t$ and $e$ to identify $\mathbb T,\mathbb T^*,\mathbb T^{**}$, $\langle x,y\rangle_\Omega = -\langle x,\Omega^2y\rangle$.  So $-\Omega^2$ is a positive self-adjoint operator on $\mathbb T$ that is bounded below.  Thus $-\Omega^2$ has a unique positive self-adjoint square root, $G$, that commutes with $\Omega$.  We now define the modified inner product by
$$\langle x,y\rangle_\Omega = \langle x, -\Omega^2y\rangle = \langle Gx,Gy\rangle.$$
Let $\Omega_G:\mathbb T\to \mathbb T$ be the operator $\Omega_G=G^{-1}\Omega$.  Then we have $\Omega_G^2=-I$, and
$$\langle x,\Omega_Gy\rangle_\Omega = \langle Gx,\Omega y\rangle = - \langle \Omega x, Gy\rangle.$$
Thus $\Omega_G$ satisfies $\Omega_G^2=-I$ and $\Omega_G^*=-\Omega_G$ in the modified (equivalent) inner product $\langle x,y\rangle_G$.  But the symplectic form is the same: $\omega(x,y) = \langle x,\Omega y\rangle = \langle x,\Omega_Gy\rangle_\Omega$.

Henceforth, we shall assume that an inner product is selected so that Definition \ref{SymplecticStructure} obtains.  Throughout this section, $\mathbb T$ is a symplectic space with symplectic structure $\Omega$ and associated symplectic form $\omega$.

\begin{definition} A linear subspace $\mathbb X\subset\mathbb T$ is called isotropic if $\Omega \mathbb X\subset\mathbb X^\perp$, i.e.:
$$\langle x,\Omega x'\rangle = 0\qquad\forall x,x'\in\mathbb X.$$
An isotropic subspace $\mathbb X$ is called Lagrangian if it is not properly included in any isotropic subspace.
\end{definition}
\begin{lemma} Lagrangian subspaces are closed.
\end{lemma}
\begin{proof}
  Let $\mathbb X$ be a Lagrangian subspace.  Fix $x\in\mathbb X$.  Then
  $$\langle x,\Omega  x'\rangle = 0,\qquad \forall x'\in\mathbb X.$$
  Since the left-hand side is continuous in $x'$, we have
  $$\langle x,\Omega  x'\rangle = 0,\qquad \forall x'\in\overline{\mathbb X}.$$
  Now, fix $x'\in\overline{\mathbb X}$.  Since $x$ was arbitrary above, we have
  $$\langle x,\Omega  x'\rangle = 0,\qquad \forall x\in\mathbb X.$$
  Again by continuity,
  $$\langle x,\Omega  x'\rangle = 0,\qquad \forall x\in\overline{\mathbb X}.$$
  A fortiori, $\mathbb X + \mathbb R x'$ is an isotropic subspace.  Since $\mathbb X$ is maximal with this property, it follows that $x'\in\mathbb X$.  Since $x'\in\overline{\mathbb X}$ was arbitrary, $\overline{\mathbb X}\subseteq\mathbb X$.  The opposite inclusion is automatic.
\end{proof}

\begin{lemma}
  If $\mathbb X$ is Lagrangian, then $\mathbb X^\perp = \Omega \mathbb X$.
\end{lemma}
\begin{proof}
  The inclusion $\Omega \mathbb X\subseteq\mathbb X^\perp$ is the definition of isotropic.  For the opposite inclusion, let $y\in\mathbb X^\perp$.  Then
  $$\langle x,y\rangle = 0,\qquad\forall x\in\mathbb X.$$
  Since $\Omega^2=-I$, we have
  \begin{equation}\label{xOmega2y}
    \langle x,\Omega^2y\rangle = 0,\qquad\forall x\in\mathbb X.
  \end{equation}
  Consider the space $\mathbb X' = \mathbb X+\mathbb R\Omega Y$.  We claim that $\mathbb X'$ is isotropic.  Let $x+\lambda \Omega y,x'+\lambda' \Omega y\in\mathbb X'$.  Then
  \begin{align*}
    \langle x+\lambda \Omega y,\Omega (x'+\lambda' \Omega y)\rangle
    &= \langle x,\Omega x'\rangle + \lambda \langle \Omega y, \Omega x'\rangle + \lambda'\langle x, \Omega^2y\rangle + \lambda\lambda'\langle \Omega y, \Omega^2y\rangle.
  \end{align*}
  Every term is zero.  Indeed, $\langle x,\Omega x'\rangle =0$ because $\mathbb X$ is isotropic; $\langle \Omega y,\Omega x'\rangle = -\langle x',\Omega^2y\rangle=0$ by \eqref{xOmega2y}; $\langle x,\Omega^2y\rangle = 0$ by \eqref{xOmega2y}; and $\langle \Omega y, \Omega^2y\rangle = -\langle \Omega y,y\rangle = \langle y,\Omega y\rangle = \langle \Omega y,y\rangle$ (by the identity $\Omega^*=-\Omega $).

  Therefore, because $\mathbb X'$ is isotropic and contains the maximal isotropic space $\mathbb X$, we conclude $\mathbb X'=\mathbb X$, i.e., $\Omega y\in\mathbb X$.  
\end{proof}

\begin{corollary}
  Every Lagrangian subspace of $\mathbb T$ has a Lagrangian complement.
\end{corollary}
\begin{corollary}
  Let $\mathbb A$ be a Lagrangian subspace of $\mathbb T$.  Then the composite mapping $\mathbb A\xrightarrow{\Omega}\mathbb T\to \Omega(\mathbb A)$ is an isomorphism, where the second arrow is the orthogonal projection onto $\mathbb A^\perp=\Omega(\mathbb A)$.
\end{corollary}
\begin{lemma}
  If $\mathbb A$ and $\mathbb B$ are complementary isotropic subspaces of $\mathbb T$, then they are Lagrangian.
\end{lemma}
\begin{proof}
  Let $z\in\mathbb T$ be such that $\mathbb A+\mathbb Rz$ is isotropic.  Decompose $z$ into its $\mathbb A$ and $\mathbb B$ components, we can without loss of generality assume that $z\in\mathbb B$, and $\mathbb A+\mathbb Rz$ is isotropic.  Therefore $z\in (\mathbb A+\mathbb Rz)\cap\mathbb B$ satisfies
  $$\langle a, \Omega z\rangle = 0\quad\forall a\in\mathbb A,\quad \langle b,\Omega z\rangle =0\quad\forall b\in\mathbb B.$$
  Thus $\langle x, \Omega z\rangle =0$ for all $x\in\mathbb T$, so that $\Omega z=0$.  Since $\Omega $ is invertible, $z=0$.  This proves that $\mathbb A$ is Lagrangian.  By symmetry, $\mathbb B$ is also Lagrangian.
\end{proof}

\begin{definition}
  An antiautomorphism of the symplectic space $(\mathbb T,\Omega)$ is a bounded linear operator $R:\mathbb T\to\mathbb T$ such that $R^*\Omega R = -\Omega$.  The set of antiautomorphisms is denoted $\op{Sp}_-(\mathbb T,\omega)$.  A symplectic reflection is an antiautomorphism $R\in\op{Sp}_-(\mathbb T,\omega)$ such that $R^2=\mathcal I$.
\end{definition}
Note that a reflection $R$ is a symplectic reflection if and only if $R^*\Omega R =-\Omega$.

\begin{lemma}\label{LagrangianSymplecticReflection}
  Two complementary subspaces $\mathbb A$ and $\mathbb B$ of $\mathbb T$ are Lagrangian if and only if $[\mathbb{AB}]$ is a symplectic reflection.
\end{lemma}
\begin{proof}
  Assume that $\mathbb A$ and $\mathbb B$ are Lagrangian.  Let $x,y\in\mathbb T$ and decompose them into $\mathbb A$ and $\mathbb B$ parts as $x=x_{\mathbb A}+x_{\mathbb B}$ and $y=y_{\mathbb A}+y_{\mathbb B}$.  Then
  \begin{align*}
    \langle [\mathbb{AB}]x,\Omega [\mathbb{AB}]y\rangle
    &= \langle x_{\mathbb A}-x_{\mathbb B},\Omega (y_{\mathbb A}-y_{\mathbb B})\rangle\\
    &=-\langle x_{\mathbb A},\Omega y_{\mathbb B}\rangle-\langle x_{\mathbb B},\Omega y_{\mathbb A}\rangle\\
    &=-\langle x,\Omega y\rangle.
  \end{align*}
  Therefore $[\mathbb{AB}]^*\Omega [\mathbb{AB}]=-\Omega $.

  Conversely, suppose that $[\mathbb{AB}]$ is a symplectic reflection.  If $a_1,a_2\in\mathbb A$, then
  $$\langle a_1,\Omega a_2\rangle= \langle [\mathbb{AB}]a_1,\Omega a_2\rangle = \langle a_1,[\mathbb{AB}]^*\Omega a_2\rangle = -\langle a_1,\Omega [\mathbb{AB}]a_2\rangle = -\langle a_1,a_2\rangle.$$
  Therefore $\mathbb A$ is isotropic.  By analogous reasoning, $\mathbb B$ is isotropic.  Since complementary isotropic spaces are Lagrangian, this completes the proof.
\end{proof}

The following also appears in \cite{furutani2004fredholm}:
\begin{corollary}\label{LagrangianSymplecticProjection}
  A projection operator $P$ of $\mathbb T$ has Lagrangian image and kernel if and only if
  $$P^*\Omega + \Omega P = \Omega.$$
\end{corollary}
\begin{proof}
  Consider the reflection $R=2P-I$.  Then
  $$R^*\Omega R = 4P^*\Omega P - 2P^*\Omega - 2\Omega P + \Omega.$$
  If $P^*\Omega+\Omega P = \Omega$, then multiplying on the left by $P^*$ and right by $P$ gives $2P^*\Omega P =P^*\Omega P$, so $P^*\Omega P=0$, and so
  $$R^*\Omega R = 4P^*\Omega P - \Omega = -\Omega$$
  so $R$ is a symplectic reflection.
  Conversely, if $R$ is a symplectic reflection, then
  \[4P^*\Omega P - 2P^*\Omega - 2\Omega P + \Omega = -\Omega.\]
  Again multiplying on the left by $P^*$ and right by $P$, we have $P^*\Omega P=0$, and so $P^*\Omega + \Omega P = \Omega$.
\end{proof}

\begin{definition}
  A canonical transformation of $\mathbb T$ is an isomorphism $g:\mathbb T\to\mathbb T$ such that $g^*\Omega g = \Omega$.
\end{definition}
The set of canonical transformations of $\mathbb T$ is a group, denoted $\op{Sp}(\mathbb T,\omega)$.

\begin{lemma}\label{GL(n)Induces}
  Let $\mathcal K = [\mathbb A\mathbb B]$ be a symplectic reflection.  Then, for any $\mathcal A\in\op{GL}(\mathbb A)$, there exists a unique $\mathcal T_{\mathcal A}\in\op{Sp}(\mathbb T,\omega)$ such that $\mathcal T_{\mathcal A}|_{\mathbb A}=\mathcal I_{\mathcal A}$ and $[\mathcal K,\mathcal T_{\mathcal A}]=0$.
\end{lemma}
\begin{proof}
  For $a\in\mathbb A$, let $La\in\mathbb B^*$ be the form given as $La(b) = \omega(a,b)$.  Then $L:\mathbb A\to\mathbb B^*$ is an isomorphism, and so is $L^*:\mathbb B\to\mathbb A^*$.  We claim that there is a unique $\mathcal B\in\op{GL}(\mathbb B)$ such that
  $$\omega(\mathcal Aa, \mathcal Bb) = \omega(a,b) $$
  for all $a\in\mathbb A$ and $b\in\mathbb B$.  Written in terms of the operator $L$, this is
  $$\mathcal B^*L\mathcal A = L^*.$$
  The operators $L$ and $\mathcal A$ are isomorphisms, and so
  $$\mathcal B = L^{-*}\mathcal A^{-*}L^*$$
  is an isomorphism.
\end{proof}

\begin{lemma}\label{UniqueBlockSymplectic}
  Given any Hermitian symplectic space $\mathbb T$ and Lagrangian subspace $\mathbb A$, there exists a unique isometric isomorphism $\alpha:\mathbb A\to\mathbb A^\perp$ such that the symplectic structure is
  $$\Omega(a+a') = -\alpha^*(a') + \alpha(a)$$
  for all $a\in\mathbb A$ and $a'\in\mathbb A^\perp$.
\end{lemma}
  Note that the symplectic form is then
  $$\omega(a+a',b+b') = \langle \alpha a, b\rangle - \langle a',\alpha b\rangle$$
  for $a,b\in\mathbb A$ and $a',b'\in\mathbb A^\perp$.  
  \begin{proof}
    We have shown that $\Omega(\mathbb A)=\mathbb A^\perp$, so the required isomorphism is necessarily $\alpha = \Omega_{\mathbb A}$.  We also have $\Omega(\mathbb A^\perp)=\mathbb A$, and because $\Omega^*=-\Omega$, $\Omega|_{\mathbb A^\perp}=-\alpha^*$.  By construction then, $\Omega(a+a') = -\alpha^*(a') + \alpha(a)$.  Also, $\alpha$ is an isometry, because $\Omega$ is compatible with the inner product and satisfies $\Omega^2=-\mathcal I$.
  \end{proof}

\begin{corollary}\label{UniqueUnitaryStandard}
  Given any Hermitian symplectic space $\mathbb T$ and Lagrangian subspace $\mathbb A$, there is an unique unitary map $U_{\mathbb A}:\mathbb T\to\mathbb A\oplus\mathbb A$ from $\mathbb T$ onto the standard Hermitian symplectic space of $\mathbb A$, which restricts to the identity on $\mathbb A\to\mathbb A\oplus 0$.
\end{corollary}

\begin{definition}
  The {\em orthogonal group} $\op{O}(\mathbb X)$ of a real Hilbert space $\mathbb X$ is the set of isometric isomorphisms of $\mathbb X$ onto itself.
\end{definition}
With this definition, we have the following improved Corollary to Lemma \ref{UniqueBlockSymplectic}:
\begin{corollary}\label{UniqueUnitaryStandardOrthogonal}
  Given any Hermitian symplectic space $\mathbb T$ and Lagrangian subspace $\mathbb A$, for any $X\in\op{O}(\mathbb A)$, there is an unique unitary map $U_{\mathbb A}:\mathbb T\to\mathbb A\oplus\mathbb A$ from $\mathbb T$ onto the standard Hermitian symplectic space of $\mathbb A$, such that $U_{\mathbb A}|_{\mathbb A}=X\oplus 0:\mathbb A\to\mathbb A\oplus 0\subset\mathbb A\oplus\mathbb A^\perp$.  
\end{corollary}

\subsection{The Lagrangian Grassmannian}
\begin{definition}
The {\em Lagrangian Grassmannian}, denoted $\mathcal{LG}(\mathbb T,\omega)\subset \mathcal G_1(\mathbb T)$, is the set of Lagrangian subspaces of $(\mathbb T,\omega)$.
\end{definition}

We proceed to give a strong topology on $\mathcal{LG}(\mathbb T,\omega)$.  For the Hilbert space $\mathbb T$, let $\mathcal B(\mathbb T)$ denote the Banach space of bounded linear operators on $\mathbb T$, equipped with the operator norm topology.  If $\mathbb A\in\mathcal{LG}(\mathbb T,\omega)$ is a Lagrangian subspace, let $\mathcal P_{\mathbb A}=[\mathbb{AA}^\perp]$ be the orthogonal projection onto $\mathbb A$.  Then $\mathbb A\mapsto \mathcal P_{\mathbb A}$ embeds $\mathcal{LG}(\mathbb T,\omega)$ into $\mathcal B(\mathbb T)$.  We give $\mathcal{LG}(\mathbb T,\omega)$ the induced topolgy from this embedding.  We record the basic properties here (see \cite{furutani2004fredholm}):
\begin{lemma}
  $\mathcal{LG}(\mathbb T,\omega)$ is closed in $\mathcal B(\mathbb T)$.  For $\mathbb A\in\mathcal{LG}(\mathbb T,\omega)$, the open set $\mathcal U(\mathbb A)$ is a Banach space, and the patching relations between the basic open sets $\mathcal U(\mathbb A)$ give $\mathcal{LG}(\mathbb T,\omega)$ the structure of an analytic Banach manifold, modeled on the space $\widehat{\mathcal B}(\mathbb A)$ of self-adjoint operators on $\mathbb A$.
\end{lemma}
\begin{proof}
  Closure follows at once from Corollary \ref{LagrangianSymplecticProjection}.  The open set $\mathcal U(\mathbb A)$ is the Banach space of continuous operators $H:\mathbb A\to\mathbb A^\perp$, such that $\Omega H:\mathbb A\to\mathbb A$ is self-adjoint.  Let $G_{\mathbb A}:\widehat{B}(\mathbb A)\to \mathcal U(\mathbb A)$ be this identification map.  Suppose that $\mathbb X$ belongs to $\mathcal U(\mathbb A)\cap\mathcal U(\mathbb B)$, is represented by $G_{\mathbb A}(X)\in\widehat{\mathcal B}(\mathbb A)$ and $G_{\mathbb B}(Y)\in\widehat{\mathcal B}(\mathbb B)$.  We have to show that the transition map is smooth.  First, note that for $x\in\mathbb B$, $-Yx+Jx$ is orthogonal to $G_{\mathbb B}(Y)$.  Hence,
  \[([G_{\mathbb B}(Y)G_{\mathbb B}(Y)^\perp]_+ + [\mathbb{BB}^\perp]_-)(-Yx + Jx) = Jx.\]
  On the other hand, $G_{\mathbb B}(Y)=G_{\mathbb A}(X)$, so that
  \[([G_{\mathbb B}(X)G_{\mathbb B}(X)^\perp]_+ + [\mathbb{BB}^\perp]_-)(-Yx + Jx) = Jx.\]
  Therefore
  $$-Yx + Jx = ([G_{\mathbb B}(X)G_{\mathbb B}(X)^\perp]_+ + [\mathbb{BB}^\perp]_-)^{-1}Jx.$$
  Hence $Y=J - ([G_{\mathbb B}(X)G_{\mathbb B}(X)^\perp]_+ + [\mathbb{BB}^\perp]_-)^{-1}J $, which we claim is analytic in $X$.  We have
  \begin{align*}
    G_{\mathbb A}(X) &= \{y + JXy | y\in\mathbb A\}\\
    G_{\mathbb A}(X)^\perp &= \{-Xy + Jy | y\in\mathbb A\}\\
    [G_{\mathbb A}(X)G_{\mathbb A}(X)^\perp]_+
                     &= \begin{bmatrix}I&-X\\ JX &J\end{bmatrix}\begin{bmatrix}1&0\\ 0&0\end{bmatrix}\begin{bmatrix}I&-X\\ JX &J\end{bmatrix}^{-1}\\
                     &= \begin{bmatrix}I&-X\\ JX &J\end{bmatrix}\begin{bmatrix}1&0\\ 0&0\end{bmatrix}(1+X^2)^{-1}\begin{bmatrix}I&XJ^{-1}\\ -X &J^{-1}\end{bmatrix}
  \end{align*}
  which is analytic.
\end{proof}


\begin{lemma}
  Fix $\mathbb A\in\mathcal{LG}(\mathbb T,\omega)$.  Then the mapping
  \begin{equation}\label{rhoeqn}
    \rho_{\mathbb A} : \mathfrak U(\mathbb T)\to \mathcal{LG}(\mathbb T,\omega),
  \end{equation}
  defined by $\rho_{\mathbb A}(\mu) = mu(\mathbb A)$, is surjective.
\end{lemma}
\begin{proof}
  Let $\mathbb B\in\mathcal{LG}(\mathbb T,\omega)$.  Let $g:\mathbb A\to\mathbb B$ be an isometry of Hilbert spaces.  Define
  $$\mu :\mathbb A\oplus\mathbb A^\perp \to \mathbb B\oplus\mathbb B^\perp$$
  by
  $$\mu(a\oplus a') = (ga)\oplus(-JgJa').$$
  Then $\mu$ is unitary, and clearly sends $\mathbb A$ to $\mathbb B$.
\end{proof}

Following \cite{furutani2004fredholm},
\begin{lemma}
  Fix $\mathbb A\in\mathcal{LG}(\mathbb T,\omega)$.  Then the mapping $\rho_{\mathbb A}:U(\mathbb T)\to\mathcal{LG}(\mathbb T,\omega)$ given by \eqref{rhoeqn} is a principal bundle, with structure group $O(\mathbb A)$.
\end{lemma}
\begin{proof}
  Since $U(\mathbb T)$ acts transitively on $\mathcal{LG}(\mathbb T,\omega)$, it is sufficient to find a local section over the open set $\mathcal U(\mathbb A)$, which we identify with the space $\widehat{\mathcal B}(\mathbb A)$ of bounded self-adjoint operators on $\mathbb A$.  Given an operator $H\in\widehat{\mathcal B}(\mathbb A)$, extend $H$ uniquely to an operator on $\mathbb T$ that commutes with $J$.  Let $\mu_H$ be a square root of the Cayley transform of $H$: 
  $$\mu_H = (I+JH)^{1/2}(I-JH)^{-1/2}.$$
  Then we claim that $H\mapsto U_H$ gives a section of $\rho_{\mathbb A}$, i.e., $\mu_H(\mathbb A) = \{x + JHx | x\in\mathbb A\} = (I+JH)\mathbb A$.  We have
  $$(I+JH) = \mu_H(I+H^2)^{1/2} $$
  so
  $$(I+JH)\mathbb A = \mu_H(I+H^2)^{1/2} = \mu_H\mathbb A,$$
  as required.  The section $H\mapsto \mu_H$ is continuous.
\end{proof}


\subsubsection{Tangent space}

Given $\mathbb Q\in\mathcal{LG}(\mathbb T,\omega)$, there is an exact sequence of topological vector spaces:
\[ 0 \rightarrow \mathbb{Q} \xrightarrow{\iota} \mathbb{T} \xrightarrow{\pi} \mathbb{T}/\mathbb{Q} \rightarrow 0.\]
Here $\iota$ is the natural inclusion and $\pi$ the natural projection.  Now the quotient space $\mathbb{T}/\mathbb{Q}$ is isomorphic to $\mathbb{Q}^*$ via the formula  $\omega(s, t) = \mu(s)(t)$, for any $t \in \mathbb{Q}$, where $\mu(s) \in \mathbb{Q}^*$ and $s \in \mathbb{T}$ represents $s\bmod \mathbb{Q}$   in $\mathbb{T}/\mathbb{Q}$.   Then the map $\mu: \mathbb{T}/\mathbb{Q} \rightarrow \mathbb{Q}^*$ is well-defined and is an isomorphism.  Putting $\rho = \mu\circ \pi$,  we may rewrite the exact sequence as:
\[ 0 \rightarrow \mathbb{Q} \xrightarrow{\iota} \mathbb{T} \xrightarrow{\rho} \mathbb{Q}^*\rightarrow 0.\]
Then the tangent space to the Lagrangian Grassmanian at the point $\mathbb{Q}$ is represented by continuous symmetric endomorphisms of $\mathbb Q$.

\subsection{Pseudo-Hermitian structures}
\begin{definition}
  A pseudo-Hermitian structure on $\mathbb T$ is a canonical transformation $\mathcal J\in\op{Sp}(\mathbb T,\omega)$, whose square is the negative of the identity: $\mathcal J^2=-\mathcal I$.
\end{definition}
The bilinear form $\omega\mathcal J$, defined for $x,y\in\mathbb T$ by $\omega\mathcal J(x,y)=\omega(\mathcal Jx,y)$, is nondegenerate, in the sense that $x\mapsto \omega\mathcal Jx$ is an isomorphism from $\mathbb T$ onto itsself.  

\begin{proposition}
Let $\mathcal J$ be a pseudo-Hermitian structure.  If $\mathbb A\in\mathcal{LG}(\mathbb T,\omega)$, then $\mathcal J(\mathbb A)\in\mathcal{LG}(\mathbb T,\omega)$ is the orthogonal complement of $\mathbb A$ with respect to the form $\omega\mathcal J$.
\end{proposition}

\begin{proof}
  We first claim that $\mathcal J(\mathbb A)$ is isotropic.  We obtain this at once from $\mathcal J\in\op{Sp}(\mathbb T,\omega)$, since For $a_i,a_2\in\mathbb A$, $\omega(\mathcal Ja_1,\mathcal Ja_2)=\omega(a_1,a_2)=0$.  We now claim that $\mathcal J(\mathbb A)$ is maximal.  Suppose that $\mathcal J(\mathbb A)\subset\mathbb X$ where $\mathbb X$ is isotropic.  Then $\mathbb A = \mathcal J^2\mathbb A \subset \mathcal J\mathbb X$, and $\mathcal J\mathbb X$ is isotropic.  Therefore, since $\mathbb A$ is maximal, we have $\mathbb A=\mathcal J\mathbb X$, and so applying $\mathcal J$ to both sides gives $\mathcal J\mathbb A=\mathbb X$.  Thus $\mathcal J(\mathbb A)$ is Lagrangian.
  
  Now, if $a_1,a_2\in\mathbb A$, then $\omega\mathcal J(\mathcal Ja_1,a_2)=-\omega(a_1,a_2)=0$ since $\mathbb A$ is Lagrangian.  Therefore $\mathcal J(\mathbb A)$ is contained in the orthogonal complement of $\mathbb A$.  On the other hand, if $b$ belongs to the orthogonal complement of $\mathbb A$, then for all $a\in\mathbb A$,
  $$0=\omega(\mathcal Jb, a) = \omega(b,\mathcal Ja).$$
  Therefore, $\mathcal J(\mathbb A)+\mathbb Rb$ is isotropic.  Because $\mathcal J(\mathbb A)$ is maximal, we have $b\in\mathcal J(\mathbb A)$.
\end{proof}

\begin{proposition}
  Let $\mathcal J$ be an pseudo-Hermitian structure and $\mathbb{A,B}\in\mathcal{LG}(\mathbb T,\omega)$ complementary Lagrangian subspaces.  Then $\mathbb A$ and $\mathbb B$ are orthogonal complements for $\omega\mathcal J$ if and only if $[\mathbb{AB}]$ anticommutes with $\mathcal J$: 
  $\mathcal J[\mathbb{AB}]+[\mathbb{AB}]\mathcal J=0$.
\end{proposition}

\begin{proof}
  If $\mathbb A$ and $\mathbb B$ are orthogonal complements, then $\mathbb B=\mathcal J(\mathbb A)$.  So $[\mathbb{AB}]\mathcal J|_{\mathbb A} = -\mathcal J|_{\mathbb A}=-\mathcal J[\mathbb{AB}]|_{\mathbb A}$ and $[\mathbb{AB}]\mathcal J|_{\mathbb B}=\mathcal J|_{\mathbb B}=-\mathcal J[\mathbb{AB}]|_{\mathbb B}$.
  
  Conversely, suppose the operators anticommute.  Let $a\in\mathbb A$.  Then $\mathcal Ja = \mathcal J[\mathbb{AB}]a=-[\mathbb{AB}]\mathcal Ja$, so that $\mathcal Ja\in\mathbb B$.  Thus $\mathcal J$ maps $\mathbb A$ into the Lagrangian subspace $\mathbb B$.  But $\mathcal J(\mathbb A)$ is Lagrangian, and thus maximal, so $\mathcal J(\mathbb A)=\mathbb B$, which is the orthogonal complement of $\mathbb A$.
\end{proof}

\begin{theorem}\label{Jinduced}
  Let $\mathbb A,\mathbb B\in\mathcal{LG}(\mathbb T,\omega)$ be complementary Lagrangian subspaces and let $J:\mathbb A\to\mathbb B$ be an invertible symmetric map.  Then, there exists a unique pseudo-Hermitian structure $\mathcal J$ that anticommutes with $[\mathbb{AB}]$, such that $\mathcal J|_{\mathbb A}=S$.  An explicit formula for $\mathcal J$ is
  $$\mathcal J=J[\mathbb{AB}]_+ - J^{-1}[\mathbb{AB}]_-.$$
\end{theorem}

\begin{proof}
  In the statement of the proposition, recall that $J$ being symmetric is equivalent to $\omega(Ja_1,a_2)+\omega(a_1,Ja_2)=0$ for all $a_1,a_2\in\mathbb A$.  We write down such a $\mathcal J$, which has block decomposition in the $\mathbb A\oplus\mathbb B$ splitting
  \begin{equation}\label{JBlock}\mathcal J = \begin{bmatrix}0&-J^{-1}\\ J&0\end{bmatrix},\end{equation}
  which has the form described in the statement.
  
  Conversely, every $\mathcal J$ that exchanges the spaces $\mathbb{A,B}$, and squares to $-\mathcal I$, has the form \eqref{JBlock}.  Since $\mathcal J|_{\mathbb A}$ is symmetric for an pseudo-Hermitian structure $\mathcal J$, so too must be $J$.
\end{proof}

In particular, Theorem \ref{Jinduced} specializes to the following case of interest:
\begin{corollary}\label{JinducedCurve}
Let $\mathbb H:\mathbb U\to\mathcal{LG}(\mathbb T,\omega)$ be a regular curve, and $u_0\in\mathbb U$.  Also, let $\mathbb A\in\mathcal{LG}(\mathbb T,\omega)$ be a fixed but arbitrary Lagrangian subspace that is complementary to $\mathbb H(u)$ for all $u$ in the domain $\mathbb U$.
  There exists a unique pseudo-Hermitian structure $\mathcal J$ on $\mathbb T$ with the following properties:
  \begin{itemize}
  \item $\mathcal J$ anticommutes with $[\mathbb{AH}(u_0)]$: $\mathcal J[\mathbb{AH}(u_0)]+[\mathbb{AH}(u_0)]\mathcal J=0$.
  \item $\mathbb H_\prime(u_0)=[\mathbb{AH}(u_0)]_-\mathcal J[\mathbb{AH}(u_0)]_+$.
  \end{itemize}
\end{corollary}

  The map $\mathcal J$ in Corollary \ref{JinducedCurve} is given explicitly as
  \begin{align*}
    \mathcal J &= [\mathbb{AH}(u_0)]_-\mathbb H_\prime(u_0)[\mathbb{AH}(u_0)]_+ - [\mathbb{AH}(u_0)]_+\mathbb H'(u_0)[\mathbb{AH}(u_0)]_-\\
               &= \mathbb H_\prime(u_0) - [\mathbb{AH}(u_0)]_+\mathbb H'(u_0)[\mathbb{AH}(u_0)]_-
                 \end{align*}
                 where $\mathbb H'(u_0):\mathbb H(u_0)\to\mathbb T/\mathbb H(u_0)$ by definition is the inverse to $\mathbb H_\prime(u_0):\mathbb T/\mathbb H(u_0)\to\mathbb H(u_0)$.

                 It follows straightforwardly from the identities that the projection operators obey (\S\ref{CrossRatioSection}) that $\mathcal J^2=-\mathcal I$ and $\mathbb H_\prime(u_0)=[\mathbb{AH}(u_0)]_-\mathcal J[\mathbb{AH}(u_0)]_+$.

                 Also, $\mathcal J\in\op{Sp}(\mathbb T,\omega)$ follows from identities
                 $$\omega([\mathbb{AH}(u_0)]_\pm\otimes\mathcal I) = \omega(\mathcal I\otimes[\mathbb{AH}(u_0)]_\mp) $$
                 $$\omega(\mathbb H_\prime\otimes\mathcal I) = -\omega(\mathbb H_\prime\otimes\mathcal I).$$

                 To show that $\mathcal J$ anticommutes with $[\mathbb{AH}(u_0)]$, first note that $\mathbb H_\prime(u_0)$ anticommutes with $[\mathbb{AH}(u_0)]$.  In fact,
                 \begin{align*}
                   \mathbb H_\prime(u_0)[\mathbb{AH}(u_0)]&=\mathbb H_\prime(u_0)\\
                   [\mathbb{AH}(u_0)]\mathbb H_\prime(u_0)&=-\mathbb H_\prime(u_0).
                 \end{align*}
                 Likewise, letting $\mathcal K=[\mathbb{AH}(u_0)]_+\mathbb H'(u_0)[\mathbb{AH}(u_0)]_-$, we have
                 \begin{align*}
                   \mathcal K[\mathbb{AH}(u_0)]&=-\mathcal K\\
                   [\mathbb{AH}(u_0)]\mathcal K&=\mathcal K.
                 \end{align*}

\section{Hyperbolic structures on a symplectic vector space}\label{CL11Symplectic}

Let $(\mathbb T,\omega)$ be a symplectic vector space.  A pair $\mathcal{J,K}$ consisting of a pseudo-Hermitian structure $\mathcal J$, and an anticommuting symplectic reflection, defines in a natural way a representation of the algebra $\binom{-1,1}{\mathfrak K}$ of split quaternions on $\mathbb T$.

\begin{definition}
  A hyperbolic structure on the symplectic vector space $(\mathbb T,\omega)$ is a pair $(\mathcal J,\mathcal K)$, with $\mathcal J\in\op{Sp}(\mathbb T,\omega),\mathcal K\in\op{Sp}_-(\mathbb T)$, such that $\mathcal J^2=-\mathcal I$, $\mathcal K^2=\mathcal I$, $\mathcal{JK}+\mathcal{KJ}=0$.
\end{definition}

\begin{proposition}
  Let $\mathcal{J,K}$ be a hyperbolic structure.  Then the stabilizer $\op{Stab}(\mathcal J,\mathcal K)\subset\op{Sp}(\mathbb T,\omega)$ of $\mathcal{J,K}$ in $\op{Sp}(\mathbb T,\omega)$ leaves $\mathbb K_+$ invariant, and the restriction mapping $\op{Stab}(\mathcal J,\mathcal K)\to\op{GL}(\mathbb K_+)$ to the invariant subspace maps one-to-one and onto the pseudo-orthogonal group $\op{O}(\mathbb K_+,\omega\mathcal J)$.  In particular, the group $\op{Stab}(\mathcal J,\mathcal K)$ is canonically isomorphic to $\op{O}(\mathbb K_+,\omega\mathcal J)$.
\end{proposition}
\begin{proof}
  The restriction mapping of the stabilizer onto the invariant subspace $\mathbb K_+$ also leaves the Euclidean structure $\omega\mathcal J$ invariant, so the image is contained in the orthogonal group.  For the converse, \ref{GL(n)Induces} implies {\em a fortiori} that every orthogonal transformation of the Euclidean space $(\mathbb K_+,\omega\mathcal J)$ lifts to a unique symplectic transformation adapted to the sympletic reflection $\mathcal K$.  This lift is then also a unitary transformation, and so preserves $\mathcal J$ as well.
\end{proof}

As a result, at the level of groups, a pseudo-Hermitian structure $\mathcal J$ reduces the automorphism group $\op{Sp}(\mathbb T,\omega)$ to the pseudo-unitary group $\op{U}(\omega\mathcal J)$.  A hyperbolic structure is associated with a reduction down to the pseudo-orthogonal group $\op{O}(\mathbb A,\omega\mathcal J_{\mathbb A})$.  Alternatively, a symplectic reflection $\mathcal K$ reduces the automorphism group down to $\op{GL}(\mathbb K_+)$, and a hyperbolic structure reduces that down to the pseudo-orthogonal group.

\begin{proposition}\label{SpinSp}
  Let $\mathscr J:\binom{-1,1}{\mathfrak K}\to\op{End}(\mathbb T)$ be the algebra homomorphism corresponding to a hyperbolic structure on the symplectic vector space $(\mathbb T,\omega)$.  Then $\mathscr J(\op{Spin}(2,1))\subset\op{Sp}(\mathbb T,\omega)$.
\end{proposition}
\begin{proof}
  It is required to show that $\omega(\mathscr JA\otimes\mathscr JA)=\omega$ for all $A\in\op{Spin}(2,1)$.  Using the isomorphism of $\binom{-1,1}{\mathfrak K}$ with the algebra of $2\times 2$ real matrices, a given $A=\begin{bmatrix}a&b\\c&d\end{bmatrix}\in\binom{-1,1}{\mathfrak K}$ is expressed as
  $$A=\frac{a+d}{2}\mathbf 1+\frac{c-b}{2}\mathbf j+ \frac{a-d}{2}\mathbf k + \frac{b+c}{2}\mathbf{jk}.$$
  Next, to compute $\omega(\mathscr JA\otimes\mathscr JA)$, the key observation is that all cross terms are zero, owing to orthogonality relations
  $$\omega(\mathcal A\otimes \mathcal B)+\omega(\mathcal B\otimes \mathcal A) = 0 $$
  for all distinct $\mathcal{A,B}$ in the set $\{\mathcal{I,J,K,JK}\}$.  This leaves only the diagonal terms to contend with:
  \begin{align*}
    \omega(\mathscr JA\otimes\mathscr JA)
    &=  \left[\left(\frac{a+d}{2}\right)^2+\left(\frac{-b+c}{2}\right)^2-\left(\frac{a-d}{2}\right)^2-\left(\frac{b+c}{2}\right)^2\right]\omega(\mathcal I\otimes\mathcal I)\\
    &=(ad-bc)\omega(\mathcal I\otimes\mathcal I)
  \end{align*}
  as required.
\end{proof}

\subsection{Lift to the symplectic group}
Proposition \ref{SpinSp} implies that, associated to any hyperbolic structure $(\mathcal{J,K})$ on the symplectic space $(\mathbb T,\omega)$, there is a canonical monomorphism from the group $\op{Spin}(2,1)=\op{SL}(2)$ into the symplectic group $\op{Sp}(\mathbb T,\omega)$.  In effect, this singles out a copy of the special linear group within the larger symplectic group.  We interpret this group homomorphism geometrically in terms of the affine geodesics $\mathbb L(\mathcal{J,K})$.

Define a one-parameter group $\tau$ in $\op{SL}(2)$, whose value at $s\in\mathbb R$ is the matrix $\tau_s=\begin{bmatrix}1&s\\0&1\end{bmatrix}\in\op{SL}(2)$.  Acting on $\mathbb R$ by fractional linear transformations, $\tau_s$ is the translation $\tau_s(u)=u+s$ which is a parabolic transformation fixing the point at infinity.

The affine geodesic $\mathbb L(\mathcal J,\mathcal K)$ takes a point $u$ in the affine line $\mathbb R$, to a point $\mathbb L(\mathcal J,\mathcal K)(u)$ of the Lagrangian Grassmannian.  Alternatively, this can be written as $\mathbb L(\op{Ad}(\tau_s)\mathcal J,\op{Ad}(\tau_s)\mathcal K)(u_0)$, where $s=u-u_0$.  So the geodesic $\mathbb L(\mathcal J,\mathcal K)(u)$ for a fixed $\mathcal J,\mathcal K$, can be re-expressed in terms of the value of the family of geodesics $\mathbb L(\op{Ad}(\tau_s)\mathcal J,\op{Ad}(\tau_s)\mathcal K)$ at the fixed (arbitrary) point $u_0\in\mathbb R$.

Let $\mathscr J:\binom{-1,1}{\mathfrak K}\to\op{End}(\mathbb T)$ denote a representation of the algebra, as in \S\ref{Cl11Structures}. The transformation $\tau_s$ is an inner automorphism of $\binom{-1,1}{\mathfrak K}$, so it makes sense to apply $\mathscr J$ to it.  The result $\mathscr J(\tau_s)$ defines a curve in $\op{GL}(\mathbb T)$.  The following observation is immediate:

\begin{proposition}
  The pullback of the Maurer--Cartan form $\Theta=\mathcal M^{-1}d\mathcal M$ in $\op{Sp}(\mathbb T,\omega)$ along $\mathscr J\tau$ is
  \begin{equation}\label{MCpullbacktau}
    (\mathscr J\tau)^*\Theta=\mathscr J\left(\tau_{-s} d\tau_s\right) = \mathscr J\begin{bmatrix}0&ds\\0&0\end{bmatrix}_{\mathcal K}.
  \end{equation}
\end{proposition}

Futhermore, $\tau$ is the unique smooth curve in $\op{Sp}(\mathbb T,\omega)$ such that $\tau_0=\mathcal I$ and the pullback relation \eqref{MCpullbacktau} holds.  We may rewrite \eqref{MCpullbacktau} as four equations:
$$\mathcal K_+\tau^*\Theta\mathcal K_+ = 0,\quad \mathcal K_+(\tau^*\Theta+\mathcal J\,ds)\mathcal K_- = 0,$$
$$\mathcal K_-\tau^*\Theta\mathcal K_- = 0,\quad\mathcal K_-\tau^*\Theta\mathcal K_+ = 0.$$

The last two equations are redundant if we require $\tau$ to be a lift:
\begin{definition}
  A lift of a regular curve $\mathbb H:\mathbb U\to\mathcal{LG}(\mathbb T,\omega)$ to the symplectic group is a smooth curve $\mathcal H:\mathbb U\to\op{Sp}(\mathbb T,\omega)$ such that $\mathcal H(u)^{-1}\mathbb H(u)\in\mathcal{LG}(\mathbb T,\omega)$ is independent of $u$.
\end{definition}

\begin{theorem}\label{UniqueLiftLine}
  Let $(\mathcal J,\mathcal K)$ be a hyperbolic structure on the symplectic vector space $(\mathbb T,\omega)$.  The affine geodesic $\mathbb H(u)=\mathbb L(\mathcal J,\mathcal K)(u)$ admits a unique lift, $\mathcal H:\mathbb R\to\op{Sp}(\mathbb T,\omega)$, such that the following conditions hold:
  \begin{itemize}
  \item $\mathcal H(0)=\mathcal I$,
  \item the pullback $\theta=\mathcal H^*\Theta$ of the Maurer--Cartan form $\Theta$ of $\op{Sp}(\mathbb T,\omega)$ along $\mathcal H$ satisfies
    \begin{equation}\label{Klift}\mathcal K_+\theta\mathcal K_+ = 0,\quad \mathcal K_+(\theta+\mathcal J\,ds)\mathcal K_- = 0
      .\end{equation}
  \end{itemize}
\end{theorem}

\begin{proof}
  For existence, we have already constructed such a lift, namely $\mathcal H(u)=\tau_u$.  The remainder of the proof is for uniqueness.

Also, in block form,
$$\mathcal J = \begin{bmatrix}0&-J^{-1}\\J&0\end{bmatrix}.$$
Here $J:\mathbb A\to\mathbb B$ is symmetric, meaning that $J^*\omega = \omega^*J$.
Define the Hilbert inner product so that $\mathbb A$ and $\mathbb B$ are orthogonal complements.  Write the symplectic form in terms of an $\omega:\mathbb A\to\mathbb B^*$ as
$$\Omega = \begin{bmatrix}0&-\omega^*\\\omega & 0\end{bmatrix}.$$
Consider the condition for a matrix to belong to the symplectic group
$$X=\begin{bmatrix}A&B\\C&D\end{bmatrix}, \quad A:\mathbb A\to\mathbb A, B:\mathbb B\to\mathbb A, C:\mathbb A\to\mathbb B, D:\mathbb B\to\mathbb B.$$
We have
\begin{align*}
X^*\Omega X &=
\begin{bmatrix}
  -A^*\omega^*C + C^*\omega A & -A^*\omega^*D + C^*\omega B\\
  -B^*\omega^*C + D^*\omega A & -B^*\omega^*D + D^*\omega B
\end{bmatrix} = \Omega\\
  0&=-A^*\omega^*C + C^*\omega A \\
  -\omega^*&=-A^*\omega^*D + C^*\omega B\\
  \omega&=-B^*\omega^*C + D^*\omega A\\
  0&=-B^*\omega^*D + D^*\omega B.
\end{align*}

Note that the curve $\mathbb L(\mathcal J,\mathcal K)(u)$ is the image of $\mathbb K_+$ under
$$\begin{bmatrix}u\mathcal I\\ J\end{bmatrix}.$$
Therefore, any lift with $\mathcal H(0)=\mathcal I$ has the form
$$\mathcal H(u) = \begin{bmatrix} a & uA^{-1}\\ c & JA^{-1}\end{bmatrix} $$
where $a,c,A$ are functions of $u$, and $A$ is invertible, with $A(0)=J$.  For
$\mathcal H(u)$ to belong to $\op{Sp}(\mathbb T,\Omega)$, 
\begin{align*}
0&=-a^*\omega^*c + c^*\omega a  \\
  -\omega^*&=-a^*\omega^*JA^{-1} + uc^*\omega A^{-1}\\
  \omega &= -uA^{-*}\omega^*c + (JA^{-1})^*\omega a
\end{align*}
so $\omega^*A=a^*\omega^*J - uc^*\omega$.

We also have
$$\mathcal H(u)^{-1} =
\begin{bmatrix}
  \omega^{-1}A^{-*}J^*\omega && -u\omega^{-1}A^{-*}\omega^*\\
  -\omega^{-*}c^*\omega && \omega^{-*}a^*\omega^*
\end{bmatrix}.
$$
Let $\alpha=A^{-1}\dot A$. The Maurer--Cartan form is
\begin{align*}
  \mathcal H(u)^{-1}\dot{\mathcal H}(u)
  &= \begin{bmatrix}
  \omega^{-1}A^{-*}J^*\omega && -u\omega^{-1}A^{-*}\omega^*\\
  -\omega^{-*}c^*\omega && \omega^{-*}a^*\omega^*
\end{bmatrix}
\begin{bmatrix}
  \dot a & A^{-1}-u\alpha A^{-1}\\
  \dot c & -J\alpha A^{-1}
\end{bmatrix}\\
  &=\begin{bmatrix}
    \omega^{-1}A^{-*}\left(J^*\omega\dot a - u\omega^*\dot c\right) && \omega^{-1}A^{-*}J^*\omega A^{-1}\\
    \omega^{-*}\left(-c^*\omega\dot a + a^*\omega^*\dot c\right) && \omega^{-*}\left(-c^*\omega  + (uc^*\omega-a^*\omega^*J)\alpha \right)A^{-1}
  \end{bmatrix}
\end{align*}
The condition $\mathcal K_+\theta\mathcal K_+=0$ implies that the top left corner is zero, so
$$ J^*\omega\dot a - u\omega^*\dot c = 0.$$
So
$$\dot A = \omega^{-*}(a^*\omega^*J - uc^*\omega)' =  -\omega^{-*}c^*\omega.$$
Therefore,
$$\alpha = A^{-1}\dot A = -A^{-1}\omega^{-*}c^*\omega.$$
Substituing,
\begin{align*}
  \theta &=  \begin{bmatrix}
    0 && \omega^{-1}A^{-*}J^*\omega A^{-1}\\
    \omega^{-*}\left(-c^*\omega\dot a + a^*\omega^*\dot c\right) && \omega^{-*}\left(-c^*\omega -\omega^*A\alpha \right)A^{-1}
  \end{bmatrix}\\
      &= \begin{bmatrix}
    0 && \omega^{-1}A^{-*}J^*\omega A^{-1}\\
    \omega^{-*}\left(-c^*\omega\dot a + a^*\omega^*\dot c\right) && 0
  \end{bmatrix}
\end{align*}
Now, the requirement is that the top right block must be constant.
We compute
\begin{align*}
  \dot A\omega^{-1}J^{-*}A^*\omega &= -\omega^{-*}c^*J^{-*}(J^*\omega a - u \omega^*c)\\
  A\omega^{-1}J^{-*}\dot A^*\omega &= \omega^{-*}(a^*J^*-uc^*)J^{-*}(-\omega^*c)\\
  (A\omega^{-1}J^{-*}A^*\omega)' &= -\omega^{-*}\left(c^*\omega a - uc^*J^{-*}\omega^*c + a^*\omega^*c - uc^*J^*\omega^*c\right)\\
                                   &=2\omega^{-*}AJ^{-1}c = 0.
\end{align*}
Therefore $c=0$.  Hence also $\dot a =0$.  Since $a(0)=I_{\mathbb A}$, we have $a=I_{\mathbb A}$, and $A=J$.

\end{proof}

\subsection{Positive hyperbolic structures}
\begin{definition}
   A hyperbolic structure is positive if the pseudo-Hermitian form $\omega\mathcal J$ is positive-definite, in which case we say that the complex structure $\mathcal J$ is Hermitian.
\end{definition}

\begin{lemma}
  For any positive hyperbolic structure $(\mathcal J,\mathcal K)$ of the symplectic space $\mathbb T$, there exists a unique Hilbert inner product on $\mathbb T$ that is compatible with the symplectic structure, such that $\mathcal J=\Omega$.
\end{lemma}
\begin{proof}
  Uniqueness is clear, because we then have $\langle x,y\rangle = \omega(\mathcal Jx,y)$ for all $x,y$.
  
  For existence, select an arbitrary Hilbert inner product $\langle-,-\rangle$ that is compatible with the symplectic structure.  An operator $\Omega:\mathbb T\to\mathbb T$ is thus defined, suitably dual to the symplectic form.  We have $\omega(\mathcal Jx,\mathcal Jy)=\omega(x,y)$ for all $x,y\in\mathbb T$.  Therefore, $\mathcal J^*\Omega\mathcal J = \Omega$.  The assumption that $\mathcal J$ is Hermitian implies that the isomorphism $J^*\Omega$ is positive definite.  It is also self-adjoint, because $J$ is compatible with $\Omega$.  Therefore, it has a unique positive self-adjoint square root $G$, $G^2=J^*\Omega$, that commutes with $J^*\Omega$.  Next, define
  $$\langle x,y\rangle_G = \langle Gx,Gy\rangle,\quad \Omega_G = G^{-2}\Omega=(J^*\Omega)^{-1}\Omega.$$
  The Hermitian structure $\Omega_G$ for $\langle-,-\rangle_G$ gives the same symplectic form as the Hermitian structure $\Omega$ for $\langle-,-\rangle$:
  $$\langle x,\Omega_Gy\rangle_G = \langle x,G^2\Omega_G y\rangle = \langle x,\Omega y\rangle.$$
  We also have
  $$\Omega_G = (J^*\Omega)^{-1}\Omega = (-\Omega J)^{-1}\Omega = J.$$
\end{proof}

Combining this with Corollary \ref{UniqueUnitaryStandard}, 
\begin{corollary}\label{PositiveHyperbolicStandard}
  Given a positive hyperbolic structure $(\mathcal J,\mathcal K)$ on the symplectic space $\mathbb T$, and let $\mathbb A=\mathbb K_+$.  Then there is a unique canonical transformation $\phi_{\mathcal J,\mathcal K}:\mathbb T\to \mathbb K_+ \oplus\mathbb K_+$, such that:
  \begin{itemize}
  \item $\phi_{\mathcal J,\mathcal K}$ is equal to the identity on $\mathbb K_+$, i.e., $\phi_{\mathcal J,\mathcal K}|_{\mathbb K_+} = \mathcal I_{\mathbb K_+} \oplus 0$; and
  \item $\phi_{\mathcal J,\mathcal K}$ is unitary when $\mathbb T$ is equipped with the Hilbert inner product $\omega\mathcal J$, and $\mathbb K_+\oplus\mathbb K_+$ is equipped with the direct sum inner product, i.e., $(\omega\mathcal J)_{\mathbb K_+\times\mathbb K_+}\oplus (\omega\mathcal J)_{\mathbb K_+\times\mathbb K_+}$.
  \end{itemize}
\end{corollary}

Therefore, a fixed positive hyperbolic structure $(\mathcal J,\mathcal K)$ is precisely the data needed to reduce a symplectic space $\mathbb T$ to a standard Hermitian space, $\mathbb K_+\oplus\mathbb K_+$.  Moreover, modulo the unique unitary isomorphism $\phi_{\mathcal J,\mathcal K}$, we can write, in terms of the block decomposition for the standard Hermitian space $\mathbb K_+\oplus\mathbb K_+$
$$\Omega = \begin{bmatrix}0&-I\\ I&0\end{bmatrix}.$$

Thus, armed with a positive hyperbolic structure, the proof of Theorem \ref{UniqueLiftLine} simplifies considerably, in that after applying the unique canonical transformation $\phi_{\mathcal J,\mathcal K}$, we can take $\omega,J$ as operators on the space $\mathbb K_+$, and $\omega=J=I$.  In the sequel, we shall make free use of this trick.


\section{Positive curves}\label{PositivitySection}

Let $\mathbb H:\mathbb U\to\mathcal{LG}(\mathbb T,\omega)$ be a regular curve.  Suppose that $\mathbb T=\mathbb A\oplus\mathbb B$ is an arbitrary but fixed Lagrangian splitting, such that
$$\mathbb H(u)=\op{im}\begin{bmatrix}I\\H(u)\end{bmatrix} $$
where $H(u):\mathbb A\to\mathbb B$ is a symmetric operator.  By the regularity hypothesis, $\dot H(u)$ is an isomorphism for all $u$.

For an intrinsic characterization, introduce the quadratic form, $\omega\mathbb H_\prime(u)$ on $\mathbb T$, defined by
$$\omega\mathbb H_\prime(u)(x,x) = \lim_{s\to 0}\omega(s[\mathbb H(u)\mathbb H(u+s)]x,x),$$
in agreement with \S\ref{DifferentialCalculus}.

\begin{definition}
  A regular curve $\mathbb H:\mathbb U\to\mathcal{LG}(\mathbb T,\omega)$ is called {\em positive} if the quadratic form $\omega\mathbb H_\prime(u)$ is positive semidefinite for all $u\in\mathbb U$.
\end{definition}

Because $\mathbb H$ is regular, $\mathbb H_\prime$ is an isomorphism $\mathbb T/\mathbb H\to\mathbb H$, and so $\omega\mathbb H_\prime(u)$ is positive-definite on $\mathbb T/\mathbb H$.  Thus:
\begin{proposition}
  A regular curve $\mathbb H$ is positive if and only if $\omega\mathbb H_\prime(u)$ is a positive-definite quadratic form on $\mathbb T/\mathbb H(u)$, for all $u\in\mathbb U$.
\end{proposition}

\subsection{Vanishing Schwarzian}\label{VanishingSchwarzianSection}

\begin{theorem}\label{UniqueLiftSchwarzianZero}
  Suppose that $\mathbb H:\mathbb U\to\mathcal{LG}(\mathbb T,\omega)$ is a positive regular curve.  Then there following are equivalent:
  \begin{enumerate}
  \item the Schwarzian of $\mathbb H$ vanishes identically;
  \item there exists a unique hyperbolic structure $(\mathcal J,\mathcal K)$ such that 
    $$\mathbb H(u) = \mathbb L(\mathcal J,\mathcal K)(u)$$
    for all $u\in\mathbb U$.
  \end{enumerate}
  The unique hyperbolic structure is, moreover, positive definite.
\end{theorem}

When $\mathbb T$ is finite dimensional, a na\"ive dimension count serves to motivate the theorem.  A positive curve with vanishing Schwarzian is determined by initial conditions at an (arbitrary, fixed) point $u_0\in\mathbb U$.  The freedom here is three symmetric tensors in $n$ dimensions, $H(u_0), \dot H(u_0), \ddot H(u_0)$, a $\frac32n(n+1)$ dimensional family.  On the other hand, any two positive hyperbolic structures are conjugate under the symplectic group (of dimension $n(2n+1)$), with $\op{O}(n)$ stabilizer (of dimension $n(n-1)/2$).  So the left and right hand sides of the equation $\mathbb H = \mathbb L(\mathcal J,\mathcal K)$ belong, respectively, to a $\frac32n(n+1)$ dimensional family, and an
$$n(2n+1)-n(n-1)/2=\frac32n(n+1)$$
dimensional family.

\begin{proof}

  We refer to the proof of Theorem \ref{VanishingSchwarzianPrime}.  It is sufficient to show that the operator $\mathcal K=\mathbb H_\prime'(u_0)$ obtained there is a symplectic reflection.  This follows at once from the symmetry of $\mathbb H_\prime$:

  $$\omega(\mathbb H_\prime\otimes\mathcal I) = -\omega(\mathcal I\otimes\mathbb H_\prime).$$

\end{proof}

  Applying Theorem \ref{UniqueLiftLine} twice gives:

\begin{corollary}
  Any two positive regular curves $\mathbb H,\mathbb K$ whose Schwarzians vanish, on the same domain, are related by an element $\mathcal Y$ of the symplectic group: $\mathbb K=\mathcal Y\mathbb H$.
\end{corollary}


\subsection{Fractional-linear curves in the symplectic group}

\begin{theorem}\label{UniqueLiftExplicit}
  Let $(\mathcal J,\mathcal K)$ be a positive hyperbolic structure on the symplectic space $\mathbb T$.  Suppose given a triple $(P,Q,R)$ of symmetric operators
  \begin{align*}
    P :& \mathbb K_-\to\mathbb K_+\\
    Q, R:& \mathbb K_+\to\mathbb K_-
  \end{align*}
  where $P$ is a positive isometry.
  There is associated an element
  $$M = \begin{bmatrix}A&B\\C&D\end{bmatrix}_{\mathcal K} $$
  in the symplectic group of $\mathbb T$, where $D:\mathbb K_-\to\mathbb K_-$ is an isomorphism, such that the curve $H(u)=(A(u-u_0) + B)(C(u-u_0)+D)^{-1}$ has vanishing Schwarzian, where we have the initial conditions
  \[ \left(H(u_0), \hspace{3pt} \dot H(u_0),  \hspace{3pt}\ddot H(u_0)\right) = \left(R,  \hspace{3pt} P^{-1}, \hspace{3pt}  -2P^{-1} QP^{-1}\right), \]
  \[ (P, \hspace{3pt}Q,\hspace{3pt} R) = \left(\dot H(u_0)^{-1}, \hspace{3pt}- 2^{-1} \dot H(u_0)^{-1} \ddot H(u_0)\dot H(u_0)^{-1} , \hspace{3pt}H(u_0)\right).\]
  The group element $M$ is unique, up to right multiplication by an element $W$ of the orthogonal group of $\mathbb K_-$, and we have 
  \[ (A,\hspace{3pt}  B, \hspace{3pt} C, \hspace{3pt} D)  =  \left((I  +RQ)Z^{-1},  \hspace{4pt}  RZ,   \hspace{4pt} QZ^{-1}, \hspace{4 pt} Z\right)W\]
  where $Z$ is the unique positive self-adjoint square root of $P$.  Conversely, every $C^3$ curve $H(u)$ such that $\dot H(u_0):\mathbb K_+\to\mathbb K_-$ is a positive self-adjoint isomorphism can be put in this fractional linear form, for $u$ in an open interval containing the point $u_0$.
\end{theorem}


\begin{proof}
  Applying Corollary \ref{PositiveHyperbolicStandard}, we can suppose without loss of generality that $(P,Q,R)$ are self-adjoint operators of the same Hilbert space $\mathbb A=\mathbb K_+$, and that the Hermitian structure is the standard one on $\mathbb A\oplus\mathbb A$.
We thus solve the equation for self-adjoint $H(u)$ a smooth (three times continuously differentiable) of the real variable $u$, near a point $u = t$, such that $\dot H(u_0)$ is an isomorphism:
\[ 2\dddot H - 3\ddot H\dot H^{-1}\ddot H = 0.\]
If $H(u)$ satisfies this equation, then the equation shows that $\dddot H(u)$ is at least once continuously differentiable,  so without loss of generality, we may assume that $H(u)$ is four times continuously differentiable.   Put $G(u) = \dot H(u)^{-1}$, so   $G(u)$ is a self-ajoint positive isomorphism, three times continuously differentiable in $u$.  Then we have:
\[ \ddot H =  \left(G^{-1}\right)'   = - G^{-1} \dot G G^{-1}, \hspace{10pt} \dddot H =  - G^{-1} \ddot G G^{-1} + 2G^{-1} \dot G G^{-1}\dot GG^{-1}\]
\[ 0 = 2\dddot H  -  3\ddot H\dot H^{-1}\ddot H = - 2G^{-1} \ddot G G^{-1} + G^{-1} \dot G G^{-1} \dot G G^{-1}\]
\[ \ddot G =  2^{-1}\dot G G^{-1} \dot G.\]
Differentiating once more, we get:
\[ 4\dddot G = 2\ddot G G^{-1} \dot G +  2\dot GG^{-1} \ddot G  -  2\dot G G^{-1} \dot G G^{-1} \dot G  = 0.\]
Integrating, we get that $G(u)$ is quadratic in the variable $u$.   So  the quantity $G(u)$ is given exactly by its quadratic Taylor series, based at $u = t$, which, using the equation $\ddot G = 2^{-1} \dot GG^{-1} \dot G$ reads as follows:
\[ G(u) =  G(u_0) + \dot G(u_0)(u - u_0) +  4^{-1}\dot G(u_0)G(u_0)^{-1} \dot G(u_0)(u - u_0)^2\]
\[  =  G(u_0)\left(I  +   2^{-1}(u - u_0)\left(G(u_0)\right)^{-1}\dot G(u_0)\right)^2  =  C^{-1} (I + A(u - u_0))^2, \]
\[ C = G(u_0)^{-1} = \dot H(u_0), \hspace{7pt} A = 2^{-1} \left(G(u_0)\right)^{-1}\dot G(u_0)  = -2^{-1}\ddot H(u_0) \dot H(u_0)^{-1}.\]
In particular, $G(u)$  is real analytic; hence  $H(u)$ is real analytic also.  In terms of $H(u)$, we have the differential relation:
\[  (I + A(u - u_0))^{2}\dot H(u) =  C = \dot H(u_0). \]

Differentiating both sides with respect to $u$ gives:
\[ (I + A(u -u_0))^2 \ddot H + 2A(I + A(u -u_0))\dot H = 0, \]
\[(I + A(u - u_0)) \ddot H + 2A\dot H = 0,\]
\[ \left((I + A(u - u_0))H\right)'' = 0.\]
Integrating this equation, we get the formula, for constant operators $T$ and $R$:
\[ (I + A(u - u_0)) H = T(u - u_0) + R. \]
Differentiating, we get:
\[ AH + (I + A(u - u_0))^{-1} C = T, \]
\[ (I + A(u -u_0))T  =  C + AP(u - u_0) + AR, \]
\[ T = C + AR, \]
\begin{align*}
  H(u) &= (I + A(u -u_0))^{-1}((C + AR)(u -u_0) + R)  \\
  &= R + (u - u_0) (I + A(u -u_0))^{-1}C\\
 &= R + (u - u_0) (C^{-1} + C^{-1}A(u -u_0))^{-1}\\
 &= R + (u - u_0) (P + Q(u -u_0))^{-1},
\end{align*}
\[ R = H(u_0), \quad P =   C^{-1} = \dot H(u_0)^{-1},\]
\[Q = C^{-1} A  = -2^{-1} \dot H(u_0)^{-1}\ddot H(u_0) \dot H(u_0)^{-1}.\]
Note that $P, Q$ and $R$ are all self-adjoint and $P$ is an isomorphism.   Summmarizing, we have the general solution for the self-adjoint $H(u)$ and for the self-adjoint isomorphism $G(u)$:
\begin{align}\label{Gequation}
  G(u) &= 4^{-1}G(u_0)\left(2I + \left(G(u_0)\right)^{-1}\dot G(u_0)(u - u_0)\right)^2 \\
  \notag &= 4^{-1} \left(2G(u_0) + \dot G(u_0)(u - u_0)\right)\left(G(u_0)\right)^{-1}\left(2G(u_0) + \dot G(u_0)(u - u_0)\right),
\end{align}
\[ H(u) = H(u_0) +   2(u - u_0)\dot H(u_0)\left(2\dot H(u_0) -  \ddot H(u_0)(u - u_0)\right)^{-1} \dot H(u_0), \]
\[ G(u) = \dot H(u)^{-1}, \hspace{10pt} \dot G(u) = -  \dot H(u)^{-1}\ddot H(u)\dot H(u)^{-1}.\]

The solution for $G(u)$ is determined by arbitrary initial conditions $G(u_0)$ and $\dot G(u_0)$ such that $G(u_0)$ is required to be an isomorphism.  Similarly, the solution $H(u)$ is given in terms of the initial data $(H(u_0),  \dot H(u_0), \ddot H(u_0))$.   Here the three self-adjoint  operators $(H(u_0),  \dot H(u_0), \ddot H(u_0))$ are arbitrary, except that $\dot H(u_0)$ is required to be an isomorphism.  Also the solutions are valid for all $u$ in the open interval $\mathbb{W}$, containing $u = t$, such that the self-adjoint operator $2\dot H(u_0) -  \ddot H(u_0)(u - u_0)$ is an isomorphism; equivalently $\mathbb{W}$ is the open interval,  containing $u = t$, on which $2G(u_0) + \dot G(u_0)(u - u_0)$ is an isomorphism.

  As noted (Corollary \ref{PositiveHyperbolicStandard}) the symplectic form of $\mathbb T$ is
$$\Omega = \begin{bmatrix}0&-I\\I&0\end{bmatrix}.  $$
  Elements of the symplectic group may be written in this block matrix form as
      \[ M = \hspace{3pt} \begin{bmatrix}A&B\\C&D\end{bmatrix}\hspace{3pt}. \]
Here $M$ is required to preserve the operator $\Omega$, so we need the relation $M^* \Omega M = \Omega$:
\[ M^*\Omega M = M^*  \begin{bmatrix}0&-I\\I&0\end{bmatrix} M =  \begin{bmatrix}A^*&C^*\\B^*&D^*\end{bmatrix}   \begin{bmatrix}0&-I\\I&0\end{bmatrix} \begin{bmatrix}A&B\\C&D\end{bmatrix} \]
\[  =  \begin{bmatrix}C^*&-A^*\\D^*&- B^*\end{bmatrix}    \begin{bmatrix}A&B\\C&D\end{bmatrix} = \Omega =     \begin{bmatrix}0&-I\\I&0\end{bmatrix}, \]
\[ 0 = C^*A - A^*C, \hspace{10pt} 0 = B^*D - D^*B, \hspace{10pt} I =  D^*A -  B^*C.\]

Now when the Schwarzian of $H(u)$ vanishes, we have $H(u) = R + (u - u_0)(P + Q(u - u_0))^{-1}$, where $P$, $Q$ and $R$ are symmetric and $P$ is positive definite.   We wish to write this in fractional linear form, where the coefficients form an element of the symplectic group. 

So we require:
\[ H(u) = (A(u - u_0) + B)(C(u - u_0) + D)^{-1} = R + (u - u_0)(P + Q(u - u_0))^{-1}. \]
Note that putting $u = t$, we need $D$ to be an isomorphism and $H(u_0) = BD^{-1} = R$.  Differentiating the formula for $H(u)$ twice, we get:
\begin{align*}
  \dot H(u) &= A(C(u - u_0) + D)^{-1} -  (A(u - u_0) + B)(C(u - u_0) + D)^{-1}C(C(u - u_0) + D)^{-1}\\
        &= (P + Q(u - u_0))^{-1}  - (u - u_0)(P + Q(u - u_0))^{-1}Q(P + Q(u - u_0))^{-1},
\end{align*}
\begin{align*}
  \ddot H(u) &= - 2A(C(u - u_0) + D)^{-1}C(C(u - u_0) + D)^{-1} +\\&\qquad + 2 (A(u - u_0) + B)(C(u - u_0) + D)^{-1}C(C(u - u_0) + D)^{-1}C(C(u - u_0) + D)^{-1}\\
         &=  - 2(P + Q(u - u_0))^{-1}Q(P + Q(u - u_0))^{-1}  +\\&\qquad+2 (u - u_0)(P + Q(u - u_0))^{-1}Q(P + Q(u - u_0))^{-1}Q(P + Q(u - u_0))^{-1}.
\end{align*}
Now we put $t = u$ in the equations for $H(u), \dot H(u)$ and $\ddot H(u)$, giving:
\begin{align*}
  H(u_0) &= R = BD^{-1}, \\
 \dot H(u_0)  &= P^{-1} = (A - BD^{-1}C)D^{-1} = (A- RC)D^{-1}, \\
  -2^{-1}\ddot H(u_0) &=  P^{-1}QP^{-1} =  (A - BD^{-1} C)D^{-1}CD^{-1} \\
         &=  (A - RC)D^{-1}CD^{-1} = P^{-1} CD^{-1},
\end{align*}
\[B = RD,   \qquad  A - RC = P^{-1} D, \qquad  C = QP^{-1} D, \]
\[ A = P^{-1} D + RC = (I  +RQ)P^{-1} D, \]
\[ (A,\quad B, \quad C,\quad  D) = \left( (I  +RQ)P^{-1}, \quad  R, \quad  QP^{-1},\quad  I\right) D.\]
Here we want the matrix $\displaystyle{M = \begin{bmatrix}A&B\\C&D\end{bmatrix}}$ to be in the symplectic group, so  first we need   $A^*C = X$ and $D^*B = Y$, with $X$ and $Y$ self-adjoint and $D^*A - B^*C = I$.  
Now we have:
\[ Y = D^*B = D^*RD.\]
So $Y$ is automatically self-adjoint, since $R $ is.  Next we  have:
\[ X = A^*C = D^* P^{-1}(I + QR)QP^{-1} D = D^*P^{-1}(Q + QRQ)P^{-1} D.\]
So $X$ is automatically self-adjoint, since $P$ is and since $Q$, $R$ and therefore $QRQ$ are self-adjoint.  Finally, we need:
\[ I = D^*A - B^*C  = D^*\left( (I  +RQ)P^{-1} - RQP^{-1}\right) D  = D^* P^{-1} D,  \]
\[ P = D D^*.\]

Since $P$ is given to be positive and self-adjoint, we can write $P = Z^2$, where $Z$ is a self-adjoint positive isomorphism.   Here $Z$ exists and is unique, given $P$.   Then the general solution for $D$ is  $D = ZW$, where $W^* W = I$, so $W$ belongs to $\op{O}(\mathbb A)$.   Note that $D$ is indeed an isomorphism, as required, since both $Z$ and $W$ are. Summarizing,  we have:
\[(A,\hspace{3pt}  B, \hspace{3pt} C, \hspace{3pt} D)  =  \left((I  +RQ)Z^{-1},  \hspace{4pt}  RZ,   \hspace{4pt} QZ^{-1}, \hspace{4 pt} Z\right)W\]
\[Z^2 = P, \hspace{10pt} W^*W = WW^* = I.\]
Then the transformation matrix $M$ becomes:
\[M =  NS, \hspace{10pt}  N =  \hspace{3pt} \begin{bmatrix} I + RQ&RP\\Q&P\end{bmatrix} \hspace{3pt},  \hspace{10pt} S = Z^{-1} W,\]\[ W^*W = I, \hspace{10pt} Z^2 = P, \hspace{10pt}  SS^* = P^{-1}.\]
Note that we have:
\begin{align*}
 N^* \hspace{5pt} \begin{bmatrix}0&-I\\I&0\end{bmatrix}\hspace{5pt} N   &=  \hspace{3pt} \begin{bmatrix} I + QR&Q\\PR&P\end{bmatrix} \hspace{5pt} \begin{bmatrix}0&-I\\I&0\end{bmatrix}\hspace{5pt} \begin{bmatrix} I + RQ&RP\\Q&P\end{bmatrix} \\
 &=   \hspace{3pt} \begin{bmatrix} Q&-I - QR\\P&-PR\end{bmatrix} \hspace{5pt} \begin{bmatrix} I + RQ&RP\\Q&P\end{bmatrix} \hspace{3pt} \\ &=  \hspace{3pt} \begin{bmatrix} 0&-P\\P&0\end{bmatrix}\hspace{3pt}.
\end{align*}
So the matrix $NS$ is symplectic, provided $S$ is an isomorphism and we have:
\[  S^*\hspace{3pt} \begin{bmatrix} 0&-P\\P&0\end{bmatrix}\hspace{3pt} S \hspace{3pt} = \hspace{3pt}  \begin{bmatrix} 0&-I\\I&0\end{bmatrix} \hspace{3pt}. \]
With $P = Z^2$ and $Z$ is a positive self-adjoint isomorphism , this gives the same solution as above:
\[ S^* P S = I,    \hspace{7pt} SS^* P = I,  \hspace{7pt} P^{-1} =   SS^*,   \hspace{7pt}S = Z^{-1} W, \hspace{7pt} W^*W = WW^* = I.\]

\end{proof}

\subsection{The osculating geodesic}\label{OsculatingLineSection}
\begin{proposition}
  Let $\mathbb H:\mathbb U\to\mathcal{LG}(\mathbb T,\omega)$ be a regular curve, and $u_0\in\mathbb U$.  There is a unique regular curve $\mathbb K:\mathbb U\to\mathcal{LG}(\mathbb T,\omega)$, whose Schwarzian vanishes, such that $\mathbb H$ and $\mathbb K$ are equal in a second-order neighborhood of $u_0$.
\end{proposition}

\begin{proof}
  When $\mathbb T$ is finite-dimensional, this follows by standard existence and uniqueness of global solutions of ordinary differential equations with prescribed initial conditions, taking values in the compact manifold $\mathcal{LG}(\mathbb T,\omega)$. In general, Theorem \ref{VanishingSchwarzian} (equation \eqref{VanishingSchwarzianH}) implies that a regular curve with vanishing Schwarzian is uniquely determined by $H(u_0), \dot H(u_0),\ddot H(u_0)$, in terms of any Lagrangian splitting $\mathbb T=\mathbb A\oplus\mathbb B$ such that $\mathbb H(u)\in\mathcal U(\mathbb A,\mathbb B)$.
  
\end{proof}

\begin{definition}
  The {\em osculating geodesic} to the regular curve $\mathbb H:\mathbb U\to\mathcal{LG}(\mathbb T,\omega)$ at the point $u_0\in\mathbb U$ is the unique curve, denoted $\mathbb H_{u_0}:\mathbb U\to\mathcal{LG}(\mathbb T,\omega)$, whose Schwarzian vanishes identically, and such that $\mathbb H_{u_0}(u)=\mathbb H(u) + O(u-u_0)^3$.
\end{definition}

The osculating geodesic to a positive curve is associated with a unique hyperbolic structure, by Theorem \ref{UniqueLiftLine}.  

\section{Lift to the symplectic group}\label{LiftSection}

\begin{theorem}\label{LiftToSymplectic}
  Let $\mathbb H$ be a positive regular curve in $\mathcal{LG}(\mathbb T,\omega)$.  Then, relative to any fixed positive hyperbolic structure $(\mathcal J,\mathcal K)$ on $\mathbb T$ such that $\mathbb H(u)$ is complementary to $\mathbb K_+$ and $\mathbb K_-$, there exists a canonical lift of $\mathbb H(u)$ to the symplectic group, given explicitly by the block matrix fomula:
  \[ M = \begin{bmatrix} \dot Y& Y\\\dot X&X\end{bmatrix}\]
  \[Y = HX, \qquad  \dot X = - 2^{-1}\dot H^{-1} \ddot H X,  \qquad  XX^* = \dot H^{-1}, \]
The freedom here is $M \rightarrow M\begin{bmatrix}R&0\\0&R\end{bmatrix}$, where $R\in \op{O}(\mathbb K_+,\omega\mathcal J)$.   The Maurer--Cartan form of the symplectic group, restricted to the lifted curve has the form $M^{-1}(u) d(M(u)) = \Theta(u) du$, where we have:
\[ \Theta = M^{-1}\dot M = \begin{bmatrix} 0&I\\ -2^{-1}\mathcal{S}&0\end{bmatrix},\]
\[\mathcal{S}  =  \mathcal S_{\mathcal{J,K},\mathbb H}=2^{-1} X^*\left(2\dddot H - 3\ddot H \dot H^{-1}\ddot H\right)X.\]
The symmetric operator $\mathcal{S}_{\mathcal{J,K},\mathbb H}(u)$ is, by definition, the Lagrange--Schwarzian of $\mathbb H(u)$, relative to the gauge $(\mathcal J,\mathcal K)$.
\end{theorem}

\begin{proof}
Let $Z(u)$ be positive definite and symmetric, such that $\left(Z(u)\right)^2 = \dot H(u)^{-1}$.   Then, following Theorem \ref{UniqueLiftExplicit}, the osculating geodesic to $\mathbb H$ at $u$ defines a natural lift of $\mathbb H$ to a curve in the symplectic group, $M(u)$, given by:
\[ M(u) = \hspace{3pt}\begin{bmatrix} I  +H(u)Q(u) & H(u)P(u)\\Q(u)&P(u)\end{bmatrix}\hspace{4pt}\left(Z(u)\right)^{-1} \hspace{4pt}, \]
\[  P(u) = Z(u)^2 = \dot H(u)^{-1}, \hspace{10pt} Q(u) = -2^{-1} \dot H(u)^{-1}\ddot H(u) \dot H(u)^{-1}.\]
Then we have:
\[ M(u) =  \hspace{3pt}\begin{bmatrix} \dot H(u)  -2^{-1} H(u)\dot H(u)^{-1}\ddot H(u)& H(u)\\ -2^{-1} \dot H(u)^{-1}\ddot H(u)&I\end{bmatrix}\hspace{4pt}Z(u), \]
\[ M^{-1}(u) =   Z(u)\hspace{3pt}\begin{bmatrix}0&I\\-I&0\end{bmatrix} \hspace{3pt} \begin{bmatrix} \dot H(u)  -2^{-1} \ddot H(u)\dot H(u)^{-1}H(u)& -2^{-1}\ddot H(u) \dot H(u)^{-1}\\H(u)&I\end{bmatrix}\hspace{3pt}\begin{bmatrix}0&-I\\I&0\end{bmatrix} \hspace{3pt}\]
\[ =  Z(u)\hspace{3pt} \begin{bmatrix}H(u)&I\\-\dot H(u)  +2^{-1} \ddot H(u)\dot H(u)^{-1}H(u)& 2^{-1}\ddot H(u) \dot H(u)^{-1}\end{bmatrix} \hspace{3pt}\begin{bmatrix}0&-I\\I&0\end{bmatrix} \hspace{3pt}\]
\[ = Z(u) \hspace{3pt} \begin{bmatrix} I& -H(u)\\2^{-1}\ddot H(u) \dot H(u)^{-1}&\dot H(u)  -2^{-1} \ddot H(u)\dot H(u)^{-1}H(u)\end{bmatrix} \hspace{3pt}.\]
Then, for the Maurer--Cartan  one-form $\Theta(u) du =  M^{-1}(u) dM(u)$, restricted to the curve, we have the formula:
\[ (Z(u))^{-1}\left(\Theta(u) - \left(Z(u)\right)^{-1} \dot Z(u)\right)(Z(u))^{-1}  =
  2^{-1} \hspace{3pt} \begin{bmatrix}\ddot H(u) &2\dot H(u)\\Y(u)&\ddot H(u)\end{bmatrix}\hspace{4pt}.\]

Here, for the lower left block, $Y(u)$, we have the formula:
\[\hspace{-20pt} Y  
  =  \frac{3}{2}\ddot H(u) \dot H(u)^{-1}\ddot H(u)-    \dddot H(u). \]
So we get:

\noindent$ \Theta = 2^{-1} Z(u)\begin{bmatrix}\ddot H(u) + 2Z^{-2}(u)\dot Z(u)Z^{-1}(u) &2\dot H(u)\\- \left(\dddot H(u) - \frac{3}{2}\ddot H(u) \dot H(u)^{-1}\ddot H(u)\right)&\ddot H(u) + 2Z^{-2}(u)\dot Z(u)Z^{-1}(u)\end{bmatrix}Z(u)$
Note that we have $Z^2\dot H = I$, so we have
\[ \ddot H  + 2 Z^{-2}\dot ZZ^{-1} =  \left(Z^{-2}\right)'  + 2 Z^{-2}\dot ZZ^{-1}  = - Z^{-2}(Z\dot Z + \dot ZZ)Z^{-2}   + 2 Z^{-2}\dot ZZ^{-1} \]
\[ =  Z^{-2}(\dot ZZ - Z\dot Z)Z^{-2}.\]
So our final expression for the Maurer--Cartan form $\Theta$  restricted to the lifted curve $M(u)$ is:
\[ \Theta = \hspace{3pt} \hspace{3pt} \begin{bmatrix}2^{-1}(Z^{-1}\dot Z - \dot ZZ^{-1}) &I\\- 2^{-1} Z\left(\dddot H - \frac{3}{2}\ddot H \dot H^{-1}\ddot H\right)Z&2^{-1}(Z^{-1}\dot Z - \dot ZZ^{-1})\end{bmatrix}\hspace{4pt}, \hspace{10pt} \dot H = Z^{-2}.\]
%

Note that if we replace $M(u)$ by $M(u)W(u)$, where $W^*(u) W(u) = I$,   the Maurer--Cartan form transforms to $\Phi du$, where we have:
\[ \Phi = W^{-1} \left(\Theta  + \mathcal{I}\dot WW^{-1}\right)W.\]
Here $\mathcal{I}$ is the block identity matrix: $\mathcal{I} = \hspace{3pt} \begin{bmatrix}I&0\\0&I\end{bmatrix}\hspace{3pt}$.

In particular the upper left block and lower right block of $W\Phi W^{-1}$ each become  $\left(\dot W + 2^{-1}(Z^{-1}\dot Z - \dot ZZ^{-1})W\right)W^{-1}$.  Now we have  that the matrix $ 2^{-1}(Z^{-1}\dot Z - \dot ZZ^{-1})$ is automatically skew.  So, when  we integrate the linear differential equation $\dot W = 2^{-1}\left(\dot ZZ^{-1} - Z^{-1} \dot Z\right)W$, we reduce each of these blocks to zero and the solution $W(u)$ automatically obeys the equation $W^*(u)W(u) = I$, provided its initial condition does so.  Further the resulting solution is unique up to an orthogonal transformation $W(u) \rightarrow W(u)R$, where $R^* R = I$, and $R$ is independent of $u$.  We have proved that the curve $\mathbb H(u)$ lifts uniquely to the the symplectic group, modulo an orthogonal transformation of $\mathbb K_+$, such that the Maurer--Cartan form,  $\Theta du$, on the lifted curve  is given by:
\begin{align*}
 \Theta &= \begin{bmatrix}0&I\\-2^{-1}\mathcal{S}(u)&0\end{bmatrix}, \\
 \mathcal{S}(u) &= W^{-1}Z\left(\dddot H - \frac{3}{2}\ddot H \dot H^{-1}\ddot H\right)ZW, \\
 \dot W &= 2^{-1}\left(\dot ZZ^{-1} - Z^{-1}\dot Z\right)W, \\
 Z^2 \dot H &= I.\\
 (ZW)' &= 2^{-1}\left(Z\dot Z + \dot ZZ\right)Z^{-1}W  = 2^{-1}(Z^2)'Z^{-1} W \\
 &= - 2^{-1}\dot H^{-1} \ddot H (ZW).
\end{align*}
Put $X = ZW$, so we have:
\begin{align*}
 \dot X &= - 2^{-1}\dot H^{-1} \ddot H X, \\
 XX^* &= \dot H^{-1}, \\
  \mathcal{S} &= X^*\left(\dddot H - \frac{3}{2}\ddot H \dot H^{-1}\ddot H\right)X.
\end{align*}

We check consistency of the equations $XX^* = \dot H^{-1}$ and  $\dot X = - 2^{-1}\dot H^{-1} \ddot H X$:
\[ \dot X = - 2^{-1}\dot H^{-1} \ddot H X,  \hspace{10pt} XX^* = \dot H^{-1}, \]
\[ -2\left(XX^* - \dot H^{-1}\right)' =   \dot H^{-1} \ddot H\left(XX^* - \dot H^{-1}\right) + \left(XX^* - \dot H^{-1}\right) \ddot H\dot H^{-1} .\]
If at some initial point $u = u_0$, $X(u_0)$ is chosen, such that $X(u_0)X(u_0)^* = \dot H(u_0)^{-1}$, then the relation  $X(u)X(u)^* = (\dot H(u))^{-1}$ holds for $u \in \mathbb{U}$ any open interval on which the solution $X(u)$ is defined, such that $t \in \mathbb{U}$.   The freedom is: $ X(u) \rightarrow X(u)R,    \hspace{7pt} RR^* = I$.
The matrices $M(u)$ and $\Theta(u) = M(u)^{-1} \dot M(u)$  are now:
\[ M = \hspace{4pt}\begin{bmatrix} \dot H  -2^{-1} H\dot H^{-1}\ddot H& H\\ -2^{-1} \dot H^{-1}\ddot H&I\end{bmatrix}\hspace{4pt}X, \hspace{10pt} \dot X = - 2^{-1}\dot H^{-1} \ddot H X,  \hspace{10pt} XX^* = \dot H^{-1}, \]
\[ \Theta = M^{-1}\dot M = \hspace{4pt} \begin{bmatrix} 0&I\\ -2^{-1}\mathcal{S}&0\end{bmatrix}\hspace{4pt}, \hspace{10pt} \mathcal{S}  =  2^{-1} X^*\left(2\dddot H - 3\ddot H \dot H^{-1}\ddot H\right)X.\]
\end{proof}

Let $\mathbb X$ be a Hilbert space and consider the Hermitian symplectic space $\mathbb X\oplus\mathbb X$, with the standard positive Hermitian structure $\mathcal K=I\oplus (-I)$, $\mathcal J=\Omega=\begin{bmatrix}0&-I\\I&0\end{bmatrix}$.  A {\em Lagrangian matrix} is a smooth one parameter family of continuous linear operators $L(u):\mathbb X\to\mathbb X$, such that $L(u)\oplus L'(u):\mathbb X\to\mathbb X\oplus\mathbb X$ is an isomorphism onto a Lagrangian subspace, for all $u$.  Thus $\mathbb H(u) = (L(u)\oplus L'(u))\mathbb X$ is a curve in $\mathcal{LG}(\mathbb X\oplus\mathbb X)$.  We say that a bounded self-adjoint operator $p(u)$ of $\mathbb X$ is a {\em potential} for $L$ if $\ddot L+pL=0$.  Given a potential $p$, consider the set $\mathbb J(p)$ of solutions to $\ddot J +pJ=0$.  For each $u_0$, let $L_{u_0}(u)$ be the solution to $\ddot L_{u_0} + pL_{u_0}=0$, $L_{u_0}(u_0)=0$, $\dot L_{u_0}(u_0)=\mathcal I_{\mathbb X}$.  One then has the Hamiltonian curve $\mathbb H(u_0)$, parameterized by $u_0$, in $\mathbb J(p)$ by $\mathbb H(u_0)=L_{u_0}\mathbb X$.  

\begin{corollary}\label{PositiveCurveIsHamiltonian}
  Let $\mathbb K(u)$ be a regular positive curve in $\mathcal{LG}(\mathbb X\oplus\mathbb X)$.  Then, there exists a potential $p$, and a canonical transformation $\phi_0:\mathbb J(p) \to \mathbb X\oplus\mathbb X$, such that $\phi_0\mathbb H_p = \mathbb K$.
\end{corollary}

\begin{proof}
  Let $p=2^{-1}\mathcal S_{\mathbb K}$, and consider the space $\mathbb J(p)$.  Then we also have $p=2^{-1}\mathcal S_{\mathbb H_p}$.  Therefore the Hamiltonian curve $\mathbb H_p$ and the curve $\mathbb K$, have the same lifts to the symplectic group $\op{Sp}(\mathbb X\oplus\mathbb X)$.  
\end{proof}

Having proven this Corollary by abstract nonsense, we now proceed to develop formulas representing the Hamiltonian curve.

First, suppose the Hamiltonian curve $\mathbb H_p(t)$ be represented by the Lagrangian matrix $M(t,u)$, so
\begin{equation}\label{HamiltonianLagrangian}
  \begin{array}{c}
M_{uu}(t,u) + p(u)M(t,u) = 0,\\
M(t,t) = 0,\quad M_u(t,t) = I.
  \end{array}
\end{equation}
Suppose that $L(u)$ is a non-singular Lagrangian matrix with potential $p$.  Let $H(u)$ be symmetric such that $L\dot HL^T = I$.  Then we have
\begin{equation}\label{MLHHL}
  M(t,u) = L(u)(H(u)-H(t))L(t)^*.
\end{equation}
Indeed, we have
\begin{align*}
  M_u &= \dot L(u)(H(u)-H(t))L(t)^* + L(u)^{-*}L(t)^*\\
  M_{uu} &= -p(u)M(u,t) + (\dot L(u)L(u)^{-1} - L(u)^{-*}\dot L(u)^*)L(u)^{-*}L(t)^*\\
      &= -p(u)M(u,t)
\end{align*}
where the bracketed terms vanish because $L$ is Lagrangian.  Also, clearly $M(t,t)=0$, and
\begin{align*}
M_u(t,t) &= L(t)^{-*}L(t) = I.
\end{align*}

Now, for the corresponding curve in $\mathcal{LG}(\mathbb X\oplus\mathbb X)$, suppose that $u=0$ is in the domain of a Lagrangian matrix $M(t,u)$ representing the Hamiltonian curve $\mathbb H_[$.  Then:
\begin{theorem}
  A curve in $\mathcal{LG}(\mathbb X\oplus\mathbb X)$ corresponding to the Hamiltonian curve is given by
$$K(t) = M_u(t,0)M(t,0)^{-1},$$
which has Lagrange--Schwarzian $2p$:
$$2p(t) = 2^{-1}M(0,t)(2K'''(t) - 3K''(t)K'(t)^{-1}K''(t))M(0,t)^*.$$
\end{theorem}
\begin{proof}
  Denoting a derivative with respect to $t$ by a prime, and suppressing the $t$ dependence:
\begin{align*}
  K' &= M_u'M^{-1}-M_uM^{-1}M'M^{-1}\\
  K'' &= -M_upM^{-1} - 2M_u'M^{-1}M'M^{-1} + 2M_u(M^{-1}M')^2 + M_uM^{-1}MpM^{-1}\\
     &=-2K'M'M^{-1}\\
  K''' &= 4K'(M'M^{-1})^2 +2K'MpM^{-1}+2K'(M'M^{-1})^2\\
     &= K'(6(M'M^{-1})^2+2MpM^{-1})\\
  2K''' - 3K''(K')^{-1}K'' &= K'(12(M'M^{-1})^2+4MpM^{-1}) - 12K'(M'M^{-1})^2\\
  &=4K'MpM^{-1}.
\end{align*}
Putting back in the $t$ dependence, this shows:
\[2K'''(t) - 3K''(t)K'(t)^{-1}K''(t) = 4K'(t)M(t,0)p(t)M(t,0)^{-1}.\]
Finally, the theorem follows from the identities:
\begin{lemma}
\[M(t,0)^{-1}M(0,t)^* = -I = M(0,t)K'(t)M(t,0).\]
\end{lemma}
The first identity is immediate from \eqref{MLHHL}, since $M(u,t)^*=-M(t,u)$.

For the second, put $G(t) = M(0,t)K'(t)M(t,0)$, and note that $G$ is symmetric (again by \eqref{MLHHL}).  Put $M_u(0,u) = S(u)M(0,u)$ for $u\not=0$, where $S(u)$ is symmetric.  Then
\begin{align*}
  G'(t) &= S(t)M(0,t) K'(t)M(t,0) + M(0,t)K'(t)M(t,0)S(t) + M(0,t)K''(t)M(t,0)\\
        &= S(t)M(0,t) K'(t)M(t,0) + M(0,t)K'(t)M(t,0)S(t) \\ &\qquad - 2M(0,t)(K'(t)M(t,0)S(t))M(t,0)\\
        &= [G(t),S(t)].
\end{align*}
But $G(t)$ and $S(t)$ are both symmetric, but this shows that $G'(t)$ is skew.  Hence $G(t)$ is constant.  We now determine the constant:

\begin{align*}
  G(t) &= M(0,t)K'(t)M(t,0) \\
       &= M(0,t)\left(M_{ut}(t,0)M(t,0)^{-1}-M_u(t,0)M(t,0)^{-1}M_t(t,0)M(t,0)^{-1}\right)M(t,0)\\
       &= M(0,t)M_{ut}(t,0) - M(0,t)M_u(t,0)M(t,0)^{-1}M_t(t,0)\\
       &=- M(0,t)M_u(t,0)M(t,0)^{-1}M_t(t,0) + o(1)\\
       &=- M(0,t)IM(t,0)^{-1}(-I) + o(1)\\
       &= -I + o(1)
\end{align*}
where we have used the facts that $M(0,0)=0$, $M_t(0,0)=-I$, $M_u(0,0)=I$, and $M(t,0)^{-1}M(0,t)^*=-I=M(0,t)^*M(t,0)^{-1}$.  Thus $G(t)=-I$, as claimed.
\end{proof}

\subsection{The Lagrange--Schwarzian}
Theorem \ref{LiftToSymplectic} was stated in the simplest way, when the given positive hyperbolic structure $(\mathcal J,\mathcal K)$ is used to identify the space $\mathbb K_-$ with $\mathbb K_+$.  Thus, in the statement, the operator $\mathcal S$ is a $\mathcal J$-Hermitian operator on $\mathbb K_+$.  We want a more invariant statement, which then requires us to unpack the identification of $\mathbb K_-$ with $\mathbb K_+$.  That is, we want to think of $\mathcal S$ as a symmetric operator
$$\mathcal S : \mathbb K_+\to\mathbb K_-.$$
Unwinding these identifications leads to the following definition:
\begin{definition}
  Let $(\mathcal J,\mathcal K)$ be a positive hyperbolic structure on $\mathbb T$, and $\mathbb H$ a positive regular curve in $\mathcal{LG}(\mathbb T)$, which is complementary to $\mathbb K_\pm$, and represented by the symmetric operator $H(u):\mathbb K_+\to\mathbb K_-$.  Denote by $*_{\mathcal J}$ the Hermitian adjoint relative to the positive Hermitian structure $\mathcal J$.  Then the {\em Lagrange--Schwarzian} of $\mathbb H$ relative to $(\mathcal J,\mathcal K)$ is the operator
  $$\mathcal S_{(\mathcal J,\mathcal K,\mathbb H)}(u) = 2^{-1}X^{*_{\mathcal J}}(2\dddot H - 3\ddot H\dot H^{-1}\ddot H)X\mathcal J : \mathbb K_+\to\mathbb K_-$$
  where $X:\mathbb K_-\to\mathbb K_+$ satisfies
  $$\dot X = -2^{-1}\dot H^{-1}\ddot HX,\quad XX^{*_{\mathcal J}} = \dot H^{-1}\mathcal J.$$
\end{definition}

We claim that $\mathcal S$ transforms tensorially under conjugation by canonical transformations:
\begin{theorem}\label{GaugeInvariance}
  Let $\mathbb H$ be a positive regular curve and $(\mathcal J,\mathcal K)$ a positive hyperbolic structure such that $\mathbb H(u)$ is complementary to $\mathbb K_+$ and $\mathbb K_+$.  Let $\bar{\mathcal J}=g\mathcal Jg^{-1}$ and $\bar{\mathcal K}=g\mathcal Kg^{-1}$ be a new positive hyperbolic structure related to the original one by a canonical transformation $g$.  Then:
  $$\mathcal S_{(\bar{\mathcal J},\bar{\mathcal K},\mathbb H)} = g\mathcal S_{(\mathcal J,\mathcal K,\mathbb H)}g^{-1}.$$
\end{theorem}
\begin{proof}
  We have $H(u):\mathbb K_+\to\mathbb K_-$.  Fix a base point $u_0$.  We first show that the Theorem holds at $u_0$ if $g$ is the canonical transformation such that $\bar H(u_0)=0$.  Thus
  $$g = \begin{bmatrix}I&0\\ -H(u_0)&I\end{bmatrix}_{\mathcal K}.$$
  The defining equations for $X$ and $\mathcal S$ are preserved under $g$, and $g\mathcal S_{(\mathcal J,\mathcal K,\mathbb H)}g^{-1} = S_{(g\mathcal Jg^{-1},g\mathcal Kg^{-1},\mathbb H)}$.

  So we can suppose that $H(u_0)=0$ without loss of generality, and we have to show that
  $$\mathcal S_{(g\mathcal Jg^{-1},g\mathcal Kg^{-1},\mathbb H)}(u_0) = g\mathcal S_{(\mathcal J,\mathcal K,\mathbb H)}(u_0)g^{-1} $$
  for any canonical transformation $g$ that preserves the space $\mathbb K_+=\mathbb H(u_0)$.

  We next consider the case where $g$ is block diagonal (relative to $\mathcal K$):
  $$g = \begin{bmatrix}A & 0\\ 0 & A^{-*}\end{bmatrix}_{\mathcal K}.$$
  Then we have
  $$\bar{\mathcal K}=\begin{bmatrix}I&0\\0&-I\end{bmatrix}_{\mathcal K}=\mathcal K,\quad \mathcal J=\begin{bmatrix}0&-J^{-1}\\J&0\end{bmatrix}_{\mathcal K},\quad  \bar{\mathcal J} = \begin{bmatrix}0&-AJ^{-1}A^*\\ A^{-*}JA^{-1}\end{bmatrix}_{\mathcal K}$$
  We next claim that $\bar X^{*_{\overline{\mathcal J}}}=gX^{*_{\mathcal J}}g^{-1}$.  This follows from the fact that $g$ is a canonical transformation:
  \begin{align*}
    \omega(\bar{\mathcal J} gX^{*_\mathcal J}g^{-1}v,w)
    &=\omega(\mathcal JX^{*_{\mathcal J}}g^{-1}v,g^{-1}w)\\
    &=\omega(\mathcal Jg^{-1}v,Xg^{-1}w)\\
    &=\omega(g\mathcal Jg^{-1}v,gXg^{-1}w).
  \end{align*}

  Putting everything together, we have
  $$\mathcal S_{(g\mathcal Jg^{-1},g\mathcal Kg^{-1},\mathbb H)}(u_0) = g\mathcal S_{(\mathcal J,\mathcal K,\mathbb H)}(u_0)g^{-1}.$$

  We have therefore reduced the problem to the case in which $g$ restricts to the identity operator on $\mathbb K_+$.  For the operator $X:\mathbb K_-\to\mathbb K_+$, put $\bar X=gXg^{-1} = Xg^{-1}$.  By Theorem \ref{SchwarzianFormula}, we have {\em both} conditions
  $$\dot H(u_0)^{-1}\ddot H(u_0) = \mathcal K_+\mathbb H_\prime'(u_0)\mathcal K_+.$$
  $$\dot {\bar H}(u_0)^{-1}\ddot {\bar H}(u_0) = \bar{\mathcal K}_+\mathbb H_\prime'(u_0)\bar{\mathcal K}_+.$$
  And, because $\mathbb H_\prime'(u_0)$ is upper triangular, $\mathcal K_+\mathbb H_\prime'(u_0)=\bar{\mathcal K}_+\mathbb H_\prime'(u_0)$.  
  So
  $$\dot X(u_0) = \dot H(u_0)^{-1}\ddot H(u_0)X(u_0) \implies \dot {\bar X}(u_0) = \dot {\bar H}(u_0)^{-1}\ddot {\bar H}(u_0)g^{-1}\bar X(u_0).$$

  And, again by Theorem \ref{SchwarzianFormula}, at $u_0$, we have {\em both} conditions
  \begin{align*}
    \dot H^{-1}(2\dddot H - 3 \ddot H\dot H^{-1}\ddot H) &= \mathcal K_+\mathbb H_\prime''(u_0)^2\mathcal K_+\\
    \dot{\bar H}^{-1}(2\dddot{\bar H} - 3 \ddot {\bar H}\dot {\bar H}^{-1}\ddot{\bar H}) &= \bar{\mathcal K}_+\mathbb H_\prime''(u_0)^2\bar{\mathcal K}_+.
  \end{align*}
  So, because $\mathbb H_\prime''(u_0)^2$ is also upper triangular, we have
  $$\dot{\bar H}^{-1}(2\dddot{\bar H} - 3 \ddot {\bar H}\dot {\bar H}^{-1}\ddot{\bar H}) = g\dot H^{-1}(2\dddot H - 3 \ddot H\dot H^{-1}\ddot H)g^{-1}. $$
  Now we have shown that all of the ingredients in $\mathcal S_{(\mathcal J,\mathcal K,\mathbb H)}$ conjugate under $g$ at $u_0$, and therefore it also conjugates under $g$.  Because the base point $u_0$ was arbitrary, this completes the proof.
  
\end{proof}

Thus, Theorem \ref{LiftToSymplectic} was proven using a fixed positive hyperbolic structure $(\mathcal J,\mathcal K)$.  For example, we could take $\mathcal J=\mathcal J_{u_0}, \mathcal K=\mathcal K_{u_0}$ to be the hyperbolic structure associated to an arbitrary basepoint of the domain $U$.  For the Hamiltonian curve of a plane wave, this corresponds to taking a Lagrangian matrix $L$ such that $L(u_0)=0$, and so the Brinkmannization becomes singular at $u=u_0$.  Thus it is a feature that we have formulated Theorem \ref{LiftToSymplectic} in terms of an arbitrary positive hyperbolic structure, especially for the Brinkmannizations that remain regular throughout the domain.  We think of the pair $(\mathcal J,\mathcal K)$ as a ``gauge'', and Theorem \ref{LiftToSymplectic} gives a gauge-dependent description of the Schwarzian.  Theorem \ref{GaugeInvariance} then establishes the gauge invariance.

For a (finite dimensional) plane wave
$$G_\beta(p) = 2\,du\,dv + x^Tpx\,du^2 - dx^Tdx, $$
each Lagrangian matrix $L$ with potential $p$ gives the associated Brinkmannization
$$G_\rho(L^TL) = 2\,du\,dv - dx^T(L^TL)dx.$$
Then the Schwarz invariant relative to $L$ is
$$\mathcal S = L\left(\dddot H - \frac{3}{2}\ddot H \dot H^{-1}\ddot H\right)L^T$$
where $L\dot HL^T=I$.
\begin{lemma}
  Let $G_\rho(h)$ be a Rosen metric on $\mathbb M=\mathbb U\times\mathbb R\times\mathbb X$, and let $P$ be the (normal) tidal curvature of the central null geodesic.  Then
  $$P = -2^{-1}dx^T\dot H^{-1}\left(\dddot H - \frac32\ddot H\dot H^{-1}\ddot H\right)\dot H^{-1}\,dx$$
  where $\dot H=h^{-1}$.
\end{lemma}

\begin{proof}
  Let $L$ be a Brinkmannization of $G_\rho(h)$.  Then $L\dot HL^T=I$ and $\dot L=SL$, $\dot S+S^2+p=0$.  We then also have $(L^{-1})' = -L^{-1}\dot LL^{-1}=-L^{-1}S$.
  \begin{align*}
    \ddot H &= (L^{-1}L^{-T})' = -2L^{-1}SL^{-T}\\
    \dddot H &= 4 L^{-1}S^2L^{-T} - 2L^{-1}\dot SL^{-T} \\
            &= 4L^{-1}S^2L^{-T}+2L^{-1}(S^2+p)L^{-T}\\
            &= 6L^{-1}S^2L^{-T} + 2L^{-1}pL^{-T}
  \end{align*}
  \[\dddot H - \frac32\ddot H\dot H^{-1}\ddot H = 2L^{-1}pL^{-T}.\]
  The Lemma now follows because the tidal curvature of the Rosen metric is
  $$P = -dx^TL^TpL\,dx.$$
\end{proof}
Thus in the plane wave formalism, gauge invariance of the Schwarzian is simply the tensor transformation law that the normal tidal curvature undergoes when it is subjected to an isomorphism of Rosen metrics \cite{SEPII}.


\bibliographystyle{hplain} 
\bibliography{planewaves} 

\end{document}